\RequirePackage{fix-cm}
\documentclass[nospthms]{svjour3}
\usepackage[numbers,sort&compress]{natbib}
\usepackage[usenames,dvipsnames]{xcolor}
\usepackage{tikz}
\usetikzlibrary{calc,decorations.pathreplacing,decorations.pathmorphing}
\usepackage[english]{babel}
\usepackage[utf8]{inputenc}
\usepackage{amsmath}
\usepackage{mathrsfs}
\usepackage{amsthm}
\usepackage{amssymb}
\usepackage{thmtools}
\IfFileExists{bbold.sty}{\usepackage{bbold}}{}
\newif\ifphysicsloaded
\IfFileExists{physics.sty}{\usepackage{physics}\physicsloadedtrue}{\physicsloadedfalse}
\ifphysicsloaded\else
\makeatletter
\DeclareRobustCommand{\qty}{%
  \@ifnextchar({\qty@paren}{%
  \@ifnextchar[{\qty@bracket}{%
  \@ifnextchar\bgroup{\qty@brace}{}%
  }}%
}
\def\qty@paren(#1){\left(#1\right)}
\def\qty@bracket[#1]{\left[#1\right]}
\def\qty@brace#1{\left\{#1\right\}}
\makeatother

\newcommand{\ket}[1]{| #1 \rangle}

\newcommand{\dd}[1]{\mathrm{d}#1}

\newcommand{\abs}[1]{\left| #1 \right|}
\newcommand{\eval}[1]{\left.#1\right|}
\newcommand{\pqty}[1]{\left(#1\right)}
\fi
\usepackage{mathtools}
\usepackage{enumitem}

\newcommand{\vertexsize}{7pt}

\makeatletter
\renewcommand\section{\@startsection{section}{1}{\z@}%
    {-21dd plus-8pt minus-4pt}{10.5dd}%
    {\large\bfseries\boldmath}}
\renewcommand\subsection{\@startsection{subsection}{2}{\z@}%
    {-21dd plus-8pt minus-4pt}{10.5dd}%
    {\normalsize\itshape}}
\makeatother

\newcommand{\hc}{\mathtt{H}}
\newcommand{\pa}{\mathcal{L}}
\newcommand{\pathset}[2]{\mathscr{P}_{#2}({#1})}
\newcommand{\pathsetG}{\pathset{\extendedgraph{G}{\chi}}{\chi}}
\newcommand{\pathsetodd}[2]{\mathscr{P}^{\mathrm{odd}}_{#2}({#1})}
\newcommand{\pathseteven}[2]{\mathscr{P}^{\mathrm{even}}_{#2}({#1})}
\newcommand{\pathsetlength}[2]{\mathscr{P}^{#2}({#1})}
\newcommand{\oddpathsetlengthcomp}[3]{\mathscr{P}^{#2}_{#3,\mathrm{odd}}({#1})}

\newcommand{\pathprod}[1]{h\qty[#1]}
\newcommand{\pathprodH}[1]{\hc\qty[#1]}
\newcommand{\indepset}[1]{\mathcal{S}_{#1}}
\newcommand{\indepcoeff}[2]{I_{#1}^{(#2)}}
\newcommand{\localcharge}[2]{\mathcal{H}_{#1}^{(#2)}}
\newcommand{\nonlocalcharge}[2]{Q_{#1}^{(#2)}}
\newcommand{\packingcharge}[3]{\mathcal{Q}_{#1}^{(#2,#3)}}
\newcommand{\catalancharge}[1]{\mathfrak{H}_{#1}^{\mathrm{Cat}}}
\newcommand{\catalanchargeXXX}[1]{\mathfrak{H}_{#1}^{\mathrm{XXX}}}

\newcommand{\edgeop}[1]{\chi_{\lower0.25ex\hbox{$\scriptscriptstyle #1$}}}

\newcommand{\floor}[1]{\left\lfloor#1\right\rfloor}
\newcommand{\ceil}[1]{\left\lceil#1\right\rceil}

\newcommand{\refst}{\ket{\emptyset}}

\newcommand{\PP}{\mathcal{P}}

\newcommand{\AAA}{\mathcal{A}}

\newcommand{\NN}{\mathcal{N}}

\newtheorem{thm}{Theorem}[section]
\newtheorem{lem}[thm]{Lemma}
\newtheorem{cor}[thm]{Corollary}
\newtheorem{define}[thm]{Definition}
\newtheorem{rmk}[thm]{Remark}

\newcommand{\complementset}[1]{\overline{#1}}
\newcommand{\extendedgraph}[2]{#1_{#2}}
\newcommand{\pathclique}{{\mathcal K}}
\newcommand{\pathsvd}{{\mathcal S}}
\newcommand{\pathloop}{{\mathcal C}}
\newcommand{\indepnum}[1]{\alpha_{#1}}

\newcommand{\jbm}{\boldsymbol{j}}
\newcommand{\ibm}{\boldsymbol{i}}
\newcommand{\ellbm}{\boldsymbol{\ell}}

\newcommand{\md}{\text{-}}
\newcommand{\imi}{\mathrm{i}}

\DeclareMathOperator{\sgn}{sgn}

\newtheorem*{theorem*}{Theorem}

\usepackage{iftex}
\usepackage{ifpdf}

\ifPDFTeX
	\ifpdf
		\usepackage{epstopdf}
		\usepackage[pdftex,colorlinks,urlcolor=NavyBlue,citecolor=NavyBlue,linkcolor=NavyBlue]{hyperref}
	\else
		\usepackage[hypertex,colorlinks,urlcolor=NavyBlue,citecolor=NavyBlue,linkcolor=NavyBlue]{hyperref}
	\fi
\else
	\usepackage[colorlinks,urlcolor=NavyBlue,citecolor=NavyBlue,linkcolor=NavyBlue]{hyperref}
\fi

\usepackage{doi}
\usepackage{url}
\ifPDFTeX
	\ifdefined\pdfadjustspacing
	\fi
\fi
\journalname{Communications in Mathematical Physics}

\begin{document}

\title{Solving models with generalized free fermions II: Path-product expansion and conserved charges}
\titlerunning{Path-product expansion and conserved charges}

\author{Kohei Fukai \and Bal\'azs Pozsgay \and Istv\'an Vona}
\authorrunning{K. Fukai et al.}

\institute{
	K.~Fukai \at RIKEN iTHEMS, Wako, Saitama 351-0198, Japan\\
	\email{kohei.fukai@riken.jp}
	\and
	B.~Pozsgay, I.~Vona \at MTA-ELTE ``Momentum'' Integrable Quantum Dynamics Research Group,
	ELTE E\"otv\"os Lor\'and University, Budapest, Hungary
	\and
	I.~Vona \at Holographic Quantum Field Theory Research Group,
	HUN-REN Wigner Research Centre for Physics, Budapest, Hungary
}

\date{Received: date / Accepted: date}

\maketitle

\begin{abstract}
	Free-fermion solvability in quantum spin systems is increasingly understood to be governed by a graph Clifford algebra defined from the frustration graph of the Hamiltonian.
	When the frustration graph belongs to certain classes, such as the even-hole-free and claw-free (ECF) class, the Hamiltonian is solvable by \emph{hidden} free fermions: it admits a free-fermion solution although it does not reduce to a Majorana bilinear under the Jordan-Wigner transformation.
	However, unlike in the Jordan-Wigner case, where each mode is a linear combination of single Majorana fermions, the explicit operator structure of the hidden free-fermion modes---and that of the local conserved charges---has remained obscure.
	In this work, we derive a path-product expansion that expresses each free-fermion mode as a linear combination of products along induced paths in the extended frustration graph.
	The expansion follows from the Krylov generating function and yields the modes directly, without using the transfer matrix or nonlocal conserved charges; the resulting mode decomposition also computes infinite-temperature dynamical correlation functions.
	We further obtain explicit expressions for local conserved charges as linear combinations of path products along induced paths; these charges apply beyond the free-fermion (ECF) class to more general claw-free frustration graphs.
	We also identify a family of generalized conserved charges containing both the known nonlocal conserved charges and these local charges.
	For the homogeneous periodic Fendley model, the local conserved charges exhibit the same Catalan-tree pattern as those of the spin-$1/2$ XXX chain.
\end{abstract}

\section{Introduction}
Free-fermion systems are among the simplest exactly solvable systems, playing a fundamental role in theoretical physics.
The Jordan-Wigner transformation~\cite{jordan-wigner} maps certain spin chains to free-fermion systems, providing the standard route to their exact solution.
Notable examples include the XY chain~\cite{XX-original} and the transverse-field Ising chain~\cite{pfeuty-transverse-ising}; the latter is equivalent to the two-dimensional classical Ising model first solved by Onsager~\cite{Onsager-Ising,Schulz-Mattis-Lieb}.
There are free-fermion models in two dimensions, for example Kitaev's honeycomb model~\cite{kitaev-honeycomb}.
In these models, the Hamiltonian is mapped to a Majorana fermion bilinear, and the fermionic eigenmodes are linear combinations of single Majorana fermions~\cite{japan-free-fermion-JW,chapman-jw}.
Generalizations of the Jordan-Wigner transformation have also been explored in a variety of settings~\cite{japan-JW-gen-1,japan-JW-gen-2,chapman-jw,japan-JW-gen-3,japan-free-fermion-JW,minami-conserved-charges}.

Recently, Fendley introduced a free-fermion spin chain termed ``free fermions in disguise'' (FFD), hereafter the Fendley model~\cite{fendley-fermions-in-disguise}.
This model is exactly solvable but lies beyond the Jordan-Wigner paradigm~\cite{fermions-behind-the-disguise,unified-graph-th,sajat-FP-model,sajat-ffd-corr,sajat-floquet,sajat-claws,eric-lorenzo-ffd-1,eric-lorenzo-circuits}: the Hamiltonian maps to a four-Majorana interaction under the Jordan-Wigner transformation, so the free-fermion solvability is hidden and the underlying algebraic structure is much less transparent than in the standard case.
Various multispin and parafermionic extensions have also been developed in parallel~\cite{alcaraz-medium-fermion-1,alcaraz-medium-fermion-2,alcaraz-medium-fermion-3,rodrigo-ising-and-ffd,free-parafermionic-graphs,alcaraz-nonhomogeneous-multispin}.
Sufficient conditions for hidden free-fermion solvability have been formulated in terms of the frustration graph, which encodes the anticommutation structure of the Hamiltonian terms~\cite{fermions-behind-the-disguise,unified-graph-th,free-parafermionic-graphs,ruh-elman-twin-collapse}.
One such condition is that the frustration graph be even-hole-free and claw-free~\cite{fermions-behind-the-disguise}.
A more general condition is that the frustration graph be simplicial and claw-free~\cite{unified-graph-th}.

Despite these developments, several aspects of these free-fermion solutions remain open.
The connection between the generalized hidden free fermions and the traditional Jordan-Wigner solvable models is unclear; the two cases appear to require entirely separate methods.
Although Ref.~\cite{unified-graph-th} provides a unified conceptual framework, the concrete construction of the free-fermion modes still proceeds indirectly through an auxiliary edge operator, the transfer matrix, and previously known nonlocal conserved charges.
This stands in stark contrast to the Jordan-Wigner case, where one simply diagonalizes a single-particle Hamiltonian matrix.
The operator content and coefficients of the hidden free-fermion modes, however, have not been made explicit.

A second unresolved issue concerns the structure of conserved charges in these models.
For claw-free frustration graphs, the independent sets of the frustration graph yield a mutually commuting family of \emph{nonlocal} conserved charges~\cite{fermions-behind-the-disguise,unified-graph-th}.
However, the structure of \emph{local} conserved charges has remained largely unresolved.
This question is significant because the appearance of higher local conserved charges beyond the Hamiltonian is a standard hallmark of integrability in spin chains~\cite{integrability-test,GM-higher-conserved-XXZ,hokkyo-integrability-test,dichotomy,resh-cond-proof,yamaguchi-shiraishi-nonhermitian-bosons}, and the presence or absence of local conserved charges has become a useful rigorous diagnostic for integrability and nonintegrability in quantum many-body systems~\cite{xyz-not-int,chiba-mixed-field-ising,shiraishi-nnn-heisenberg,pxp-nonintegrable,futami-tasaki-compass,fan-korepin-fredkin-local-charges,chiba-ising-higher-dim,yamaguchi-chiba-shiraishi-nn-absence,futami-hubbard-higher-dim,shiraishi-tasaki-higher-d-xy-xyz,fan-korepin-minimal-nonintegrable-three-site,ishii-yamaguchi-holstein}.
Even in the free-fermion-solvable ECF case, to our knowledge no closed-form operator expression for local conserved charges was available: for the Fendley model, the logarithmic derivative of the transfer matrix was known to generate them~\cite{fendley-fermions-in-disguise}, but the resulting charges had not been written explicitly.
For general claw-free frustration graphs beyond the free-fermion setting, the situation is more delicate still, and even the existence of local conserved charges has remained entirely open.

This paper is the second part of a series on graph Clifford algebras and hidden free fermions.
In the first part~\cite{sajat-solving-ffd-1}, we studied the algebras behind these models and the defining representations of these algebras.
Furthermore, we constructed few-body eigenstates for specially chosen antisymmetric Hamiltonians.

In this article, we answer the above questions by deriving an explicit path-product expansion for the free-fermion modes in generalized free-fermion spin chains, focusing on the even-hole-free and claw-free case.
We show that each free-fermion mode is a linear combination of path products---ordered products of Hamiltonian terms along induced paths in the frustration graph extended by an auxiliary edge operator---with coefficients given by independence polynomials of suitable residual graphs.
This result follows from the Krylov generating function in operator space, whose path-product expansion yields the free-fermion modes as residues.
The construction gives the modes directly, without using the transfer matrix or the previously known nonlocal conserved charges as input.
The path-product expansion also produces graph-theoretic identities for independence polynomials, revealing a concrete link between hidden free fermions and induced-path combinatorics.

For the inverse problem, we rewrite the edge operator in terms of the free-fermion modes and compute its infinite-temperature autocorrelation.
More generally, this mode-decomposition method computes infinite-temperature dynamical correlation functions of arbitrary Krylov basis elements, extending the edge-operator autocorrelation analysis of the Fendley model~\cite{sajat-ffd-corr} to arbitrary ECF frustration graphs.

As further consequences of the path-product expansion, we obtain an asymptotic estimate for the dominant induced paths contributing to a fermionic mode.
As an application, we study the path-product expansion in the open Fendley model.
We show that as the paths extend from the boundary into the bulk of the chain, their number grows but the coefficients decrease.
The two effects cancel at leading exponential order, and we show that a precisely defined operator weight neither grows nor decays exponentially with distance from the boundary.

Our second main result concerns the structure of \emph{local} conserved charges.
Here and below, ``local'' is used in the graph-Clifford sense introduced in Section~\ref{sec:setup}: the range is measured by induced-path size in the frustration graph, rather than by a background physical geometry unless a concrete spin-chain realization is specified.
We find that local conserved charges are linear combinations of path products along induced paths with an odd number of vertices, with coefficients given by independence polynomials of the associated residual graphs.
The same construction extends beyond the free-fermion (ECF) setting to arbitrary claw-free frustration graphs.
The open Fendley model has frustration graph $P_M^2$, which is ECF and hence satisfies the sufficient conditions for a hidden free-fermion solution.
By contrast, the periodic graph $C_M^2$ is claw-free but not simplicial, and hence is not covered by these sufficient conditions.
Nevertheless, our conserved-charge construction applies to $C_M^2$, the frustration graph of the periodic Fendley model~\cite{fendley-fermions-in-disguise}, for which, to our knowledge, no exact solution is currently known.
The formula yields an extensive family of explicit local conserved charges for arbitrary couplings.
For the homogeneous chain, suitable triangular recombinations of these charges exhibit the same Catalan-tree pattern that organizes the local conserved charges of the spin-$1/2$ XXX chain and its $SU(N)$-invariant generalizations~\cite{anshelevich-heisenberg-first-integrals,GM-higher-conserved-XXX,GM-higher-conserved-XXZ,GM-catalan-tree}.
More generally, for every claw-free graph these local conserved charges are exactly conserved up to the range set by the smallest \emph{even bubble wand}---a combinatorial obstruction we identify---and an extensive family is obtained whenever this smallest wand grows with the system size or is absent altogether.
Our construction thus points to integrability in such claw-free families by exhibiting an explicit hierarchy of higher local conserved charges beyond the Hamiltonian.

Finally, we introduce a unified family of generalized conserved charges that interpolates between the nonlocal independent-set charges and the local induced-path charges.
The family is built from collections of mutually non-adjacent odd induced paths.
The formula reduces to the nonlocal independent-set charges generated by the transfer matrix~\cite{fendley-fermions-in-disguise,fermions-behind-the-disguise,unified-graph-th} when every path in the collection is a singleton, and to the local conserved charges constructed above when the collection contains a single path.
The same range bound follows by extending the local conserved-charge cancellation mechanism to disconnected families of odd induced paths.

The paper is organized as follows.
Section~\ref{sec:setup} introduces the frustration-graph framework and path products.
Section~\ref{sec:proof} derives the path-product expansion of the Krylov generating function for ECF frustration graphs, uses it to prove the free-fermion relations without invoking the transfer matrix, and computes the infinite-temperature autocorrelation of the edge operator.
Section~\ref{sec:transfer-matrix-formalism} gives an alternative transfer-matrix expression for the Krylov generating function and relates the path-product modes to the transfer-matrix construction.
Section~\ref{sec:conserved-charges} collects the conserved-charge consequences: local conserved charges, generalized conserved charges, and the oriented even path operators.
Section~\ref{sec:HA} proves the companion-paper theorem announced in Ref.~\cite{sajat-solving-ffd-1} for the few-body eigenstate construction of the antisymmetric Hamiltonian.
Section~\ref{sec:ising-example} applies the path-product construction to the transverse-field Ising chain and verifies that it reproduces the standard Jordan-Wigner solution.
Section~\ref{sec:ffd-example} applies the path-product expansion to the Fendley model, analyzes the asymptotic distribution of path weights in the modes, and works out explicit conserved charges for the open and homogeneous periodic cases.
Finally, Section~\ref{sec:outlook} summarizes the results and open directions.

\section{Frustration graph formalism}
\label{sec:setup}

This section defines the frustration graph, its graph Clifford algebra, the extended graph, and the path products used to construct the free-fermion modes and conserved charges.

\subsection{Frustration graph and graph Clifford algebra}
\label{subsec:frustration-graph}

Let $G = (V, E)$ be a finite simple graph, where $V$ is the set of vertices and $E$ is the set of edges.
The adjacency matrix $A$ of $G$ is defined by $A_{\jbm \jbm'} = 1$ if $\{\jbm, \jbm'\} \in E$ and $A_{\jbm \jbm'} = 0$ otherwise.
We call $G$ the \emph{frustration graph}.

To each vertex $\jbm \in V$, we associate a generator $h_{\jbm}$ of an algebra.
The generators satisfy
\begin{equation}
	\label{eq:algebra-relations}
	h_{\jbm}^2 = 1, \qquad h_{\jbm} h_{\jbm'} = (-1)^{A_{\jbm \jbm'}} h_{\jbm'} h_{\jbm}.
\end{equation}
Such algebras are known as \emph{graph Clifford algebras}~\cite{graph-clifford}.

We also introduce coupling-dependent generators:
\begin{align}
	\label{eq:coupling-dependent-generators}
	\hc_{\jbm} \equiv b_{\jbm} h_{\jbm},
\end{align}
where $b_{\jbm} \in \mathbb{R} \setminus \{0\}$ are arbitrary coupling constants.

Fix a total order on $V$.
For a subset $S=\{\jbm_1<\cdots<\jbm_r\}\subseteq V$, define
\begin{align}
	h_S
	\equiv
	h_{\jbm_1}\cdots h_{\jbm_r},
	\qquad
	h_{\emptyset}\equiv 1.
\end{align}
These ordered monomials form a vector-space basis of the abstract graph Clifford algebra.
Equivalently, one may use the coupling-dependent ordered monomials $\hc_{\jbm_1}\cdots\hc_{\jbm_r}$.

Given a frustration graph $G$ and coupling constants $\{b_{\jbm}\}_{\jbm \in V}$, we define the Hamiltonian:
\begin{equation}
	\label{eq:Hamiltonian}
	H_G = \sum_{\jbm \in V} \hc_{\jbm}.
\end{equation}
This Hamiltonian is solvable by hidden free fermions when $G$ is even-hole-free and claw-free~\cite{fermions-behind-the-disguise}, and more generally when it is claw-free and contains a simplicial clique~\cite{unified-graph-th}, as reviewed below.

\subsection{Claw-free graphs and simplicial cliques}
\label{subsec:scf}

\begin{figure}[tb]
	\centering
	\begin{tikzpicture}
		\node[circle, draw, fill=black, minimum size=\vertexsize, inner sep=0pt, label=below:{center}] (center) at (0,0) {};
		\node[circle, draw, fill=black, minimum size=\vertexsize, inner sep=0pt] (a) at (-1.2,1) {};
		\node[circle, draw, fill=black, minimum size=\vertexsize, inner sep=0pt] (b) at (0,1.2) {};
		\node[circle, draw, fill=black, minimum size=\vertexsize, inner sep=0pt] (c) at (1.2,1) {};
		\draw (center) -- (a);
		\draw (center) -- (b);
		\draw (center) -- (c);
		\node at (0,1.6) {leaves};
	\end{tikzpicture}
	\caption{The claw graph $K_{1,3}$ consists of one center vertex adjacent to three pairwise non-adjacent leaves.}
	\label{fig:claw}
\end{figure}

A graph is \emph{claw-free} if it contains no claw $K_{1,3}$ (Fig.~\ref{fig:claw}) as an induced subgraph.
A \emph{clique} is a subset of vertices $K \subseteq V$ such that every pair of vertices in $K$ is adjacent.

The construction of free-fermion modes requires selecting a special clique $K_s$ in the frustration graph and introducing an auxiliary vertex connected to it.
Given a clique $K_s \subseteq V$, consider adding a new vertex $v_\chi$ that is adjacent to all vertices in $K_s$ and to no other vertices.
The resulting \emph{extended graph} $\extendedgraph{G}{\chi}$ has vertex set $V \cup \{v_\chi\}$ and edge set $E \cup \{(v_\chi, \jbm) \mid \jbm \in K_s\}$.
When $G$ is claw-free and $\extendedgraph{G}{\chi}$ also remains claw-free, the clique $K_s$ is called \emph{simplicial}.
A claw-free graph that contains a simplicial clique is called \emph{simplicial claw-free} (SCF).

Equivalently, a simplicial clique can be characterized intrinsically, without reference to the extended graph, as follows.
For a vertex $\jbm \in V$, we denote by $\Gamma(\jbm)$ its \emph{open neighborhood}, i.e., the set of vertices adjacent to $\jbm$:
\begin{equation}
	\Gamma(\jbm) = \{ \jbm' \in V \mid \{\jbm, \jbm'\} \in E \}.
\end{equation}
The \emph{closed neighborhood} $\Gamma[\jbm]$ includes the vertex itself: $\Gamma[\jbm] = \Gamma(\jbm) \cup \{\jbm\}$.
A clique $K_s$ is simplicial if and only if, for every vertex $\jbm \in K_s$, the open neighborhood of $\jbm$ outside $K_s$ forms a clique; equivalently, since $\jbm\in K_s$, one may write this condition with the closed neighborhood:
\begin{equation}
	\Gamma[\jbm] \setminus K_s \quad \text{is a clique for all } \jbm \in K_s.
\end{equation}

A key fact is that every even-hole-free and claw-free (ECF) graph is also SCF~\cite{chudnovsky-seymour,evenhole-simplicial}: every connected ECF graph contains a simplicial clique, so the auxiliary edge vertex can be attached without creating a claw.
The ECF condition was introduced in Ref.~\cite{fermions-behind-the-disguise} as a sufficient condition for the existence of free-fermion modes, and was subsequently generalized to the SCF condition in Ref.~\cite{unified-graph-th}.
The present path-product proof of the mode formula uses the ECF condition, because the even-hole-free condition is used in the loop cancellation in Theorem~\ref{thm:Krylov-path-expansion}.
The conserved-charge construction in Section~\ref{sec:conserved-charges}, by contrast, is formulated for arbitrary claw-free graphs and its range is controlled by the even bubble wand obstruction introduced there.

\subsection{Edge operator}
\label{subsec:edge-operator}

The \emph{edge operator} $\chi$ is the additional generator corresponding to the auxiliary vertex $v_\chi$ in the graph Clifford algebra of the extended graph $\extendedgraph{G}{\chi}$.
It satisfies $\chi^2 = 1$ together with the commutation relations
\begin{equation}
	\label{eq:edge-operator}
	\{\hc_{\jbm}, \chi\} = 0 \text{ for } \jbm \in K_s, \quad [\hc_{\jbm}, \chi] = 0 \text{ for } \jbm \notin K_s.
\end{equation}
We use $\chi$ to denote both the edge operator and its corresponding vertex $v_\chi$ in $\extendedgraph{G}{\chi}$.
For example, if $G=P_M$ is a path graph, the endpoint clique $K_s=\{1\}$ is simplicial, and $\chi$ is represented by a new vertex attached only to vertex $1$ in $\extendedgraph{G}{\chi}$.

\subsection{Independent sets and independence polynomial}
\label{subsec:indep}

An \emph{independent set} $S \subseteq V$ is a subset of vertices such that no two vertices in $S$ are adjacent in $G$.
Equivalently, the generators $\{\hc_{\jbm}\}_{\jbm \in S}$ mutually commute.
We denote by $\indepset{G}$ the set of all independent sets in $G$, and by $\alpha_G$ the \emph{independence number}, i.e., the maximum size of an independent set.

The \emph{independence polynomial} of $G$ is defined as
\begin{equation}
	\label{eq:independence-polynomial}
	P_G(x) \equiv \sum_{S \in \indepset{G}} (-x)^{|S|} \prod_{\jbm \in S} b_{\jbm}^2.
\end{equation}
We denote its $k^{\text{th}}$ coefficient by
\begin{align}
	\label{eq:indep-coeff}
	\indepcoeff{G}{k}
	\equiv
	[x^{k}]\,P_G(x)
	=
	\sum_{S \in \indepset{G}^{k}} \prod_{\jbm \in S} (-b_{\jbm}^2),
\end{align}
where $\indepset{G}^{k}$ denotes the set of independent sets of size $k$ in $G$, and $\indepcoeff{G}{0}=1$.
Here and throughout, $[x^{k}]\,f(x)$ denotes the coefficient of $x^{k}$ in the formal power series expansion of $f(x)$ in the variable $x$.
For a polynomial such as $P_G(x)$, this is the ordinary polynomial coefficient; for a rational function regular at the origin, it is the coefficient in the Taylor expansion at $x=0$.
For claw-free graphs, the weighted independence polynomial has only real roots~\cite{chudnovsky-seymour,engstrom-weighted-clawfree,bencs1,leake-ryder-multivariate-independence}\footnote{The cited results are commonly stated for the unsigned multivariate independence polynomial.
	Specializing the vertex variables to the positive weights $b_{\jbm}^2 x$ gives real negative roots in the unsigned convention; our convention replaces $x$ by $-x$, so the roots of $P_G(x)$ are real and positive.}.
Consequently, the number of positive roots of $P_G(u^2)=0$, counted with multiplicity, equals the independence number $\alpha_G$.

There is a recursion relation for the independence polynomial with respect to a clique $K \subset V$:
\begin{equation}
	\label{eq:independence-polynomial-recursion}
	P_G(x) = P_{G \setminus K}(x) - x \sum_{\jbm \in K} b_{\jbm}^2 P_{G \setminus \Gamma[\jbm]}(x).
\end{equation}
Equation~\eqref{eq:independence-polynomial-recursion} suffices for the alternative proof of hidden free-fermion solvability below, without the transfer matrix or nonlocal conserved charges used in Refs.~\cite{fendley-fermions-in-disguise,fermions-behind-the-disguise,unified-graph-th}.

\subsection{Induced paths}
\label{subsec:induced-paths}

An \emph{induced path} in a graph is a sequence of distinct vertices $\pa = \ell_1\md\ell_2\md\cdots\md\ell_m$ such that, for any $i \ne j$, $(\ell_i, \ell_j) \in E$ if and only if $|i - j| = 1$.
We define the \emph{size} of the path as $|\pa| = m$, the number of vertices it contains.
Thus, throughout the paper, locality measured by $|\pa|$ refers to this vertex size, not to the number of edges.
Unless a root or orientation is explicitly specified, a path is unrooted and identified with its reverse ordering.
We write
\begin{equation}
	\pathset{G}{}
	\equiv
	\bigl\{ \pa \mid \pa \text{ is an induced path in } G \bigr\}.
\end{equation}
Every path $\pa \in \pathset{G}{}$ has $|\pa| \geq 1$, with the minimal path being a single vertex.

For a path $\pa = \ell_1\md\ell_2\md\cdots\md\ell_m$, we write $\pa_{[a;b]} \equiv \ell_a\md\ell_{a+1}\md\cdots\md\ell_b$ ($1\le a\le b\le m$) for the subpath.
We also use the boundary shortcuts $\pa_{[a;]} = \pa_{[a;m]}$, $\pa_{[;b]} = \pa_{[1;b]}$, $\pa_{[;-1]} = \pa_{[1;m-1]}$, and $\pa_{[a;-1]} = \pa_{[a;m-1]}$.

For $v \in V(G)$, the subscripted variant $\pathset{G}{v}$ denotes paths rooted at $v$, represented by the ordering whose first vertex is $v$.
For $l \geq 1$, the superscripted variant $\pathsetlength{G}{l}$ denotes the subset of paths in $\pathset{G}{}$ with $|\pa|=l$.
For parity restrictions, we write $\pathsetodd{G}{}$ and $\pathseteven{G}{}$ for the subsets of paths with odd and even size, respectively; the same parity superscripts may be combined with a starting vertex, as in $\pathsetodd{G}{v}$.
Unrooted path sums, including $\pathsetlength{G}{l}$, are therefore sums over reversal classes.
Even path operators require an additional orientation choice, which is introduced in Subsection~\ref{subsec:even-path-operators}.
The starting-vertex notation is used in Section~\ref{sec:proof}, while the fixed-size and parity variants are used in Section~\ref{sec:conserved-charges}.

\subsection{Path products}
\label{subsec:path-products}

For each path $\pa = \ell_1\md\ell_2\md\cdots\md\ell_m \in \pathset{G}{}$, we define the \emph{path product}:
\begin{align}
	\label{eq:path-product}
	\pathprodH{\pa} & \equiv \hc_{\ell_1} \hc_{\ell_2} \cdots \hc_{\ell_m}\,.
\end{align}
The path product satisfies $\pathprodH{\pa} = (-1)^{\abs{\pa}-1} \pathprodH{\pa^{-1}}$, where $\pa^{-1} = \ell_m\md \ell_{m-1}\md\cdots\md\ell_{1}$ is the reverse ordering of the path $\pa$.
For odd $|\pa|$, the path product is therefore well-defined on reversal classes.
For even paths, reversing the path changes the sign of the path product, so an operator built from even paths must choose an orientation.

\paragraph{Locality convention.}\label{par:locality-convention}
Unless a concrete spin-chain realization is specified, the graph-Clifford locality of a path product $\pathprodH{\pa}$ is measured by its vertex size $|\pa|$, and a local conserved charge is a finite linear combination of path products of uniformly bounded size.
This graph-local notion is intrinsic to the frustration graph and reduces to ordinary spatial locality in the one-dimensional examples discussed later.

\subsection{Residual graph}
\label{subsec:residual-graph}
We define the \emph{residual graph} $\complementset{\pa}$ for a path $\pa \in \pathset{G}{}$ as the induced subgraph of $G$ obtained by removing the closed neighborhood of $\pa$:
\begin{equation}
	\label{eq:residual-graph}
	\complementset{\pa} \equiv G \setminus \Gamma[\pa].
\end{equation}
Here $\Gamma[\pa]=\bigcup_{\ell\in\pa}\Gamma[\ell]$ is computed in the graph in which the path lives.
For paths in the extended graph $\extendedgraph{G}{\chi}$, the resulting residual graph is always regarded as an induced subgraph of the original graph $G$; in particular, $\Gamma[\chi]\cap V(G)=K_s$ and $\complementset{(\chi)}=G\setminus K_s$.
The notation $\complementset{\pa}$ is used throughout only for this residual graph, not for the path itself.

\section{Free-fermion modes from the path-product expansion}
\label{sec:proof}

We construct the generating function for the Krylov basis, expand it as a sum over induced paths, extract the modes as residues at its poles, prove the free-fermion relations, and record the mode decomposition together with the finite linear relations in the Krylov basis.

\subsection{Path-product expansion of free-fermion modes}
\label{sec:main-result}
\label{subsec:main-result}

Let $G$ be a connected ECF frustration graph, and let $\chi$ be the edge operator associated with a simplicial clique $K_s$.
Let $u_k>0$, $k=1,\ldots,\indepnum{G}$, denote the positive roots of $P_G(u^2)=0$, and set $u_{-k}=-u_k$.
Whenever normalized modes and residues are used, we assume generic nonzero couplings, so that these positive roots are simple and $P'_G(u_k^2)\neq0$, $P_{G\setminus K_s}(u_k^2)\neq0$.

In the free-fermion-mode expansion, paths in $\extendedgraph{G}{\chi}$ are rooted at the auxiliary vertex $\chi$ and written $\pa=\chi\md\ell_1\md\cdots\md\ell_n \in \pathset{G_\chi}{\chi}$, so that $|\pa|=n+1$.
In formulas that use endpoint indices we set $\ell_0=\chi$, so that $\pa=\ell_0\md\ell_1\md\cdots\md\ell_n$ and $\pa_{[;-1]}=\chi\md\ell_1\md\cdots\md\ell_{n-1}$.
For such a rooted path, $\pathprodH{\pa}=\chi \hc_{\ell_1}\hc_{\ell_2}\cdots\hc_{\ell_n}$, and $\pathprodH{\chi}\equiv\chi$ for the minimal rooted path $(\chi)$.

\begin{define}
	\label{define:path-product-operator}
	We define a path-product operator $\Psi_k$ by
	\begin{equation}
		\label{eq:free-fermion-mode-path-expansion}
		\Psi_k \equiv \frac{1}{\NN_k} \sum_{\pa \in \pathset{G_\chi}{\chi}} (-u_k)^{|\pa|-1} P_{\complementset{\pa}}(u_k^2) \, \pathprodH{\pa},
	\end{equation}
	where $\NN_k$ is the normalization factor defined by
	\begin{align}\label{eq:normalization-factor-Psi}
		\NN_k
		 & =
		2 \sqrt{- u_k^2 P_{G \setminus K_s}(u_k^2) P'_{G}(u_k^2)}.
	\end{align}
	Here $P'_G(x) \equiv \frac{\partial P_G(x)}{\partial x}$.
\end{define}
The radicand in Eq.~\eqref{eq:normalization-factor-Psi} is positive under these assumptions, as shown in Lemma~\ref{lem:normalization-positivity}.
With this convention, $\NN_{-k}=\NN_k$.
Compared with the convention used in \cite{sajat-solving-ffd-1}, the present normalization differs by a factor of $2$.

\begin{thm}
	\label{thm:path-product-expansion}
	Let $G$ be a connected ECF frustration graph, let $K_s$ be a simplicial clique, and let $\chi$ be the associated edge operator.
	For each $k=\pm1,\ldots,\pm\indepnum{G}$, the operator $\Psi_k$ in Eq.~\eqref{eq:free-fermion-mode-path-expansion} is a free-fermion mode of $H_G$ with single-particle energy $\epsilon_k = 1/u_k$.
	It satisfies the eigenoperator equation
	\begin{equation}
		\label{eq:eigenoperator}
		[H_G, \Psi_k] = 2\epsilon_k \Psi_k,
	\end{equation}
	and the fermionic anticommutation relation
	\begin{equation}
		\label{eq:anticommutation}
		\{\Psi_k, \Psi_l\} = \delta_{k+l, 0}\,.
	\end{equation}
\end{thm}

The free-fermion modes $\Psi_k$ defined in Eq.~\eqref{eq:free-fermion-mode-path-expansion} coincide with those obtained from the transfer-matrix construction in Refs.~\cite{fendley-fermions-in-disguise,fermions-behind-the-disguise,unified-graph-th}; this connection is worked out in Section~\ref{sec:transfer-matrix-formalism}.
The rest of this section develops the Krylov generating function $\Phi_G(u)$, and Subsection~\ref{sec:proof-path-product-expansion} proves Eqs.~\eqref{eq:eigenoperator} and~\eqref{eq:anticommutation}.

\subsection{Generating function for the Krylov basis}
\label{subsec:krylov-basis}
Theorem~\ref{thm:Krylov-path-expansion} identifies $\Phi_G(u)$ with $\Xi_G(u;\chi)/P_G(u^2)$, and Corollary~\ref{cor:Psi-Phi-relation} extracts $\Psi_k$ as residues of $\Phi_G(u)$ at $u=u_k$.

The Krylov basis (in operator space) for the frustration graph $G$ is defined by the following recursive relation:
\begin{align}
	\label{eq:Krylov-def}
	\phi_{j+1} = \frac{1}{2}\qty[H_G, \phi_j],
\end{align}
where $\phi_0 \equiv \chi$ is the edge operator.

We define the generating function of the Krylov basis as
\begin{align}
	\label{eq:Phi-def}
	\Phi_G(u; \chi) \equiv \sum_{j=0}^{\infty} u^j \phi_j.
\end{align}
In the following, we omit the explicit dependence on the seed edge operator whenever it is clear from the context.

\begin{lem}\label{lem:Krylov-uniqueness}
	The generating function $\Phi_G(u;\chi)$ defined by~\eqref{eq:Phi-def} is the unique formal power series satisfying
	\begin{align}\label{eq:Krylov-diff-eq}
		\frac{u}{2} \qty[H_G, \Phi_G(u; \chi)] = \Phi_G(u; \chi) - \chi
	\end{align}
	with initial condition $\Phi_G(0; \chi) = \chi$.
\end{lem}
\begin{proof}
	The definition~\eqref{eq:Krylov-def} implies that $\Phi_G(u;\chi)$ satisfies Eq.~\eqref{eq:Krylov-diff-eq}.
	For uniqueness, the difference $\Delta\Phi(u) \equiv \Phi_1(u) - \Phi_2(u)$ of two formal power series solutions satisfies $\tfrac{u}{2} [H_G, \Delta \Phi(u)] = \Delta \Phi(u)$.
	Expanding $\Delta\Phi(u) = \sum_{n=0}^{\infty} u^n c_n$ and comparing coefficients gives the recursion $c_{n+1} = \tfrac{1}{2} [H_G, c_n]$ for $n \geq 0$.
	Since $c_0 = \Delta\Phi(0) = 0$, it follows inductively that $c_n = 0$ for all $n \geq 0$.
\end{proof}

\begin{define}\label{def:gen-path-product-operator}
	We define the path-product operator depending on a complex parameter $u$ by
	\begin{align}\label{eq:gen-path-product-operator}
		\Xi_G(u; \chi)
		\equiv
		\sum_{\pa \in \pathset{\extendedgraph{G}{\chi}}{{\chi}}}
		(-u)^{|\pa|-1}
		P_{\complementset{\pa}}(u^2)\,
		\pathprodH{\pa}
		\,.
	\end{align}
	The sum uses the rooted convention of Subsection~\ref{sec:main-result}.
\end{define}

The path-product operator~\eqref{eq:free-fermion-mode-path-expansion} is then
\begin{align}\label{eq:Xi-Psi-relation}
	\Psi_k = \frac{1}{\NN_k}\, \Xi_G(u_k; \chi)\,.
\end{align}

\begin{lem}\label{lem:odd-neighbor-classification}
	Let $G$ be claw-free and let $\pa=\ell_1\md\ell_2\md\cdots\md\ell_n$ be an induced path in $G$ with $n\geq2$.
	For an off-path vertex $\jbm$, let $T_\pa(\jbm)=\{\ell_i\in V(\pa): \jbm\sim\ell_i\}$ and assume that $|T_\pa(\jbm)|$ is odd.
	After reversing the listing of $\pa$ if necessary, exactly one of the following holds:
	\begin{enumerate}[label=\textup{(\roman*)},nosep]
		\item $T_\pa(\jbm)=\{\ell_n\}$; see Fig.~\ref{fig:non-path}(a).
		\item $T_\pa(\jbm)=\{\ell_{k-1},\ell_k,\ell_{k+1}\}$ for some $2\leq k\leq n-1$; see Fig.~\ref{fig:non-path}(b).
		\item $T_\pa(\jbm)=\{\ell_{k-1},\ell_k,\ell_n\}$ for some $2\leq k<n-1$; see Fig.~\ref{fig:non-path}(c).
	\end{enumerate}
	Moreover, the off-path vertices satisfying \textup{(i)}, together with $\ell_n$, form a clique.
\end{lem}
\begin{proof}
	Claw-freeness excludes two configurations, both drawn in Fig.~\ref{fig:odd-neighbor-forbidden}.
	An interior vertex $\ell_i\in T_\pa(\jbm)$ with $\ell_{i-1},\ell_{i+1}\notin T_\pa(\jbm)$ gives a claw centered at $\ell_i$ with leaves $\ell_{i-1},\ell_{i+1},\jbm$.
	Three vertices of $T_\pa(\jbm)$ no two of which are consecutive along $\pa$ are pairwise non-adjacent because $\pa$ is induced, and hence give a claw centered at $\jbm$.
	Decompose $T_\pa(\jbm)$ into maximal blocks of consecutive path vertices.
	The first exclusion forces every singleton block to be an endpoint of $\pa$.
	The second rules out three or more blocks, a block of length $\geq5$ (its first, third, and fifth vertices are pairwise non-consecutive), and, in a two-block configuration, a block of length $\geq3$ (its two extreme vertices together with any vertex of the other block).
	Since $|T_\pa(\jbm)|$ is odd, the only survivors are a single endpoint of $\pa$, a single block of three consecutive vertices, and an endpoint together with a disjoint pair of consecutive vertices.
	Reversing the listing if necessary places the isolated endpoint at $\ell_n$, which gives \textup{(i)}, \textup{(ii)}, and \textup{(iii)}.
	Finally, two vertices $\jbm,\jbm'$ satisfying \textup{(i)} are adjacent, since otherwise $\ell_{n-1},\jbm,\jbm'$ would be the leaves of a claw centered at $\ell_n$.
\end{proof}

\begin{figure}[t]
	\centering
	\begin{tikzpicture}[
			scale=0.98,
			line cap=round,
			line join=round,
			vertex/.style={circle, draw, thick, fill=white, minimum size=\vertexsize, inner sep=0pt},
			offpath/.style={circle, draw, thick, fill=black, minimum size=\vertexsize, inner sep=0pt},
			center/.style={circle, draw, double, double distance=1.2pt, thick, fill=white, minimum size=\vertexsize, inner sep=0pt},
			hamcenter/.style={circle, draw, double, double distance=1.2pt, thick, fill=black, minimum size=\vertexsize, inner sep=0pt},
			path/.style={thick},
			cont/.style={thick, dotted},
			clawedge/.style={thick},
			paneltitle/.style={font=\scriptsize, align=center},
			note/.style={font=\scriptsize, align=center}
		]
		\begin{scope}
			\node[paneltitle] at (1.9,2.35) {(a) Isolated touched interior vertex};
			\node[vertex] (ll) at (0,0) {};
			\node[vertex] (lm) at (0.95,0) {};
			\node[center] (li) at (1.9,0) {};
			\node[vertex] (lp) at (2.85,0) {};
			\node[vertex] (rr) at (3.8,0) {};
			\node[offpath] (j) at (1.9,1.35) {};
			\draw[path] (-0.25,0) -- (ll);
			\draw[cont] (-0.70,0) -- (-0.25,0);
			\draw[path] (ll) -- (lm) -- (li) -- (lp) -- (rr);
			\draw[path] (rr) -- (4.05,0);
			\draw[cont] (4.05,0) -- (4.50,0);
			\draw[clawedge] (li) -- (j);
			\node[center] at (li.center) {};
			\node[below=2pt] at (lm) {\small $\ell_{i-1}$};
			\node[below=2pt] at (li) {\small $\ell_i$};
			\node[below=2pt] at (lp) {\small $\ell_{i+1}$};
			\node[above=2pt] at (j) {\small $\jbm$};
			\node[above=2pt] at (4.25,0) {\small $\pa$};
			\node[note] at (1.9,-0.85) {claw centered at $\ell_i$};
		\end{scope}
		\begin{scope}[xshift=6.25cm]
			\node[paneltitle] at (1.85,2.35) {(b) Three separated touched vertices};
			\node[vertex] (la) at (0,0) {};
			\node[vertex] (xa) at (0.60,0) {};
			\node[vertex] (xb) at (1.25,0) {};
			\node[vertex] (lb) at (1.85,0) {};
			\node[vertex] (ya) at (2.45,0) {};
			\node[vertex] (yb) at (3.10,0) {};
			\node[vertex] (lc) at (3.70,0) {};
			\node[hamcenter] (jj) at (1.85,1.45) {};
			\draw[path] (-0.25,0) -- (la);
			\draw[cont] (-0.70,0) -- (-0.25,0);
			\draw[path] (la) -- (xa);
			\draw[cont] (xa) -- (xb);
			\draw[path] (xb) -- (lb) -- (ya);
			\draw[cont] (ya) -- (yb);
			\draw[path] (yb) -- (lc);
			\draw[path] (lc) -- (3.95,0);
			\draw[cont] (3.95,0) -- (4.40,0);
			\draw[clawedge] (jj) -- (la);
			\draw[clawedge] (jj) -- (lb);
			\draw[clawedge] (jj) -- (lc);
			\node[hamcenter] at (jj.center) {};
			\node[below=2pt] at (la) {\small $\ell_a$};
			\node[below=2pt] at (lb) {\small $\ell_b$};
			\node[below=2pt] at (lc) {\small $\ell_c$};
			\node[above=2pt] at (jj) {\small $\jbm$};
			\node[above=2pt] at (4.15,0) {\small $\pa$};
			\node[note] at (1.85,-0.85) {claw centered at $\jbm$};
		\end{scope}
	\end{tikzpicture}
	\caption{Claw obstructions excluded in Lemma~\ref{lem:odd-neighbor-classification}.
		Panel (a) has a claw centered at $\ell_i$, and panel (b) a claw centered at $\jbm$.
		Double circles mark the centers, and dotted segments denote omitted parts of the induced path.}
	\label{fig:odd-neighbor-forbidden}
\end{figure}

\begin{lem}\label{lem:single-path-commutator}
	Let $G$ be a claw-free frustration graph, let $K_s$ be a simplicial clique, and let $\chi$ be the associated edge operator.
	Let $\pa = \chi\md\ell_1 \md \cdots \md \ell_n$ be an induced path in $\extendedgraph{G}{\chi}$ with $|\pa| \geq 2$, and set $\ell_0=\chi$.
	Then the single-path commutator is
	\begin{align}
		\label{eq:commutator-single-path}
		\frac{1}{2}[H_G, \pathprodH{\pa}]
		 & =
		- b_{\ell_n}^2 \hc[\pa_{[;-1]}]
		- \sum_{\jbm \in \pathclique_{\pa}} \hc[\pa \md \jbm]
		+ \sum_{\jbm \in \pathsvd_{\pa}} \hc_{\jbm} \pathprodH{\pa}
		+ \sum_{\jbm \in \pathloop_{\pa}} \hc_{\jbm} \pathprodH{\pa}\,.
	\end{align}
	Here $\pathclique_{\pa} \coloneqq \Gamma[\ell_n] \setminus \Gamma[\pa_{[;-1]}]$ is the clique of vertices adjacent only to the terminal vertex $\ell_n$.
	The set $\pathsvd_{\pa}$ consists of off-path vertices adjacent to three consecutive path vertices $\ell_{k-1}, \ell_k, \ell_{k+1}$.
	The set $\pathloop_{\pa}$ consists of off-path vertices adjacent to $\ell_{k-1}, \ell_k$, and the terminal vertex $\ell_n$ for some $k < n{-}1$.
	These three classes partition the off-path vertices of $G$ with an odd number of path-neighbors; see Fig.~\ref{fig:non-path}.
	The first two terms are path products (shortening and extending $\pa$), while the last two are non-path operators.
	The minimal rooted path $(\chi)$ is excluded because the shortening term and the terminal Hamiltonian vertex are absent.
\end{lem}
\begin{proof}
	Since $H_G = \sum_{\jbm \in V} \hc_{\jbm}$, we evaluate $\frac{1}{2}[H_G, \pathprodH{\pa}] = \frac{1}{2}\sum_{\jbm}[\hc_{\jbm}, \pathprodH{\pa}]$ by commuting each $\hc_{\jbm}$ past $\pathprodH{\pa} = \chi\hc_{\ell_1}\cdots\hc_{\ell_n}$.

	\emph{On-path vertices} ($\jbm = \ell_k$, $1 \le k \le n$).
	Since $\pa$ is induced, the path vertex $\ell_k$ is adjacent within $\pa$ only to $\ell_{k-1}$ and, when $k<n$, to $\ell_{k+1}$.
	Pushing $\hc_{\ell_k}$ through $\pathprodH{\pa}$ from the left, it anticommutes only with the adjacent factor $\hc_{\ell_{k-1}}$, then squares to $b_{\ell_k}^2$; from the right, it anticommutes only with $\hc_{\ell_{k+1}}$.
	For an interior vertex ($k<n$), the two commutator terms give the same product with $\hc_{\ell_k}$ removed but opposite signs, so $[\hc_{\ell_k},\pathprodH{\pa}]=0$.
	For the terminal vertex ($k = n$), there is no right neighbor, so $\frac{1}{2}[\hc_{\ell_n}, \pathprodH{\pa}] = -b_{\ell_n}^2\, \hc[\pa_{[;-1]}]$.

	\emph{Off-path vertices} ($\jbm \notin \pa$).
	Since $\hc_{\jbm}$ anticommutes with each adjacent factor and commutes with each non-adjacent one, $[\hc_{\jbm}, \pathprodH{\pa}] = (1 - (-1)^{N_{\jbm}})\,\hc_{\jbm}\pathprodH{\pa}$, where $N_{\jbm} = \abs{\{\ell_i \in \pa : \jbm \sim \ell_i\}}$.
	This vanishes when $N_{\jbm}$ is even and equals $2\hc_{\jbm}\pathprodH{\pa}$ when $N_{\jbm}$ is odd.
	Since $K_s$ is simplicial, the extended graph $G_\chi$ is claw-free, so Lemma~\ref{lem:odd-neighbor-classification} applies to the induced path $\chi\md\ell_1\md\cdots\md\ell_n$ in $G_\chi$.
	The endpoint singled out in Lemma~\ref{lem:odd-neighbor-classification} cannot be the initial edge vertex $\chi$.
	Indeed, if $\jbm\sim\chi$, then $\jbm\in K_s$, while $\ell_1\in K_s$ and $K_s$ is a clique, so $\jbm\sim\ell_1$ as well.
	Thus no off-path vertex is adjacent only to the initial edge vertex, and the endpoint separated in a loop pattern cannot be $\chi$.
	The only odd-$N_{\jbm}$ possibilities are therefore the three classes in the statement.
	In the notation of Fig.~\ref{fig:non-path}, case~(a) gives $\pathclique_{\pa}$, case~(b) gives $\pathsvd_{\pa}$, and case~(c) gives $\pathloop_{\pa}$.
	For $\jbm \in \pathclique_{\pa}$, pushing $\hc_{\jbm}$ to the right end gives $\hc_{\jbm}\pathprodH{\pa} = -\hc[\pa \md \jbm]$.
	The vertices in $\pathsvd_{\pa}$ and $\pathloop_{\pa}$ do not extend the path; they remain as the two non-path sums in Eq.~\eqref{eq:commutator-single-path}.
\end{proof}

\begin{figure}[t]
	\centering
	\begin{tikzpicture}[
			vertex/.style={circle, draw, thick, fill=white, minimum size=\vertexsize, inner sep=0pt},
			scale=1.0
		]

		\begin{scope}[xshift=4.0cm, yshift=0cm]
			\fill[gray!15] (4.55, 1.15) -- (5.25, -0.05) -- (4.15, -1.10) -- cycle;

			\node[vertex] (chia) at (-0.7, 0) {};
			\node[vertex] (l1a) at (0.3, 0) {};
			\node[vertex] (lnm1a) at (2.3, 0) {};
			\node[vertex] (lna) at (3.3, 0) {};

			\draw[thick] (chia) -- (l1a);
			\draw[thick] (l1a) -- (1.0, 0);
			\draw[thick, dotted] (1.0, 0) -- (1.6, 0);
			\draw[thick] (1.6, 0) -- (lnm1a) -- (lna);

			\draw[decorate,decoration={brace,amplitude=4pt}] ([xshift=-4pt,yshift=8pt]chia.north west) -- ([xshift=4pt,yshift=8pt]lna.north east) node[midway,yshift=10pt] {\scriptsize $\pa$};

			\node[vertex, fill=black] (ja) at (4.55, 1.15) {};
			\node[vertex, draw=gray!40] (j1a) at (5.25, -0.05) {};
			\node[vertex, draw=gray!40] (j2a) at (4.15, -1.10) {};

			\draw[thick] (lna) -- (ja);

			\draw[thin, gray] (lna) -- (j1a);
			\draw[thin, gray] (lna) -- (j2a);

			\draw[thin, gray] (ja) -- (j1a);
			\draw[thin, gray] (ja) -- (j2a);
			\draw[thin, gray] (j1a) -- (j2a);

			\node[below=2pt] at (chia) {\small $\chi$};
			\node[below=2pt] at (l1a) {\small $\ellbm_1$};
			\node[below=2pt] at (lnm1a) {\small $\ellbm_{n-1}$};
			\node[below left=1pt] at (lna) {\small $\ellbm_n$};
			\node[above=2pt] at (ja) {\small $\jbm$};
			\node at (4.65, -0.05) {\small $\pathclique_{\pa}$};

			\node at (2.2, -2.0) {(a) Clique extension ($\jbm \in \pathclique_{\pa}$)};
		\end{scope}

		\draw[gray, dashed, thin] (-0.8, -2.7) -- (13.6, -2.7);

		\begin{scope}[xshift=0cm,yshift=-5.0cm]
			\node[vertex] (l1) at (0, 0) {};
			\node[vertex] (l2) at (1.2, 0) {};
			\node[vertex] (l3) at (2.4, -0.8) {};
			\node[vertex] (l4) at (3.6, 0) {};
			\node[vertex] (l5) at (4.8, 0) {};

			\node[vertex, fill=black] (j) at (2.4, 0.8) {};

			\draw[thick] (l1) -- ++(-0.3, 0);
			\draw[thick, dotted] (-0.3, 0) -- (-0.8, 0);
			\draw[thick] (l1) -- (l2) -- (l3) -- (l4) -- (l5);
			\draw[thick] (l5) -- ++(0.3, 0);
			\draw[thick, dotted] (5.1, 0) -- (5.6, 0);

			\draw[thick] (l2) -- (j);
			\draw[thick] (l3) -- (j);
			\draw[thick] (l4) -- (j);

			\node[below=2pt] at (l1) {\small $\ell_{k-2}$};
			\node[below=2pt] at (l2) {\small $\ell_{k-1}$};
			\node[below=2pt] at (l3) {\small $\ell_k$};
			\node[below=2pt] at (l4) {\small $\ell_{k+1}$};
			\node[below=2pt] at (l5) {\small $\ell_{k+2}$};
			\node[above=2pt] at (j) {\small $\jbm$};

			\node at (2.4, -2.0) {(b) Three-neighbor case ($\jbm \in \pathsvd_{\pa}$)};
		\end{scope}

		\draw[gray, dashed, thin] (6.2, -3.2) -- (6.2, -7.5);

		\begin{scope}[xshift=8.2cm,yshift=-5.0cm]
			\node[vertex] (lkm2) at (-0.7, 0) {};
			\node[vertex] (lkm1) at (0.27, 0) {};

			\node[vertex, fill=black] (j) at (1.21, 0.54) {};
			\node[vertex] (ln) at (1.96, 1.29) {};
			\node[vertex] (lnm1) at (3.04, 1.29) {};
			\node[vertex] (lkp2) at (3.04, -1.29) {};
			\node[vertex] (lkp1) at (1.96, -1.29) {};
			\node[vertex] (lk) at (1.21, -0.54) {};

			\draw[thick] (lkm2) -- ++(-0.3, 0);
			\draw[thick, dotted] (-1.0, 0) -- (-1.4, 0);
			\draw[thick] (lkm2) -- (lkm1) -- (lk);

			\draw[thick] (lk) -- (lkp1) -- (lkp2);

			\draw[thick] (lkp2) -- (3.40, -1.07);
			\draw[thick, dotted] (2.5, 0) ++(310:1.4) arc (310:410:1.4);
			\draw[thick] (3.40, 1.07) -- (lnm1);

			\draw[thick] (lnm1) -- (ln) -- (j);

			\draw[thick] (lkm1) -- (j);
			\draw[thick] (lk) -- (j);

			\node[below=2pt] at (lkm2) {\small $\ell_{k-2}$};
			\node[below=2pt] at (lkm1) {\small $\ell_{k-1}$};
			\node[above left=1pt] at (j) {\small $\jbm$};
			\node[above=2pt] at (ln) {\small $\ell_n$};
			\node[above=2pt] at (lnm1) {\small $\ell_{n-1}$};
			\node[below=2pt] at (lkp2) {\small $\ell_{k+2}$};
			\node[below=2pt] at (lkp1) {\small $\ell_{k+1}$};
			\node[below left=1pt] at (lk) {\small $\ell_k$};

			\node at (2.4, -2.0) {(c) Loop case ($\jbm \in \pathloop_{\pa}$)};
		\end{scope}

	\end{tikzpicture}
	\caption{The three off-path configurations in Lemma~\ref{lem:single-path-commutator}.
		Panel (a) extends $\pa$ through $\jbm\in\pathclique_{\pa}$; the shaded face and gray edges show the clique $\{\ell_n\}\cup\pathclique_{\pa}$.
		Panels (b) and (c) give the pairwise-canceling non-path terms, and the induced cycle in panel (c) is odd by even-hole-freeness.
		Dotted segments denote omitted parts of $\pa$.}
	\label{fig:non-path}
\end{figure}

\begin{lem}\label{lem:path-cancellation-pairings}
	Let $G$ be a claw-free frustration graph, and let $\pa=\ell_0\md\ell_1\md\cdots\md\ell_n$ be either an induced path in $G$ or, for a simplicial clique with edge operator $\chi$, a rooted induced path in $\extendedgraph{G}{\chi}$ with $\ell_0=\chi$.
	The following two configurations admit pairings.
	\begin{enumerate}[label=\textup{(\roman*)},nosep]
		\item \emph{Three-neighbor replacement.}
		      If an off-path vertex $\jbm$ is adjacent exactly to $\ell_{k-1},\ell_k,\ell_{k+1}$ among the path vertices, where $1\le k\le n-1$, let
		      \begin{equation}
			      \pa'=\ell_0\md\cdots\md\ell_{k-1}\md\jbm\md\ell_{k+1}\md\cdots\md\ell_n.
		      \end{equation}
		      Then $\pa'$ is induced, $\ell_k\in\pathsvd_{\pa'}$, the partner of $(\pa',\ell_k)$ is $(\pa,\jbm)$, and $\Gamma[\pa]=\Gamma[\pa']$.
		\item \emph{Loop pairing.}
		      If an off-path vertex $\jbm$ is adjacent exactly to $\ell_{k-1},\ell_k,\ell_n$ among the path vertices, where $1\le k<n-1$, and the induced cycle $(\ell_k,\ell_{k+1},\ldots,\ell_n,\jbm)$ is odd, let
		      \begin{equation}
			      \pa'=\ell_0\md\cdots\md\ell_{k-1}\md\jbm\md\ell_n\md\ell_{n-1}\md\cdots\md\ell_{k+1}.
		      \end{equation}
		      Then $\pa'$ is induced, $\ell_k\in\pathloop_{\pa'}$, the partner of $(\pa',\ell_k)$ is $(\pa,\jbm)$, and $\Gamma[\pa]=\Gamma[\pa']$.
	\end{enumerate}
	If $G$ is even-hole-free, the cycle in~\textup{(ii)} is automatically odd.
	In both cases $|\pa|=|\pa'|$ and $P_{\complementset{\pa}}=P_{\complementset{\pa'}}$.
	Moreover,
	\begin{equation}
		\hc_{\jbm}\pathprodH{\pa}+\hc_{\ell_k}\pathprodH{\pa'}=0.
	\end{equation}
	Thus the pair cancels whenever its scalar weight depends only on the path size and residual graph.
\end{lem}
\begin{proof}
	Set $X_0=\chi$ in the rooted case and $X_0=\hc_{\ell_0}$ otherwise.
	In case~\textup{(i)}, all path edges of $\pa'$ are present by construction, a chord of $\pa'$ avoiding $\jbm$ would be a chord of $\pa$, and the hypothesis on $\jbm$ excludes every other chord.
	Hence $\pa'$ is induced.
	If $x$ is adjacent to $\ell_k$ and is not already in the closed neighborhood of $\pa\setminus\{\ell_k\}$, then $x$ is non-adjacent to $\ell_{k-1}$ and $\ell_{k+1}$.
	If $x$ were also non-adjacent to $\jbm$, then $x,\ell_{k-1},\ell_{k+1}$ would be three pairwise non-adjacent neighbors of $\ell_k$, a claw.
	Thus $x\in\Gamma[\jbm]$, and interchanging $\ell_k$ and $\jbm$ gives the reverse inclusion, so $\Gamma[\pa]=\Gamma[\pa']$.
	The same construction applied to $(\pa',\ell_k)$ returns $(\pa,\jbm)$, so the map is its own inverse.
	The paired Clifford contributions therefore cancel:
	\begin{equation}
		\hc_{\jbm}\pathprodH{\pa}+\hc_{\ell_k}\pathprodH{\pa'}
		=
		-X_0\hc_{\ell_1}\cdots\hc_{\ell_{k-1}}
		\bigl(\hc_{\jbm}\hc_{\ell_k}+\hc_{\ell_k}\hc_{\jbm}\bigr)
		\hc_{\ell_{k+1}}\cdots\hc_{\ell_n}
		=0.
	\end{equation}
	Here the product $\hc_{\ell_1}\cdots\hc_{\ell_{k-1}}$ is omitted when $k=1$.
	The minus sign comes from the anticommutation with the factor at $\ell_{k-1}$.
	The last equality uses $\jbm\sim\ell_k$.

	In case~\textup{(ii)}, the cycle is odd by hypothesis, which is equivalent to $n-k$ being odd.
	The vertex sequence defining $\pa'$ is a path: the old path supplies the prefix and reversed-segment edges, and the two bridge edges are $\jbm\sim\ell_{k-1}$ and $\jbm\sim\ell_n$.
	Indeed, any chord of $\pa'$ whose endpoints both lie on the original path would be a chord of $\pa$.
	A chord incident to $\jbm$ is also impossible, since $\ell_{k-1}$ and $\ell_n$ are consecutive to $\jbm$ in $\pa'$ and $\ell_k$ is not a vertex of $\pa'$.
	Thus $\pa'$ is induced.
	The same claw argument, applied at the replacement point, gives $\Gamma[\pa]=\Gamma[\pa']$.
	Applying the loop pairing to $(\pa',\ell_k)$ returns $(\pa,\jbm)$, so the map is again its own inverse.
	First,
	\begin{equation}
		\hc_{\jbm}\pathprodH{\pa}
		=
		-X_0\hc_{\ell_1}\cdots\hc_{\ell_{k-1}}
		\hc_{\jbm}\hc_{\ell_k}
		\hc_{\ell_{k+1}}\cdots\hc_{\ell_n}.
	\end{equation}
	For the partner path,
	\begin{align}
		\hc_{\ell_k}\pathprodH{\pa'}
		 & =
		\hc_{\ell_k}X_0\hc_{\ell_1}\cdots\hc_{\ell_{k-1}}
		\hc_{\jbm}\hc_{\ell_n}\hc_{\ell_{n-1}}\cdots\hc_{\ell_{k+1}}
		\nonumber \\
		 & =
		-X_0\hc_{\ell_1}\cdots\hc_{\ell_{k-1}}
		\hc_{\ell_k}\hc_{\jbm}
		\hc_{\ell_n}\hc_{\ell_{n-1}}\cdots\hc_{\ell_{k+1}}
		\nonumber \\
		 & =
		-X_0\hc_{\ell_1}\cdots\hc_{\ell_{k-1}}
		\hc_{\ell_k}\hc_{\jbm}
		\hc_{\ell_{k+1}}\cdots\hc_{\ell_n}.
	\end{align}
	In the last equality, reversing the segment gives no sign, since only adjacent path vertices anticommute and $(-1)^{n-k-1}=+1$ because $n-k$ is odd.
	Therefore
	\begin{equation}
		\hc_{\jbm}\pathprodH{\pa}+\hc_{\ell_k}\pathprodH{\pa'}
		=
		-X_0\hc_{\ell_1}\cdots\hc_{\ell_{k-1}}
		\bigl(\hc_{\jbm}\hc_{\ell_k}+\hc_{\ell_k}\hc_{\jbm}\bigr)
		\hc_{\ell_{k+1}}\cdots\hc_{\ell_n}
		=0,
	\end{equation}
	again because $\jbm\sim\ell_k$.
\end{proof}

\begin{thm}\label{thm:Krylov-path-expansion}
	Let $G$ be a connected ECF graph, let $K_s$ be a simplicial clique, and let $\chi$ be the associated edge operator.
	The generating function of the Krylov basis~\eqref{eq:Krylov-def} is given by the generalized path-product operator~\eqref{eq:gen-path-product-operator}:
	\begin{align}
		\label{eq:Krylov-path-expansion}
		\Phi_G(u)
		 & =
		\frac{\Xi_G(u)}{P_G(u^2)}
		=
		\sum_{\pa \in \pathsetG}
		(-u)^{|\pa|-1}
		\frac{P_{\complementset{\pa}}(u^2)}{P_G(u^2)}
		\pathprodH{\pa}
		\,.
	\end{align}
\end{thm}
\begin{proof}
	With $\Xi_G(u)$ as in Definition~\ref{def:gen-path-product-operator}, the theorem is equivalent to $\Xi_G(u) = P_G(u^2)\Phi_G(u)$.
	The Krylov equation~\eqref{eq:Krylov-diff-eq} implies that $P_G(u^2)\Phi_G(u)$ satisfies
	\begin{align}
		\label{eq:Xi-G-eq}
		\frac{u}{2} \qty[H_G, \Xi_G(u)] = \Xi_G(u) - P_G(u^2) \chi\,,
	\end{align}
	and the solution analytic at $u=0$ is unique (by the same argument as in Lemma~\ref{lem:Krylov-uniqueness}).
	It therefore suffices to verify that $\Xi_G(u)$ satisfies Eq.~\eqref{eq:Xi-G-eq}.

	Step~1 shows that all non-path contributions cancel pairwise, reducing the commutator to
	\begin{align}
		\label{eq:commutator-path-expansion}
		\tfrac{u}{2}[H_G, \Xi_G(u)] = \sum_{\pa \in \pathsetG} C_{\pa}\, \pathprodH{\pa}\,.
	\end{align}

	\paragraph{Step 1: Cancellation of non-path terms.}
	For the minimal rooted path $(\chi)$, Eq.~\eqref{eq:edge-operator} gives
	\begin{align}
		\frac{1}{2}[H_G, \chi]
		=
		-\sum_{\jbm \in K_s}\hc[\chi\md\jbm]\,,
	\end{align}
	so it contributes only path-extension terms.
	For $|\pa| \geq 2$, the non-path terms of Eq.~\eqref{eq:commutator-single-path} are the $\pathsvd_{\pa}$ and $\pathloop_{\pa}$ summands of the sum over induced paths in Eq.~\eqref{eq:gen-path-product-operator}.
	By Lemma~\ref{lem:path-cancellation-pairings}, the three-neighbor and loop-type non-path summands are paired so that the two Clifford contributions in each pair sum to zero.
	Both pairings preserve the path size and the closed neighborhood, hence preserve the scalar coefficient $(-u)^{|\pa|-1}P_{\complementset{\pa}}(u^2)$.
	Therefore every non-path Clifford monomial has zero total coefficient in $\tfrac{u}{2}[H_G,\Xi_G(u)]$, and Eq.~\eqref{eq:commutator-path-expansion} follows.

	\paragraph{Step 2: Computation of coefficients.}

	We compute $C_{\pa}$ in Eq.~\eqref{eq:commutator-path-expansion} from the two path-product terms of Lemma~\ref{lem:single-path-commutator}.
	After multiplying Eq.~\eqref{eq:commutator-single-path} by the prefactor $u(-u)^{|\pa'|-1}P_{\complementset{\pa'}}(u^2)$ of a summand path $\pa'$ in $\frac{u}{2}[H_G,\Xi_G(u)]$, the only ways to generate $\pathprodH{\pa}$ for $|\pa|\geq2$ are:
	\begin{enumerate}[nosep]
		\item the shortening term applied to the longer paths $\pa \md \jbm$ with $\jbm \in \pathclique_{\pa}$, contributing
		      \begin{equation*}
			      (-u)^{|\pa|+1} \sum_{\jbm \in \pathclique_{\pa}} b_{\jbm}^2 P_{\complementset{\pa \md \jbm}}(u^2);
		      \end{equation*}
		\item the extension term applied to the shorter path $\pa_{[;-1]}$, contributing
		      \begin{equation*}
			      (-u)^{|\pa|-1} P_{\complementset{\pa_{[;-1]}}}(u^2).
		      \end{equation*}
	\end{enumerate}

	By claw-freeness $\pathclique_{\pa}$ is a clique, and the closed neighborhoods satisfy
	\begin{equation}
		\Gamma[\pa]=\Gamma[\pa_{[;-1]}]\cup \pathclique_{\pa},
		\qquad
		\Gamma[\pa\md\jbm]=\Gamma[\pa_{[;-1]}]\cup \Gamma[\jbm]
		\quad (\jbm\in\pathclique_{\pa}).
	\end{equation}
	These identities imply $\complementset{\pa}=\complementset{\pa_{[;-1]}}\setminus\pathclique_{\pa}$ and $\complementset{\pa\md\jbm}=\complementset{\pa_{[;-1]}}\setminus\Gamma[\jbm]$.

	The total coefficient is therefore
	\begin{align*}
		C_{\pa}
		 & =
		(-u)^{|\pa|-1} \qty[u^2 \sum_{\jbm \in \pathclique_{\pa}} b_{\jbm}^2 P_{\complementset{\pa \md \jbm}}(u^2) + P_{\complementset{\pa_{[;-1]}}}(u^2)]
		\\
		 & = (-u)^{|\pa|-1} P_{\complementset{\pa}}(u^2)
		\,,
	\end{align*}
	where in the last equality, we used the recursion in Eq.~\eqref{eq:independence-polynomial-recursion}: $P_{\complementset{\pa_{[;-1]}}}(x) = P_{\complementset{\pa}}(x) - x \sum_{\jbm \in \pathclique_{\pa}} b_{\jbm}^2 P_{\complementset{\pa \md \jbm}}(x)$.

	For $|\pa| = 1$ (i.e., $\pa = (\chi)$), only source~1 contributes, since $\pa_{[;-1]}$ does not exist.
	We have $\pathclique_{(\chi)} = K_s$ and $\complementset{(\chi)} = G \setminus K_s$.
	The independence-polynomial recursion gives
	\begin{align*}
		C_{(\chi)}
		 & =
		u^2 \sum_{\jbm \in K_s} b_{\jbm}^2 P_{G \setminus \Gamma[\jbm]}(u^2)
		\\
		 & = -P_G(u^2) + P_{G \setminus K_s}(u^2)\,.
	\end{align*}
	The term $-P_G(u^2) \chi$ matches the right-hand side of Eq.~\eqref{eq:Xi-G-eq}, while $P_{G \setminus K_s}(u^2) \chi = P_{\complementset{(\chi)}}(u^2) \chi$ is the $|\pa| = 1$ contribution to $\Xi_G(u)$.

	Combining both cases,
	\begin{align*}
		\sum_{\pa} C_{\pa} \pathprodH{\pa}
		 & =
		\sum_{\pa} (-u)^{|\pa|-1} P_{\complementset{\pa}}(u^2)\, \pathprodH{\pa} - P_G(u^2) \chi
		\\
		 & =
		\Xi_G(u) - P_G(u^2) \chi\,.
	\end{align*}
	The uniqueness of the solution analytic at $u = 0$ follows from Lemma~\ref{lem:Krylov-uniqueness}.
\end{proof}

From Eq.~\eqref{eq:Xi-Psi-relation} and Theorem~\ref{thm:Krylov-path-expansion}, the fermionic operators defined by Eq.~\eqref{eq:free-fermion-mode-path-expansion} are obtained from the generating function at its poles.
\begin{cor}\label{cor:Psi-Phi-relation}
	The free-fermion mode $\Psi_k$ is given by
	\begin{align}
		\Psi_k = \frac{1}{\NN_k} \lim_{u \rightarrow u_k} P_G(u^2) \Phi_G(u)
		\,.
	\end{align}
\end{cor}

\begin{lem}\label{lem:krylov-anticommutator}
	The Krylov basis elements satisfy
	\begin{align}
		\qty{\phi_i, \phi_{j}}
		=
		2 (-1)^{j}\,[u^{i+j}]\,\frac{P_{G \setminus K_s}(u^2)}{P_G(u^2)}\,.
	\end{align}
\end{lem}
\begin{proof}
	First consider the case in which one index is zero.
	Theorem~\ref{thm:Krylov-path-expansion} gives
	\begin{align}
		\qty{\Phi_G(u),\chi}
		=
		2\frac{P_{G\setminus K_s}(u^2)}{P_G(u^2)},
	\end{align}
	because only the minimal rooted path $(\chi)$ has a nonzero anticommutator with $\chi$.
	Comparing coefficients proves the claimed formula whenever $i=0$ or $j=0$.
	We now prove the remaining cases by induction on $i+j$.
	Assume the formula holds for all $(i,j)$ with $i+j\le n$, and let $i,j\ge1$ with $i+j=n+1$.
	From the definition of the Krylov basis~\eqref{eq:Krylov-def}, we have
	\begin{align}
		\qty{\phi_i, \phi_j}
		 & =
		\frac{1}{2}\qty{\qty[H_G, \phi_{i-1}], \phi_j}
		\nonumber \\
		 & =
		\frac{1}{2} \qty[H_G, \qty{\phi_{i-1}, \phi_j}]
		-
		\frac{1}{2} \qty{\phi_{i-1}, \qty[H_G, \phi_{j}]}
		\nonumber \\
		 & =
		-
		\qty{\phi_{i-1}, \phi_{j+1}}
		\nonumber \\
		 & =
		(-1)^{i}
		\qty{\phi_{0}, \phi_{i+j}}
		\nonumber \\
		 & =
		(-1)^{i}
		\qty{\phi_{n+1}, \chi}
		\nonumber \\
		 & =
		2 (-1)^{j}\,[u^{n+1}]\,\frac{ P_{G \setminus K_s}(u^2)}{P_G(u^2)}.
	\end{align}
	In the second line we used $\qty[H_G, \qty{\phi_{i-1}, \phi_j}] = 0$ from the induction hypothesis (since $\{\phi_{i-1},\phi_j\}$ is a scalar for $i-1+j = n$).
	Repeating the same argument $i-1$ more times gives $(-1)^i\{\phi_0,\phi_{i+j}\}$.
	Since $P_{G \setminus K_s}(u^2)/P_G(u^2)$ is an even power series in $u$, the coefficient vanishes for odd $i+j$; for even $i+j$ we have $i \equiv j \pmod{2}$, so $(-1)^i = (-1)^j$ in the last line.
\end{proof}

\begin{thm}\label{thm:Phi-anti-commu}
	The anticommutator of the generating function at different parameters is given by
	\begin{align}
		 &
		\qty{
			\Phi_G(u), \Phi_G(v)
		}
		=
		\frac{2}{u+v}
		\qty[
			u\frac{P_{G \setminus K_s}(u^2)}{P_G(u^2)}
			+
			v
			\frac{P_{G \setminus K_s}(v^2)}{P_G(v^2)}
		].
		\label{eq:Phi-anti-commu}
	\end{align}
\end{thm}

\begin{proof}
	We first prove
	\begin{align}
		\label{eq:Krylov-basis-anticommutator-commutator}
		\qty[H_G, \qty{\Phi_G(u), \Phi_G(v)}] = 0
	\end{align}
	for any $u, v$.
	Expanding both generating functions by Eq.~\eqref{eq:Phi-def} gives $\qty{\Phi_G(u), \Phi_G(v)}=\sum_{i,j\ge 0} u^i v^j \qty{\phi_i,\phi_j}$.
	By Lemma~\ref{lem:krylov-anticommutator}, each $\qty{\phi_i,\phi_j}$ is a scalar and hence commutes with $H_G$, which gives Eq.~\eqref{eq:Krylov-basis-anticommutator-commutator}.
	Using this, we have
	\begin{align}
		0 = & \qty[H_G, \qty{\Phi_G(u), \Phi_G(v)}]
		\nonumber                                   \\
		=   &
		\qty{\qty[H_G, \Phi_G(u)], \Phi_G(v)}
		+
		\qty{\Phi_G(u), \qty[H_G, \Phi_G(v)]}
		\nonumber                                   \\
		=   &
		\frac{2}{u}\qty{\Phi_G(u) - \chi, \Phi_G(v)}
		+
		\frac{2}{v}\qty{\Phi_G(u), \Phi_G(v) - \chi}
		\nonumber                                   \\
		=   &
		\qty[\frac{2}{u} + \frac{2}{v}] \qty{\Phi_G(u), \Phi_G(v)}
		-
		\frac{4}{u} \frac{P_{G \setminus K_s}(v^2)}{P_G(v^2)}
		-
		\frac{4}{v} \frac{P_{G \setminus K_s}(u^2)}{P_G(u^2)},
	\end{align}
	where we used $\qty{\Phi_G(u),\chi}=2P_{G \setminus K_s}(u^2)/P_G(u^2)$, which follows from Eq.~\eqref{eq:Krylov-path-expansion}.
	Then we have
	\begin{align}
		\label{eq:Krylov-basis-anticommutator}
		\qty{\Phi_G(u), \Phi_G(v)}
		 & =
		\frac{1}{\frac{2}{u} + \frac{2}{v}}
		\qty[
			\frac{4}{u} \frac{P_{G \setminus K_s}(v^2)}{P_G(v^2)}
			+
			\frac{4}{v} \frac{P_{G \setminus K_s}(u^2)}{P_G(u^2)}
		]
		\nonumber \\
		 & =
		\frac{2}{u + v}
		\qty[
			v \frac{P_{G \setminus K_s}(v^2)}{P_G(v^2)}
			+
			u \frac{P_{G \setminus K_s}(u^2)}{P_G(u^2)}
		].
	\end{align}
\end{proof}

Substituting the path-product expansion of $\Phi_G(u)$ into Theorem~\ref{thm:Phi-anti-commu} makes all non-scalar Clifford components of the anticommutator cancel, leaving the following scalar identity for independence polynomials.
\begin{cor}\label{cor:path-product-identity}
	Let $G$ be a connected ECF graph, let $K_s$ be a simplicial clique, and let $\chi$ be the associated edge operator.
	Then the following identity holds:
	\begin{align}
		\sum_{\pa \in \pathset{\extendedgraph{G}{\chi}}{\chi}}
		(-uv)^{|\pa|-1}
		P_{\complementset{\pa}}(u^2)
		P_{\complementset{\pa}}(v^2)
		\qty[\prod_{\jbm \in \pa} b_{\jbm}^2]
		=
		\frac{P_G(u^2) P_G(v^2)}{u + v}
		\qty[
			v \frac{P_{G \setminus K_s}(v^2)}{P_G(v^2)}
			+
			u \frac{P_{G \setminus K_s}(u^2)}{P_G(u^2)}
		],
	\end{align}
	where $b_{\chi}^2 = 1$.
	Setting $u = v$ yields
	\begin{align}\label{eq:path-square-identity}
		\sum_{\pa \in \pathset{\extendedgraph{G}{\chi}}{{\chi}}}
		(-u^2)^{|\pa|-1}
		\qty[P_{\complementset{\pa}}(u^2)]^2
		\qty[\prod_{\jbm \in \pa} b_{\jbm}^2]
		=
		P_G(u^2) P_{G \setminus K_s}(u^2),
	\end{align}
	while taking the limit $v \rightarrow -u$ gives
	\begin{align}\label{eq:path-square-identity-vminus}
		\sum_{\pa \in \pathset{\extendedgraph{G}{\chi}}{{\chi}}}
		u^{2(|\pa|-1)}
		\qty[P_{\complementset{\pa}}(u^2)]^2
		\qty[\prod_{\jbm \in \pa} b_{\jbm}^2]
		=
		\qty[P_G(u^2)]^2
		\frac{\partial}{\partial u}
		\qty[
			u \frac{P_{G \setminus K_s}(u^2)}{P_G(u^2)}
		].
	\end{align}

\end{cor}

\begin{lem}\label{lem:normalization-positivity}
	Assume that $u_k>0$ is a simple root of $P_G(u^2)=0$ and that $P_{G\setminus K_s}(u_k^2)\ne0$.
	Then
	\begin{align}
		-u_k^2 P_{G\setminus K_s}(u_k^2)P'_G(u_k^2)>0.
	\end{align}
	Consequently $\NN_k$ in Eq.~\eqref{eq:normalization-factor-Psi} and $-P_{G\setminus K_s}(u_k^2)/P'_G(u_k^2)$ in Theorem~\ref{thm:Phi-Psi-decomposition} are real and positive.
\end{lem}
\begin{proof}
	We evaluate Eq.~\eqref{eq:path-square-identity-vminus} at $u=u_k$ by continuity.
	Since the root is simple,
	\begin{align*}
		P_G(u^2)=2u_kP'_G(u_k^2)(u-u_k)+O((u-u_k)^2).
	\end{align*}
	Therefore the right-hand side tends to
	\begin{align*}
		-2u_k^2P_{G\setminus K_s}(u_k^2)P'_G(u_k^2).
	\end{align*}
	The left-hand side tends to
	\begin{align*}
		\sum_{\pa\in\pathset{\extendedgraph{G}{\chi}}{\chi}}
		u_k^{2(|\pa|-1)}
		\qty[P_{\complementset{\pa}}(u_k^2)]^2
		\prod_{\jbm\in\pa}b_{\jbm}^2.
	\end{align*}
	This is a sum of nonnegative real terms and is not zero, because the minimal rooted path $(\chi)$ contributes $P_{G\setminus K_s}(u_k^2)^2>0$.
	Thus the right-hand limit is strictly positive.
\end{proof}

\subsection{Mode decomposition of the Krylov generating function}

Corollary~\ref{cor:Psi-Phi-relation} also leads to the decomposition of the generating function:
\begin{thm}\label{thm:Phi-Psi-decomposition}
	Let $G$ be a connected ECF graph, let $K_s$ be a simplicial clique, and fix the associated edge operator $\chi$.
	Then
	\begin{align}
		\label{eq:Phi-Psi-decomposition}
		\Phi_G(u) & =
		\sum_{\substack{k = - \indepnum{G} \\ k \neq 0}}^{\indepnum{G}}
		\frac{C_k}{u_k - u} \Psi_k
		+
		C_0 \Psi_{0}\,,
	\end{align}
	where the roots satisfy $u_{-k} = -u_k$, and the coefficients are given by
	\begin{align}
		C_k^2
		=
		-\frac{P_{G \setminus K_s}(u_k^2)}{P_G'(u_k^2)}\,,
		\qquad
		C_{-k} = -C_k
		\quad \text{for } k=1,\ldots,\indepnum{G}\,,
	\end{align}
	where the sign of $C_k$ is determined by $C_k\, P_G'(u_k^2) < 0$.
	The zero-mode coefficient $C_0\ge0$ is defined by
	\begin{align}
		\label{eq:zero-mode-weight}
		C_0^2
		\equiv
		\lim_{x\rightarrow \infty}
		\frac{P_{G \setminus K_s}(x)}{P_G(x)}
		=
		\frac{\indepcoeff{G\setminus K_s}{\indepnum{G}}}{\indepcoeff{G}{\indepnum{G}}}\,.
	\end{align}
	Here the numerator is zero when $\indepnum{G \setminus K_s}<\indepnum{G}$.
	When $C_0\neq0$, equivalently when $\indepnum{G \setminus K_s} = \indepnum{G}$, the last term is present and the Majorana zero mode is defined by
	\begin{align}
		\label{eq:zero-mode-limit}
		\lim_{u \rightarrow \infty} \Phi_G(u)
		=
		C_0 \Psi_0\,.
	\end{align}
	In this case, $\Psi_0^2 = \mathbb{1}$ and $\qty{\Psi_0, \Psi_k} = 0$ for $k \neq 0$.
	If $C_0=0$, the last term is absent.
\end{thm}
\begin{proof}
	Since the roots $\pm u_j$ are simple, $\Phi_G(u)$ has only simple poles at these points, and Corollary~\ref{cor:Psi-Phi-relation} shows that the principal part at $u=u_j$ is proportional to $\Psi_j$.
	Fix $k>0$ and define $C_k$ by the principal part at $u=u_k$ in the convention of Eq.~\eqref{eq:Phi-Psi-decomposition}, so that $\Phi_G(u)+C_k\Psi_k/(u-u_k)$ is holomorphic at $u=u_k$.
	Since $P_G(u^2)=2u_kP_G'(u_k^2)(u-u_k)+O((u-u_k)^2)$, only this principal part contributes to
	\begin{align}
		\lim_{u \to u_k} P_G(u^2) \Phi_G(u)
		=
		-2u_k P_G'(u_k^2) C_k \Psi_k\,.
	\end{align}
	Comparing with Corollary~\ref{cor:Psi-Phi-relation} gives
	\begin{align}
		\Psi_k
		=
		\frac{1}{\NN_k}
		\qty[-2u_k P_G'(u_k^2) C_k]\Psi_k,
	\end{align}
	and hence $C_k = -\NN_k/(2u_k P_G'(u_k^2))$.
	Since $u_k>0$ and $\NN_k>0$, this fixes the sign of $C_k$ as $-\sgn(P_G'(u_k^2))$.
	Squaring this expression and using Eq.~\eqref{eq:normalization-factor-Psi} gives
	\begin{align}
		C_k^2 = -\frac{P_{G\setminus K_s}(u_k^2)}{P_G'(u_k^2)}\,.
	\end{align}
	The relation $C_{-k}=-C_k$ follows from $u_{-k}=-u_k$ and $\NN_{-k}=\NN_k$.

	Subtracting the principal parts at all the poles $\pm u_j$ leaves a rational operator with no finite poles, that is, a polynomial in $u$; we must show that it is constant, and that it vanishes when $C_0=0$.
	Both follow from the degree of $\Xi_G(u)=P_G(u^2)\Phi_G(u)$, to be compared with $\deg P_G(u^2)=2\indepnum{G}$.

	By Theorem~\ref{thm:Krylov-path-expansion}, the term of $\Xi_G(u)$ associated with a rooted path $\pa=\chi\md\ell_1\md\cdots\md\ell_n$ has degree at most $n+2\indepnum{\complementset{\pa}}$.
	Alternate path vertices are pairwise nonadjacent, and no path vertex is adjacent to $\complementset{\pa}$, so a maximum independent set in $\complementset{\pa}$ remains independent after adding either $\ell_1,\ell_3,\ldots$ or $\ell_2,\ell_4,\ldots$ to it.
	The first choice gives $\ceil{n/2}+\indepnum{\complementset{\pa}}\le\indepnum{G}$.
	The second one stays inside $G\setminus K_s$, because $\ell_1$ is the only path vertex in $K_s$ and $\complementset{\pa}\subseteq\complementset{(\chi)}=G\setminus K_s$, and gives $\floor{n/2}+\indepnum{\complementset{\pa}}\le\indepnum{G\setminus K_s}$.
	With $n\le2\ceil{n/2}$ and $n\le2\floor{n/2}+1$, the two independent sets bound the degree of every path term, so that
	\begin{align}
		\deg\Xi_G(u)\le 2\indepnum{G}
		\qquad\text{and}\qquad
		\deg\Xi_G(u)\le 2\indepnum{G\setminus K_s}+1.
	\end{align}

	By the first bound $\Phi_G(u)$ is bounded at infinity, so the polynomial is the constant $\lim_{u\to\infty}\Phi_G(u)$.
	If $C_0=0$, that is $\indepnum{G\setminus K_s}<\indepnum{G}$, the second bound gives $\deg\Xi_G(u)<\deg P_G(u^2)$, and this constant vanishes.

	When $C_0>0$, setting $v=u$ in Eq.~\eqref{eq:Phi-anti-commu} gives
	\begin{align}
		\Phi_G(u)^2 = \frac{P_{G \setminus K_s}(u^2)}{P_{G}(u^2)}\,.
	\end{align}
	Taking $u\to\infty$ and using Eq.~\eqref{eq:zero-mode-weight}, the operator $\Psi_0$ defined by Eq.~\eqref{eq:zero-mode-limit} satisfies $\Psi_0^2=\mathbb{1}$.
	Rewriting Eq.~\eqref{eq:Krylov-diff-eq} as $[H_G,\Phi_G(u)]=2(\Phi_G(u)-\chi)/u$ and taking $u\to\infty$ gives $[H_G,\Psi_0]=0$.
	Finally, Eq.~\eqref{eq:Phi-anti-commu} gives
	\begin{align}
		\qty{\Psi_k, \Phi_G(u)}
		=
		\frac{2}{\NN_k} \frac{u_k P_{G\setminus K_s}(u_k^2)}{u_k + u}\,,
	\end{align}
	which vanishes as $u\to\infty$, and hence $\qty{\Psi_0,\Psi_k}=0$ for $k\neq0$.
\end{proof}

Setting $u=0$ in Theorem~\ref{thm:Phi-Psi-decomposition} and using $\Phi_G(0)=\chi$ gives the following generalization of the Fendley-model result~\cite{sajat-ffd-corr} to arbitrary ECF frustration graphs:
\begin{cor}\label{cor:edge-decomposition}
	The edge operator decomposes as
	\begin{align}
		\label{eq:edge-decomposition}
		\chi & =
		\sum_{\substack{k = - \indepnum{G} \\ k \neq 0}}^{\indepnum{G}}
		\frac{C_k}{u_k} \Psi_k
		+
		C_0 \Psi_{0}\,.
	\end{align}
	The last term is omitted when $C_0=0$.
\end{cor}

\begin{cor}\label{cor:edge-autocorrelation}
	Let $G$ be a connected ECF graph, let $K_s$ be a simplicial clique, and let $\chi$ be the associated edge operator.
	Let $\langle A\rangle_{\infty} \equiv \operatorname{Tr}(A)/\operatorname{Tr}(\mathbb{1})$ denote the normalized trace in a fixed faithful finite-dimensional representation of the graph Clifford algebra, or in the physical spin representation when one is specified, and let $\chi(t) \equiv \mathrm{e}^{\imi H_G t}\chi \mathrm{e}^{-\imi H_G t}$.
	Then
	\begin{align}
		\label{eq:edge-autocorrelation}
		\langle \chi(t)\chi\rangle_{\infty}
		=
		C_0^2
		+
		\sum_{k=1}^{\indepnum{G}}
		\pqty{\frac{C_k}{u_k}}^2
		\cos\pqty{\frac{2t}{u_k}}.
	\end{align}
\end{cor}

\begin{proof}
	The eigenoperator equation gives $\Psi_k(t)=\mathrm{e}^{2\imi t/u_k}\Psi_k$ for $k\neq0$, while the zero-mode term, when present, is time independent.
	Using Corollary~\ref{cor:edge-decomposition}, this gives
	\begin{align}
		\chi(t)
		=
		\sum_{\substack{k = - \indepnum{G} \\ k \neq 0}}^{\indepnum{G}}
		\frac{C_k}{u_k}
		\mathrm{e}^{2\imi t/u_k}
		\Psi_k
		+
		C_0 \Psi_0.
	\end{align}
	For $k,l\neq0$, cyclicity of the normalized trace and $\{\Psi_k,\Psi_l\}=\delta_{k+l,0}$ imply
	\begin{align}
		\langle \Psi_k\Psi_l\rangle_{\infty}
		=
		\frac{1}{2}\delta_{k+l,0}.
	\end{align}
	The zero-mode mixed terms vanish by the same argument, and $\langle\Psi_0^2\rangle_{\infty}=1$ when the zero mode is present.
	Substitution therefore yields
	\begin{align}
		\langle \chi(t)\chi\rangle_{\infty}
		=
		C_0^2
		+
		\frac{1}{2}
		\sum_{\substack{k = - \indepnum{G} \\ k \neq 0}}^{\indepnum{G}}
		\pqty{\frac{C_k}{u_k}}^2
		\mathrm{e}^{2\imi t/u_k}.
	\end{align}
	Since $u_{-k}=-u_k$ and $C_{-k}=-C_k$, the terms $k$ and $-k$ combine into Eq.~\eqref{eq:edge-autocorrelation}.
\end{proof}

\subsection{Proof of Theorem~\ref{thm:path-product-expansion}}
\label{sec:proof-path-product-expansion}

\begin{proof}[Proof of Theorem~\ref{thm:path-product-expansion}]
	From Eq.~\eqref{eq:Krylov-diff-eq},
	\begin{align}
		\frac{u}{2}\qty[H_G, P_G(u^2)\Phi_G(u)] = P_G(u^2)\Phi_G(u) - P_G(u^2)\chi.
	\end{align}
	Taking the limit $u \to u_k$ and using Corollary~\ref{cor:Psi-Phi-relation}, we obtain
	\begin{align}
		[H_G, \Psi_k]
		=
		\frac{2}{u_k}\Psi_k
		=
		2\epsilon_k \Psi_k\,.
	\end{align}
	Since $u_{-k} = -u_k$, this is Eq.~\eqref{eq:eigenoperator} for $\Psi_{\pm k}$.

	For the anticommutation relations, Corollary~\ref{cor:Psi-Phi-relation} and Theorem~\ref{thm:Phi-anti-commu} give
	\begin{align}
		\qty{\Psi_k, \Psi_l}
		 & =
		\frac{2}{\NN_k \NN_l}
		\lim_{\substack{u \to u_k \\ v \to u_l}}
		\frac{
			u P_{G \setminus K_s}(u^2) P_G(v^2)
			+
			v P_{G \setminus K_s}(v^2) P_G(u^2)
		}{u+v}.
	\end{align}
	If $u_k + u_l \neq 0$, then the denominator is regular in the limit, while each term in the numerator contains a factor that vanishes at the corresponding root of $P_G(u^2)$.
	Hence $\qty{\Psi_k, \Psi_l} = 0$ for $k + l \neq 0$.

	For $l = -k$, taking the limit $v \to -u$ in Eq.~\eqref{eq:Phi-anti-commu} gives
	\begin{align}
		\qty{\Phi_G(u), \Phi_G(-u)}
		=
		2\dv{u}\qty[u \frac{P_{G \setminus K_s}(u^2)}{P_G(u^2)}]\,.
	\end{align}
	Hence
	\begin{align}
		\qty{\Psi_k, \Psi_{-k}}
		 & =
		\frac{2}{\NN_k \NN_{-k}}
		\lim_{u \to u_k}
		P_G(u^2)^2
		\dv{u}\qty[u \frac{P_{G \setminus K_s}(u^2)}{P_G(u^2)}]
		=
		-\frac{4u_k^2 P_{G \setminus K_s}(u_k^2) P'_G(u_k^2)}{\NN_k \NN_{-k}}
		=
		1\,.
	\end{align}
	Here we used $P_G(u_k^2)=0$ and $\dv{u} P_G(u^2)\eval_{u=u_k} = 2u_k P'_G(u_k^2)$.
	Combining the two cases,
	\begin{align}
		\{\Psi_k,\Psi_l\}=\delta_{k+l,0}.
	\end{align}
\end{proof}

\subsection{Linear relations among Krylov basis elements}
\label{subsec:krylov-linear-dependence}

The path-product expansion gives a recurrence that expresses all sufficiently high $\phi_j$ in terms of lower ones.

Let $A_G$ be the linear operator on operator space defined by
\begin{align}
	A_G(X)
	 & \equiv
	\frac{1}{2}\qty[H_G,X].
\end{align}
Then the Krylov basis elements in Eq.~\eqref{eq:Krylov-def} are written as
\begin{align}
	\phi_j
	=
	A_G^j\chi
	\,.
\end{align}
The Krylov equation~\eqref{eq:Krylov-diff-eq} is equivalently
\begin{align}
	\label{eq:Krylov-resolvent-equation}
	(1-uA_G)\Phi_G(u)
	 & =
	\chi,
	 &
	\Phi_G(u)
	 & =
	(1-uA_G)^{-1}\chi
\end{align}
as a formal power series.

By Theorem~\ref{thm:Krylov-path-expansion}, $P_G(u^2)\Phi_G(u)=\Xi_G(u)$, which by the degree estimate in the proof of Theorem~\ref{thm:Phi-Psi-decomposition} is a polynomial in $u$ of degree at most $2\alpha_G$.
We introduce the reciprocal polynomial of $P_G$
\begin{align}
	\widehat P_G(z)
	 & \equiv
	z^{\alpha_G}P_G(z^{-1})
	=
	\sum_{r=0}^{\alpha_G}\indepcoeff{G}{r}z^{\alpha_G-r}.
\end{align}
Substituting $z=A_G^2$ gives the operator polynomial
\begin{align}
	\label{eq:Krylov-formal-independence-polynomial}
	\widehat P_G(A_G^2)\chi
	=
	\sum_{r=0}^{\alpha_G}\indepcoeff{G}{r}\,\phi_{2\alpha_G-2r},
\end{align}
which involves only non-negative powers of $A_G$.
In terms of the original independence polynomial, the same operator may be written formally as
\begin{align}
	\widehat P_G(A_G^2)\chi
	=
	A_G^{2\alpha_G}P_G(A_G^{-2})\chi,
\end{align}
where the right-hand side is understood through Eq.~\eqref{eq:Krylov-formal-independence-polynomial}, without assuming that $A_G$ is invertible.
Taking the coefficient of $u^{m+2\alpha_G}$ gives, for $m\geq0$,
\begin{align}
	\label{eq:Krylov-coefficient-identity}
	[u^{m+2\alpha_G}]\Xi_G(u)
	=
	[u^{m+2\alpha_G}]P_G(u^2)\Phi_G(u)
	=
	A_G^m\widehat P_G(A_G^2)\chi.
\end{align}
For $m\geq1$ the left-hand side vanishes because $\Xi_G(u)$ has degree at most $2\alpha_G$, and therefore
\begin{align}
	\label{eq:Krylov-polynomial-kernel}
	A_G^m\widehat P_G(A_G^2)\chi
	=
	0,
	\qquad
	m\geq1.
\end{align}
In expanded form, this relation is
\begin{align}
	\label{eq:Krylov-recurrence-independence-polynomial}
	\sum_{r=0}^{\alpha_G}\indepcoeff{G}{r}\,\phi_{m+2\alpha_G-2r}
	=
	0,
	\qquad
	m\geq1.
\end{align}
Since $\indepcoeff{G}{0}=1$, this recurrence expresses each sufficiently high $\phi_j$ as a linear combination of earlier Krylov elements of the same parity.
Separating the two parities, Eq.~\eqref{eq:Krylov-recurrence-independence-polynomial} gives
\begin{align}
	\label{eq:odd-krylov-recurrence}
	\phi_{2j+1}
	 & =
	-\sum_{r=1}^{\alpha_G}\indepcoeff{G}{r}\,\phi_{2j+1-2r},
	\qquad
	j\geq\alpha_G,
	\\
	\label{eq:even-krylov-recurrence}
	\phi_{2j}
	 & =
	-\sum_{r=1}^{\alpha_G}\indepcoeff{G}{r}\,\phi_{2j-2r},
	\qquad
	j\geq\alpha_G+1.
\end{align}
In particular,
\begin{align}
	\operatorname{span}\{\phi_{2j+1}:j\geq0\}
	 & =
	\operatorname{span}\{\phi_1,\phi_3,\ldots,\phi_{2\alpha_G-1}\},
	\nonumber \\
	\operatorname{span}\{\phi_{2j}:j\geq1\}
	 & =
	\operatorname{span}\{\phi_2,\phi_4,\ldots,\phi_{2\alpha_G}\}.
\end{align}
Together with $\phi_0=\chi$, this gives
\begin{align}
	\operatorname{span}\{\phi_j:j\geq0\}
	\subseteq
	\operatorname{span}\{\phi_0,\phi_1,\ldots,\phi_{2\alpha_G}\}.
\end{align}

The coefficient with $m=0$ must be treated separately.
Using Eq.~\eqref{eq:Phi-Psi-decomposition} and $\Xi_G(u)=P_G(u^2)\Phi_G(u)$, we have
\begin{align}
	\Xi_G(u)
	=
	\sum_{\substack{k=-\alpha_G \\ k\neq0}}^{\alpha_G}
	C_k\frac{P_G(u^2)}{u_k-u}\Psi_k
	+
	C_0P_G(u^2)\Psi_0.
\end{align}
For $k\neq0$, the quotient $P_G(u^2)/(u_k-u)$ has degree at most $2\alpha_G-1$.
Therefore the coefficient of $u^{2\alpha_G}$ comes only from the zero-mode term in Eq.~\eqref{eq:zero-mode-limit}:
\begin{align}
	\label{eq:zero-mode-krylov-coefficient}
	\widehat P_G(A_G^2)\chi
	=
	\sum_{r=0}^{\alpha_G}\indepcoeff{G}{r}\,\phi_{2\alpha_G-2r}
	=
	\indepcoeff{G}{\alpha_G}C_0\Psi_0.
\end{align}
Since the zero-mode term is annihilated by $A_G$, applying $A_G^\ell$ with $\ell\geq1$ to Eq.~\eqref{eq:zero-mode-krylov-coefficient} reproduces Eq.~\eqref{eq:Krylov-recurrence-independence-polynomial} with $m=\ell$.
For $C_0\neq0$, the zero mode is
\begin{align}
	\label{eq:zero-mode-independence-polynomial-form}
	\Psi_0
	=
	\frac{1}{C_0\indepcoeff{G}{\alpha_G}}\widehat P_G(A_G^2)\chi
	=
	\frac{A_G^{2\alpha_G}}{C_0\indepcoeff{G}{\alpha_G}}P_G(A_G^{-2})\chi\,,
\end{align}
where the last expression is understood in the sense of Eq.~\eqref{eq:Krylov-formal-independence-polynomial}.
If $C_0=0$, the right-hand side of Eq.~\eqref{eq:zero-mode-krylov-coefficient} vanishes, and therefore
\begin{align}
	\operatorname{span}\{\phi_{2j}:j\geq0\}
	=
	\operatorname{span}\{\phi_0,\phi_2,\ldots,\phi_{2\alpha_G-2}\}.
\end{align}
If $C_0\neq0$, $\phi_{2\alpha_G}$ is not eliminated by this argument.

We next check when these finite spanning sets are actually bases.
Comparing coefficients of powers of $u$ in Eq.~\eqref{eq:Phi-Psi-decomposition} gives
\begin{align}
	\phi_j
	=
	\sum_{\substack{k=-\alpha_G \\ k\neq0}}^{\alpha_G}
	C_k u_k^{-j-1}\Psi_k
	+
	\delta_{j0}C_0\Psi_0.
\end{align}
For $k>0$, set $E_k=\Psi_k+\Psi_{-k}$ and $O_k=\Psi_k-\Psi_{-k}$.
Since $u_{-k}=-u_k$ and $C_{-k}=-C_k$, this becomes
\begin{align}
	\phi_{2q}
	 & =
	\sum_{k=1}^{\alpha_G}
	C_k u_k^{-2q-1}E_k
	+
	\delta_{q0}C_0\Psi_0,
	\nonumber \\
	\phi_{2q+1}
	 & =
	\sum_{k=1}^{\alpha_G}
	C_k u_k^{-2q-2}O_k.
\end{align}
Thus, when the positive numbers $u_k^2$ are distinct and all residues $C_k$ are nonzero, the odd sector and the nonzero-mode part of the even sector are ordinary Vandermonde systems in $u_k^{-2}$.
Under this nondegeneracy assumption, $\phi_1,\phi_3,\ldots,\phi_{2\alpha_G-1}$ are linearly independent.
If also $C_0=0$, then $\phi_0,\phi_2,\ldots,\phi_{2\alpha_G-2}$ are linearly independent and $\phi_{2\alpha_G}$ is determined by the preceding even-sector relation.
If instead $C_0\neq0$, then $\phi_0,\phi_2,\ldots,\phi_{2\alpha_G}$ are linearly independent, because only $\phi_0$ has a $\Psi_0$ component and the remaining even vectors form a Vandermonde system on the nonzero modes.

\section{Relation to the transfer-matrix construction}
\label{sec:transfer-matrix-formalism}

In this section we connect the Krylov generating function $\Phi_G(u)$ constructed in Section~\ref{sec:proof} to the transfer-matrix construction of Refs.~\cite{fendley-fermions-in-disguise,chapman-jw,fermions-behind-the-disguise,unified-graph-th,sajat-FP-model}, thereby identifying the path-product modes $\Psi_k$ with those of the transfer-matrix construction.

Throughout this section we assume that $G$ is a connected ECF frustration graph and that the same simplicial clique $K_s$ and auxiliary edge operator $\chi$ as in Section~\ref{sec:proof} have been fixed.
We define the \emph{transfer matrix} as
\begin{equation}\label{Tudef}
	T_G(u) \equiv \sum_{S \in \indepset{G}} (-u)^{|S|} \prod_{\jbm \in S} \hc_{\jbm}\,,
\end{equation}
where $\indepset{G}$ is the set of independent sets of $G$ (Section~\ref{subsec:indep}).
Grouping the terms by the size of the independent set gives the generating function for the nonlocal conserved charges:
\begin{align}
	T_G(u)
	=
	\sum_{k=0}^{\indepnum{G}}
	(-u)^k \nonlocalcharge{G}{k},
	\label{eq:transfer-matrix-nonlocal-charge-expansion}
\end{align}
where
\begin{align}
	\nonlocalcharge{G}{k}
	\equiv
	\sum_{S\in\indepset{G}^{k}}
	\prod_{\jbm\in S}\hc_{\jbm}.
	\label{eq:nonlocal-charge-independent-set}
\end{align}
The product is order-independent because the vertices in $S$ are pairwise non-adjacent.
These charges are nonlocal because their monomials may be supported on disconnected, widely separated vertices of the frustration graph.
For a claw-free graph, these nonlocal conserved charges form a commuting family~\cite{fendley-fermions-in-disguise,fermions-behind-the-disguise,unified-graph-th}:
\begin{align}
	\qty[\nonlocalcharge{G}{k}, \nonlocalcharge{G}{l}] = 0\,.
\end{align}
Consequently, the transfer matrices commute for arbitrary spectral parameters,
\begin{align}
	\qty[T_G(u), T_G(v)] = 0\,.
\end{align}
The transfer matrix also satisfies the inversion relation
\begin{equation}
	\label{eq:inversion}
	T_G(u)T_G(-u)=P_G(u^2)\,.
\end{equation}

\begin{thm}
	\label{thm:Krylov-generating-function}
	The generating function for the Krylov basis~\eqref{eq:Krylov-def} is written in terms of the transfer matrix as
	\begin{align}
		\label{eq:Krylov-generating-function}
		\Phi_G(u)
		 & =
		\frac{1}{2}
		\pqty{ \chi + \frac{T_G(-u) \chi T_G(u)}{P_G(u^2)} }
		=
		T_G(u)^{-1} T_{G \setminus K_s}(u) \chi
		=
		\chi T_{G \setminus K_s}(-u) T_G(-u)^{-1}
		.
	\end{align}
\end{thm}
\begin{proof}
	References~\cite{fendley-fermions-in-disguise,fermions-behind-the-disguise} give the identity
	\begin{align}
		\label{eq:T-chi-commutator}
		\qty[\chi, T_G(u)] = \frac{u}{2} \qty{T_G(u), \qty[H_G, \chi]}\,.
	\end{align}
	Multiplying from the left by $T_G(-u)$ and using Eq.~\eqref{eq:inversion}, we obtain
	\begin{align}
		\label{eq:T-chi-expanded}
		 & T_G(-u) \qty[\chi, T_G(u)]
		=
		\frac{u}{2} \pqty{T_G(-u) \qty[H_G, \chi] T_G(u) + \qty[H_G, \chi] P_G(u^2)}\,.
	\end{align}
	Rearranging this expression yields
	\begin{align}
		\label{eq:key-identity}
		 & \frac{T_G(-u) \chi T_G(u)}{P_G(u^2)} - \chi
		= \frac{u}{2} \pqty{ \frac{T_G(-u) \qty[H_G, \chi] T_G(u)}{P_G(u^2)} + \qty[H_G, \chi] }\,.
	\end{align}

	Let $\widetilde{\Phi}_G(u)\equiv\frac{1}{2} \pqty{ \chi + \frac{T_G(-u) \chi T_G(u)}{P_G(u^2)} }$ denote the first expression in Eq.~\eqref{eq:Krylov-generating-function}.
	We verify that $\frac{u}{2}\qty[H_G,\widetilde{\Phi}_G(u)]=\widetilde{\Phi}_G(u)-\chi$, i.e.\ Eq.~\eqref{eq:Krylov-diff-eq} with $\Phi_G$ replaced by $\widetilde{\Phi}_G$.
	Since $[H_G, T_G(\pm u)] = 0$, the left-hand side becomes
	\begin{align}
		\frac{u}{2} \qty[H_G, \widetilde{\Phi}_G(u)]
		 & =
		\frac{u}{4} \pqty{\qty[H_G, \chi] + \frac{T_G(-u) \qty[H_G, \chi] T_G(u)}{P_G(u^2)}}
		\nonumber \\
		 & =
		\frac{1}{2}
		\pqty{\frac{T_G(-u) \chi T_G(u)}{P_G(u^2)} - \chi}
		\nonumber \\
		 & =
		\widetilde{\Phi}_G(u) - \chi\,,
	\end{align}
	where we used Eq.~\eqref{eq:key-identity} in the second equality.
	Since also $\widetilde{\Phi}_G(0)=\chi$, Lemma~\ref{lem:Krylov-uniqueness} gives $\widetilde{\Phi}_G(u)=\Phi_G(u)$.
	The formula also has the equivalent forms
	\begin{align}
		\Phi_G(u)
		 & =
		\widetilde{\Phi}_G(u)
		\nonumber \\
		 & =
		\frac{1}{2}
		\pqty{\frac{T_G(-u) \chi T_G(u)}{P_G(u^2)} + \chi}
		\nonumber \\
		 & =
		\frac{T_G(-u)}{2P_G(u^2)}
		\qty{\chi, T_G(u)}
		\nonumber \\
		 & =
		\frac{T_G(-u)}{P_G(u^2)}
		T_{G\setminus K_s}(u)\chi
		\nonumber \\
		 & =
		T_G(u)^{-1}
		T_{G\setminus K_s}(u)\chi
		\nonumber \\
		 & =
		\chi T_{G\setminus K_s}(-u) T_G(-u)^{-1}
		\,,
	\end{align}
	where we used $\qty{\chi, T_G(u)} = 2 T_{G\setminus K_s}(u)\chi$ and $\qty{\chi, T_G(-u)} = 2\chi T_{G\setminus K_s}(-u)$, which hold because an independent set contains at most one vertex of the clique $K_s$, so the anticommutator doubles the terms avoiding $K_s$.
\end{proof}

\begin{cor}[Transfer-matrix form of the fermion modes~\cite{fendley-fermions-in-disguise,fermions-behind-the-disguise,unified-graph-th}]\label{cor:Psi-transfer-matrix-form}
	For each $k=\pm1,\ldots,\pm\indepnum{G}$,
	\begin{align}
		\label{eq:Psi-transfer-matrix-form}
		\Psi_k
		=
		\frac{1}{2\NN_k}
		T_G(-u_k)\chi T_G(u_k)\,.
	\end{align}
	The normalization used in those references is twice the present one.
\end{cor}
\begin{proof}
	Multiplying Eq.~\eqref{eq:Krylov-generating-function} by $P_G(u^2)$ and taking $u\to u_k$ gives
	\begin{align}
		\lim_{u\to u_k}P_G(u^2)\Phi_G(u)
		=
		\frac{1}{2}T_G(-u_k)\chi T_G(u_k)\,,
	\end{align}
	because the term $P_G(u^2)\chi$ vanishes at the root.
	Combining this identity with Corollary~\ref{cor:Psi-Phi-relation} gives Eq.~\eqref{eq:Psi-transfer-matrix-form}.
\end{proof}

The following exchange relation between $\Phi_G(u)$ and $T_G(v)$ specializes at $u=v$ to the transfer-matrix formula of Theorem~\ref{thm:Krylov-generating-function}.
\begin{lem}\label{lem:Phi-T-relation}
	The generating function and transfer matrix satisfy
	\begin{align}
		(u - v) \Phi_G(u) T_G(v)
		 & =
		(u + v)  T_G(v) \Phi_G(u)
		-
		v \pqty{ \chi T_G(v) + T_G(v) \chi}\,.
		\label{eq:Phi-T-relation}
	\end{align}
\end{lem}
\begin{proof}
	From Eqs.~\eqref{eq:T-chi-commutator} and~\eqref{eq:inversion}, we have
	\begin{align}\label{eq:T-chi-commutator-expand}
		T_G(-u) \chi T_G(u)
		=
		P_G(u^2) \chi
		+
		\frac{u}{2} \pqty{ T_G(-u) \qty[H_G, \chi] T_G(u) + P_G(u^2) \qty[H_G, \chi] }\,.
	\end{align}
	From Eq.~\eqref{eq:Krylov-generating-function}, we also have
	\begin{align}\label{eq:T-chi-Phi}
		T_G(-u) \chi T_G(u)
		=
		P_G(u^2) \pqty{2\Phi_G(u) - \chi}\,.
	\end{align}

	Substituting Eq.~\eqref{eq:Krylov-generating-function} into $T_G(-v) \Phi_G(u) T_G(v)$ gives
	\begin{align}
		T_G(-v) \Phi_G(u) T_G(v)
		 & =
		\frac{1}{2} T_G(-v) \chi T_G(v)
		+
		\frac{1}{2P_G(u^2)} T_G(-u) T_G(-v) \chi T_G(v) T_G(u)\,.
	\end{align}
	Applying Eq.~\eqref{eq:T-chi-commutator-expand} to the second term, we obtain
	\begin{align}
		T_G(-v) \Phi_G(u) T_G(v)
		 & =
		\frac{1}{2} T_G(-v) \chi T_G(v)
		+
		\frac{P_G(v^2)}{2P_G(u^2)} T_G(-u) \chi T_G(u)
		\nonumber \\
		 & \quad
		+
		\frac{v}{4P_G(u^2)} T_G(-u) \pqty{ T_G(-v) [H_G, \chi] T_G(v) + P_G(v^2) [H_G, \chi] } T_G(u)\,.
	\end{align}
	Using Eq.~\eqref{eq:T-chi-Phi} and noting that $[H_G, T_G(\pm u)] = 0$, this becomes
	\begin{align}
		T_G(-v) \Phi_G(u) T_G(v)
		 & =
		P_G(v^2) \pqty{\Phi_G(v) - \frac{1}{2}\chi}
		+
		P_G(v^2) \pqty{\Phi_G(u) - \frac{1}{2}\chi}
		\nonumber \\
		 & \quad
		+
		\frac{v}{2} T_G(-v) [H_G, \Phi_G(u) - \frac{1}{2}\chi] T_G(v)
		+
		\frac{P_G(v^2) v}{2} [H_G, \Phi_G(u) - \frac{1}{2}\chi]\,.
	\end{align}
	Using Eq.~\eqref{eq:Krylov-diff-eq} to replace $\frac{u}{2}[H_G, \Phi_G(u)] = \Phi_G(u) - \chi$ and $\frac{v}{2}[H_G, \Phi_G(v)] = \Phi_G(v) - \chi$, we find
	\begin{align}
		T_G(-v) \Phi_G(u) T_G(v)
		 & =
		P_G(v^2) \pqty{\Phi_G(v) + \Phi_G(u) - \chi}
		+
		\frac{v}{u} T_G(-v) \pqty{\Phi_G(u) - \chi} T_G(v)
		\nonumber \\
		 & \qquad
		+
		\frac{P_G(v^2) v}{u} \pqty{\Phi_G(u) - \chi}
		-
		P_G(v^2) \pqty{\Phi_G(v) - \chi}
		\nonumber \\
		 & =
		P_G(v^2) \Phi_G(u)
		+
		\frac{v}{u} T_G(-v) \Phi_G(u) T_G(v)
		-
		\frac{v}{u} T_G(-v) \chi T_G(v)
		\nonumber \\
		 & \qquad
		+
		\frac{P_G(v^2) v}{u} \Phi_G(u)
		-
		\frac{P_G(v^2) v}{u} \chi\,.
	\end{align}
	Collecting the terms involving $T_G(-v) \Phi_G(u) T_G(v)$ on the left-hand side and multiplying by $u$, we obtain
	\begin{align}
		 & (u - v) T_G(-v) \Phi_G(u) T_G(v)
		=
		P_G(v^2) (u + v) \Phi_G(u)
		-
		v T_G(-v) \chi T_G(v)
		-
		P_G(v^2) v \chi\,.
	\end{align}
	Multiplying from the left by $T_G(v)$ and using Eq.~\eqref{eq:inversion}, we arrive at
	\begin{align}
		 & (u - v) \Phi_G(u) T_G(v)
		=
		(u + v) T_G(v) \Phi_G(u)
		-
		v \pqty{ \chi T_G(v) + T_G(v) \chi }\,.
	\end{align}
\end{proof}

\section{Conserved charges}
\label{sec:conserved-charges}

The constructions below require neither a simplicial clique nor an edge operator, and they are not tied to the even-hole-free setting: they are formulated for arbitrary claw-free frustration graphs, including those that do not admit a hidden free-fermion solution.
The adjective ``local'' in this section is used in the graph-Clifford sense stated in Section~\ref{sec:setup}.

Earlier transfer-matrix constructions produced the nonlocal independent-set charges $\nonlocalcharge{G}{k}$ in Eq.~\eqref{eq:nonlocal-charge-independent-set}~\cite{fendley-fermions-in-disguise,fermions-behind-the-disguise,unified-graph-th}.

Here we first construct explicit local conserved charges and prove their conservation using the path-product commutator identities.
We then introduce a unified family of generalized conserved charges that contains both the known nonlocal conserved charges and the new local conserved charges as special cases.
Finally, we discuss oriented even path operators, whose conservation requires an induced-path orientation and an additional balance condition on the residual graphs at each vertex.

\subsection{Local conserved charges}
\label{sec:local-charges}

\begin{define}[Odd local charge]
	We define the odd local charge $\localcharge{G}{2k+1}$ by
	\begin{align}
		\label{eq:local-charge}
		\localcharge{G}{2k+1}
		\equiv
		\sum_{n = 0}^{k}
		\sum_{\pa \in \pathsetlength{G}{2k+1-2n}}
		\indepcoeff{\complementset{\pa}}{n}
		\pathprodH{\pa} \,,
	\end{align}
	where $\pathsetlength{G}{l}$ is defined in Section~\ref{sec:setup} and the residual coefficients $\indepcoeff{G}{k}$ in Eq.~\eqref{eq:indep-coeff}.
\end{define}

Every path-product summand in $\localcharge{G}{2k+1}$ has path size at most $2k+1$, so the charge is local in the graph-Clifford sense.
The lowest-order case recovers the Hamiltonian: $H_G = \localcharge{G}{1}$.

The obstruction to the non-path cancellation in the local conserved-charge proof is the following.

\begin{define}[Even bubble wand]
	\label{def:even-bubble-wand}
	An \emph{even bubble wand} consists of an induced even hole $\jbm_1\md\cdots\md\jbm_{2m}\md\jbm_1$ and an induced path with at least two vertices.
	We write the last edge of the path as $\ellbm\md\ellbm'$, where $\ellbm'$ is an endpoint of the path.
	The endpoint $\ellbm'$ is adjacent to exactly the two consecutive hole vertices $\jbm_1$ and $\jbm_{2m}$.
	The predecessor $\ellbm$, as well as every earlier path vertex, is non-adjacent to every vertex of the hole (see Figure~\ref{fig:even-bubble-wand}).
	Its \emph{size} is $2m$.
\end{define}
A graph is \emph{even-bubble-wand-free} if it contains no induced even bubble wand.

\begin{figure}[t]
	\centering
	\begin{tikzpicture}[
			vertex/.style={circle, draw, thick, fill=black, minimum size=\vertexsize, inner sep=0pt},
			scale=0.95, line cap=round
		]
		\node[vertex] (ell) at (1.0,0) {};
		\node[vertex] (ellp) at (2.2,0) {};

		\draw[thick,dotted] (-0.1,0)--(0.5,0);
		\draw[thick] (0.5,0)--(ell);
		\draw[thick] (ell)--(ellp);

		\node[vertex] (j1) at (3.41,0.54) {};
		\node[vertex] (j2) at (4.16,1.29) {};
		\node[vertex] (ju) at (5.24,1.29) {};
		\node[vertex] (jd) at (5.24,-1.29) {};
		\node[vertex] (j2m1) at (4.16,-1.29) {};
		\node[vertex] (j2m) at (3.41,-0.54) {};

		\draw[thick] (j1)--(j2)--(ju)--(5.60,1.07);
		\draw[thick,dotted] (4.70,0) ++(310:1.4) arc (310:410:1.4);
		\draw[thick] (5.60,-1.07)--(jd)--(j2m1)--(j2m)--(j1);

		\draw[thick] (ellp)--(j1);
		\draw[thick] (ellp)--(j2m);

		\node[below=2pt] at (ell) {\small $\ellbm$};
		\node[below=2pt] at (ellp) {\small $\ellbm'$};
		\node[above left=1pt] at (j1) {\small $\jbm_1$};
		\node[above=2pt] at (j2) {\small $\jbm_2$};
		\node[below left=1pt] at (j2m1) {\small $\jbm_{2m-1}$};
		\node[below left=1pt] at (j2m) {\small $\jbm_{2m}$};
		\node at (0.3,0.40) {\small path};
		\node at (4.70,0.10) {\small even hole};
	\end{tikzpicture}
	\caption{An even bubble wand, with only the terminal path edge $\ellbm\md\ellbm'$ shown.
		The endpoint $\ellbm'$ is adjacent to the consecutive vertices $\jbm_1$ and $\jbm_{2m}$ of the induced even hole.}
	\label{fig:even-bubble-wand}
\end{figure}

\begin{lem}\label{lem:endpoint-loop-ebw}
	Let $G$ be a graph and let $\rho=\ell_1\md\cdots\md\ell_m$ be an induced path with $m$ odd.
	Let $\jbm\notin V(\rho)$, and suppose that the only vertices of $\rho$ adjacent to $\jbm$ are $\ell_{q-1},\ell_q,\ell_m$ for some $q$ with $2\le q<m-1$.
	Then $C=\ell_q\md\ell_{q+1}\md\cdots\md\ell_m\md\jbm\md\ell_q$ is an induced cycle.
	If $C$ is even, then $C$ and the remaining path segment $\ell_1\md\cdots\md\ell_{q-1}$ form an even bubble wand.
	In particular, the hole size is $|C|=m-q+2\le m-1$.
\end{lem}
\begin{proof}
	$C$ is a cycle since $\jbm$ is adjacent to both endpoints of the segment $\ell_q\md\cdots\md\ell_m$ of $\rho$, and it is induced since $\rho$ is induced and $\jbm$ has no path-neighbor other than $\ell_{q-1},\ell_q,\ell_m$.
	Assume that $C$ is even.
	Because $m$ is odd, $q=2$ would give the odd value $|C|=m$; hence $q\ge3$ and the segment $\ell_1\md\cdots\md\ell_{q-1}$ has at least two vertices.
	Its last vertex $\ell_{q-1}$ is adjacent to exactly the two consecutive vertices $\ell_q$ and $\jbm$ of $C$.
	The preceding vertices $\ell_1,\ldots,\ell_{q-2}$ are non-adjacent to $C$, by inducedness of $\rho$ and by the assumption on the path-neighbors of $\jbm$.
	Finally $|C|=m-q+2\le m-1$, since $q\ge3$.
\end{proof}

\begin{lem}\label{lem:odd-local-nonpath-cancellation}
	Let $G$ be claw-free, and let $2K$ be the size of its smallest even bubble wand, with $K=\infty$ if none exists.
	Fix $k<K$.
	In the termwise expansion of $\frac{1}{2}[H_G,\localcharge{G}{2k+1}]$ over the paths in Eq.~\eqref{eq:local-charge}, all non-path summands sum to zero.
\end{lem}
\begin{proof}
	In Eq.~\eqref{eq:local-charge} the path $\rho$ has odd size $|\rho|\le 2k+1$ and weight $\indepcoeff{\complementset{\rho}}{(2k+1-|\rho|)/2}$.

	For a singleton path $\rho=(\ell)$, the commutator is
	\begin{equation*}
		\tfrac{1}{2}[H_G,\hc_{\ell}]=\sum_{\jbm\sim\ell}\pathprodH{\jbm\md\ell},
	\end{equation*}
	which contains no non-path term.

	For an odd induced path $\rho$ with $|\rho|\ge 3$, commuting each Hamiltonian term through the path product as in the proof of Lemma~\ref{lem:single-path-commutator} produces, besides path shortening and extension, only the non-path families $\pathsvd_{\rho}$ and $\pathloop_{\rho}$ at the two endpoints.
	The classification of off-path vertices with odd $|T_\rho(\jbm)|$ used there is Lemma~\ref{lem:odd-neighbor-classification}, which needs only claw-freeness.

	Lemma~\ref{lem:path-cancellation-pairings}(i) maps each $(\rho,\jbm)$ with $\jbm\in\pathsvd_{\rho}$ to a partner $(\rho',\ell_k)$ with $|\rho'|=|\rho|$ and $\Gamma[\rho']=\Gamma[\rho]$.
	Since the residual coefficient in Eq.~\eqref{eq:local-charge} depends only on $\complementset{\rho}$ and $|\rho|$, the two paired summands enter with equal weight and opposite Clifford signs, so they cancel.
	For an endpoint-loop summand at the terminal endpoint of $\rho$, with $\jbm\in\pathloop_{\rho}$ and associated cycle $(\ell_k,\ldots,\ell_n,\jbm)$, the loop pairing of Lemma~\ref{lem:path-cancellation-pairings}(ii) applies whenever this cycle is odd, producing a partner with the same size and closed neighborhood.
	Applying the same pairing to the reversed path handles endpoint-loop summands at the initial endpoint.
	The pair cancels for the same reason as above.

	The only uncanceled possibility is an endpoint-loop summand whose associated cycle is even.
	By Lemma~\ref{lem:endpoint-loop-ebw}, such a summand carries an induced even bubble wand of hole size at most $|\rho|-1\le 2k$.
	The assumption $k<K$ excludes this, since $2K$ is the smallest even-bubble-wand size in $G$.

\end{proof}

\begin{thm}
	\label{thm:local-charge-conservation-pbc}
	Let $G$ be a claw-free graph and let $2K$ denote the size of the smallest even bubble wand in $G$ (set $K=\infty$ if none exists).
	Then
	\begin{align}
		\qty[H_G,\localcharge{G}{2k+1}] = 0,
		\qquad 0\le k<K.
	\end{align}
\end{thm}
\begin{proof}
	Fix $k$ with $0\le k<K$.
	Throughout this proof, we use the convention $\indepcoeff{F}{s}=0$ for $s<0$.

	For an induced path $\pa=\ell_1\md\cdots\md\ell_m$ in $G$ with $m\ge 2$, define the endpoint extension cliques
	\begin{align}
		\mathcal{K}_{\pa}^{-}
		 & \equiv
		\Gamma[\ell_1]\setminus\Gamma[\pa_{[2;]}],
		 &
		\mathcal{K}_{\pa}^{+}
		 & \equiv
		\Gamma[\ell_m]\setminus\Gamma[\pa_{[;-1]}].
	\end{align}
	Applying the same factor-by-factor commutator calculation at the two endpoints and using Lemma~\ref{lem:odd-neighbor-classification} gives
	\begin{align}
		\label{eq:local-charge-single-path-comm}
		\frac{1}{2}[H_G,\pathprodH{\pa}]
		 & =
		+b_{\ell_1}^2\,\pathprodH{\pa_{[2;]}}
		-b_{\ell_m}^2\,\pathprodH{\pa_{[;-1]}}
		+\sum_{\jbm\in \mathcal{K}_{\pa}^{-}}\pathprodH{\jbm\md\pa}
		-\sum_{\jbm\in \mathcal{K}_{\pa}^{+}}\pathprodH{\pa\md\jbm}
		+\NN_{\pa},
	\end{align}
	where $\NN_{\pa}$ collects the three-neighbor terms $\pathsvd_{\pa}$ and the endpoint-loop terms $\pathloop_{\pa}^{\pm}$, with $\pathloop_{\pa}^{-}$ defined by reversing $\pa$.
	For a singleton path $(\ell)$, we instead use directly
	\begin{align}
		\label{eq:local-charge-singleton-comm}
		\frac{1}{2}[H_G,\hc_{\ell}]
		=
		\sum_{\jbm\sim\ell}\pathprodH{\jbm\md\ell}.
	\end{align}

	By Lemma~\ref{lem:odd-local-nonpath-cancellation}, summing Eq.~\eqref{eq:local-charge-single-path-comm} over the odd paths in Eq.~\eqref{eq:local-charge} and using Eq.~\eqref{eq:local-charge-singleton-comm} for singleton paths, all non-path terms cancel.

	After the non-path terms have been removed, we write the remaining even-path contribution as
	\begin{align}
		\label{eq:local-charge-even-path-expansion}
		\frac{1}{2}\qty[H_G,\localcharge{G}{2k+1}]
		=
		\sum_{m=1}^{k+1}
		\sum_{\pa\in\pathsetlength{G}{2m}}
		D_{\pa}^{(k)}\,\pathprodH{\pa}.
	\end{align}
	Fix an even induced path $\pa=\ell_1\md\cdots\md\ell_{2m}$ with $1\le m\le k+1$ and set $s=k-m+1$.
	The coefficient $D_{\pa}^{(k)}$ receives four endpoint contributions: shortening from $\pa\md\jbm$ and $\jbm\md\pa$, extension from $\pa_{[;-1]}$ and $\pa_{[2;]}$.
	When $|\pa|=2$, the two extension contributions are read from the singleton formula~\eqref{eq:local-charge-singleton-comm}; in particular, the contribution from $\pa_{[;-1]}=(\ell_1)$ is $\pathprodH{\ell_2\md\ell_1}=-\pathprodH{\pa}$.
	When $m=k+1$ only extensions of the longest odd paths contribute: $s=0$, so the shortening sums below vanish by the convention $\indepcoeff{F}{-1}=0$ and the two extension coefficients cancel, both equal to $\indepcoeff{F}{0}=1$.
	Explicitly,
	\begin{align}
		\label{eq:local-charge-even-coeff}
		D_{\pa}^{(k)}
		 & =
		\indepcoeff{\complementset{\pa_{[2;]}}}{s}
		-
		\indepcoeff{\complementset{\pa_{[;-1]}}}{s}
		+
		\sum_{\jbm\in\mathcal{K}_{\pa}^-}
		b_{\jbm}^2\,
		\indepcoeff{\complementset{\jbm\md\pa}}{s-1}
		-
		\sum_{\jbm\in\mathcal{K}_{\pa}^+}
		b_{\jbm}^2\,
		\indepcoeff{\complementset{\pa\md\jbm}}{s-1}.
	\end{align}

	Applying the independence-polynomial recursion~\eqref{eq:independence-polynomial-recursion} at the left endpoint clique $\mathcal{K}_{\pa}^-$ and taking the coefficient of $x^s$ gives
	\begin{align}
		\indepcoeff{\complementset{\pa_{[2;]}}}{s}
		=
		\indepcoeff{\complementset{\pa}}{s}
		-
		\sum_{\jbm\in\mathcal{K}_{\pa}^-}
		b_{\jbm}^2\,
		\indepcoeff{\complementset{\jbm\md\pa}}{s-1}.
	\end{align}
	The analogous coefficient recursion at the right endpoint clique $\mathcal{K}_{\pa}^+$ yields
	\begin{align}
		\indepcoeff{\complementset{\pa_{[;-1]}}}{s}
		=
		\indepcoeff{\complementset{\pa}}{s}
		-
		\sum_{\jbm\in\mathcal{K}_{\pa}^+}
		b_{\jbm}^2\,
		\indepcoeff{\complementset{\pa\md\jbm}}{s-1}.
	\end{align}
	Subtracting the two identities gives
	\begin{align}
		\indepcoeff{\complementset{\pa_{[2;]}}}{s}
		-
		\indepcoeff{\complementset{\pa_{[;-1]}}}{s}
		 & =
		\sum_{\jbm\in\mathcal{K}_{\pa}^+}
		b_{\jbm}^2\,
		\indepcoeff{\complementset{\pa\md\jbm}}{s-1}
		-
		\sum_{\jbm\in\mathcal{K}_{\pa}^-}
		b_{\jbm}^2\,
		\indepcoeff{\complementset{\jbm\md\pa}}{s-1}.
	\end{align}
	Substituting this into Eq.~\eqref{eq:local-charge-even-coeff} gives $D_{\pa}^{(k)}=0$ for every even induced path $\pa$.
\end{proof}

The restriction $k<K$ concerns only this construction and does not exclude other conserved charges or modified formulas at larger ranges.

Setting $K=\infty$ gives the following.

\begin{cor}
	\label{cor:local-charge-ebwf}
	If $G$ is claw-free and even-bubble-wand-free, then $\qty[H_G,\localcharge{G}{2q+1}]=0$ for all $q\ge 0$.
\end{cor}
The periodic frustration graphs of the Fendley model (i.e., the cycle squares $C_M^2$) are claw-free and even-bubble-wand-free, so Corollary~\ref{cor:local-charge-ebwf} yields an extensive family of local conserved charges.

Every ECF graph is claw-free and even-bubble-wand-free.
\begin{cor}
	\label{cor:local-charge-conservation}
	If $G$ is an ECF graph, then $\qty[H_G,\localcharge{G}{2q+1}]=0$ for all $q\ge 0$.
\end{cor}

Beyond the range allowed by the smallest even bubble wand, the scalar independence polynomial in Eq.~\eqref{eq:indep-coeff} may need to be replaced by a suitable extension.
One possible extension is to incorporate the generalized cycle symmetries appearing in the unified graph-theoretic construction of Ref.~\cite{unified-graph-th}.
We expect that such an enlarged, possibly operator-valued, independence polynomial should produce local conserved charges for general claw-free graphs.

\subsection{Generalized conserved charges}
\label{sec:path-packing-charges}

The generalized charges below unify the nonlocal independent-set charges in Eq.~\eqref{eq:nonlocal-charge-independent-set} and the odd local charges in Eq.~\eqref{eq:local-charge}.
They are built from packings of pairwise non-adjacent odd induced paths.
Their guaranteed range is governed by the same even-bubble-wand obstruction as in Theorem~\ref{thm:local-charge-conservation-pbc}.

\begin{define}[Odd path packing]
	Fix integers $c\ge 1$ and $m\ge c$ with $m\equiv c\pmod 2$.
	Define $\oddpathsetlengthcomp{G}{m}{c}$ by
	\begin{align}
		\oddpathsetlengthcomp{G}{m}{c}
		 & \equiv
		\left\{
		\PP=\pa_1\sqcup\cdots\sqcup\pa_c
		\,\middle|\,
		\begin{array}{l}
			\pa_a\in\pathsetodd{G}{}\quad (1\le a\le c),    \\
			\Gamma[\pa_a]\cap\pa_b=\emptyset\quad (a\ne b), \\
			|\PP|=\sum_{a=1}^{c}|\pa_a|=m
		\end{array}
		\right\}.
		\label{eq:odd-path-packing-set}
	\end{align}
	Since distinct components are therefore non-adjacent, generators in different components commute and the product below is independent of the ordering of the components.
	For such a packing, define
	\begin{align}
		\pathprodH{\PP}
		 & \equiv
		\prod_{a=1}^{c}\pathprodH{\pa_a},
		 &
		\complementset{\PP}
		 & \equiv
		G\setminus\Gamma[\PP],
		\label{eq:path-packing-operator}
	\end{align}
	where $\Gamma[\PP]$ is the union of the closed neighborhoods of all vertices in the packing.
\end{define}

In the packing notation, the nonlocal conserved charges in Eq.~\eqref{eq:nonlocal-charge-independent-set} are the singleton-packing case:
\begin{align}
	\nonlocalcharge{G}{k}
	 & =
	\sum_{\PP\in \oddpathsetlengthcomp{G}{k}{k}}
	\pathprodH{\PP}
	=
	\sum_{S\in\indepset{G}^{k}}
	\prod_{\jbm\in S}\hc_{\jbm}.
	\label{eq:nonlocal-charge-packing-special-case}
\end{align}
Indeed, each component of $\PP\in\oddpathsetlengthcomp{G}{k}{k}$ is a singleton and the packing condition makes the $k$ vertices pairwise non-adjacent, identifying $\oddpathsetlengthcomp{G}{k}{k}$ with $\indepset{G}^{k}$.

\begin{define}[Generalized charge]
	For integers $c\ge 1$ and $m\ge c$ with $m\equiv c\pmod 2$, define
	\begin{align}
		\packingcharge{G}{m}{c}
		\equiv
		\sum_{n=0}^{(m-c)/2}
		\sum_{\PP\in \oddpathsetlengthcomp{G}{m-2n}{c}}
		\indepcoeff{\complementset{\PP}}{n}\,
		\pathprodH{\PP}.
		\label{eq:path-packing-charge}
	\end{align}
\end{define}

\paragraph{Specializations.}
\begin{align}
	\packingcharge{G}{k}{k}
	 & =
	\nonlocalcharge{G}{k},
	 &
	\packingcharge{G}{m}{1}
	 & =
	\localcharge{G}{m}
	\qquad (m\ \mathrm{odd}).
	\label{eq:path-packing-special-cases}
\end{align}
Indeed, when $c=m=k$ only the $n=0$ layer contributes and Eq.~\eqref{eq:nonlocal-charge-packing-special-case} applies, while for $c=1$ with $m$ odd Eq.~\eqref{eq:path-packing-charge} reduces to Eq.~\eqref{eq:local-charge}.

\paragraph{Touch and spectator terminology.}
Say that $\jbm$ \emph{touches} a packed component if $\jbm$ is adjacent to at least one vertex of that component.
The vertices of that component adjacent to $\jbm$ are then said to be \emph{touched} by $\jbm$.
If $\jbm$ itself lies on a non-singleton packed component, it is counted as touching that component through its path-neighbors.
We call a packed component not touched by $\jbm$ a \emph{spectator component}; together the spectator components form the \emph{spectator packing}.

\begin{lem}\label{lem:mixed-touch-involution}
	Let $G$ be claw-free and let $\PP$ be an odd path packing.
	Consider a nonzero commutator term $[\hc_{\jbm},\pathprodH{\PP}]$ in which $\jbm$ touches more than one packed component.
	Then the following hold.
	\begin{enumerate}[label=\textup{(\roman*)},nosep]
		\item The vertex $\jbm$ touches exactly two packed components.
		      It touches two consecutive vertices on one component, the \emph{doubly touched} component, and one endpoint on the other, the \emph{singly touched} component; this is the $(2,1)$ profile.
		\item Write the local configuration as in Fig.~\ref{fig:mixed-touch-cancellation}:
		      \begin{equation*}
			      \PP=\mathcal R\sqcup\pa\sqcup\pa',
			      \qquad
			      \pa=\pa_E\md\pa_O,
			      \qquad
			      \pa_E=\pa'_E\md\boldsymbol{j}'.
		      \end{equation*}
		      Here $\mathcal R$ denotes the spectator packing.
		      The partner summand has Hamiltonian vertex $\boldsymbol{j}'$ and packing
		      \begin{equation*}
			      \PP'
			      =
			      \mathcal R\sqcup\widetilde\pa\sqcup\pa'_E,
			      \qquad
			      \widetilde\pa=\pa'\md\jbm\md\pa_O.
		      \end{equation*}
		      This correspondence is an involution, preserves the number of components, and satisfies
		      \begin{equation*}
			      |\PP|=|\PP'|,
			      \qquad
			      \Gamma[\PP]=\Gamma[\PP'].
		      \end{equation*}
		      In particular $\complementset{\PP}=\complementset{\PP'}$, so the two summands carry the same scalar coefficient $\indepcoeff{\complementset{\PP}}{n}=\indepcoeff{\complementset{\PP'}}{n}$ in each fixed $n$ layer of Eq.~\eqref{eq:path-packing-charge}.
		\item The paired Clifford monomials satisfy
		      \begin{equation*}
			      \pathprodH{\PP}\hc_{\jbm}
			      +
			      \pathprodH{\PP'}\hc_{\boldsymbol{j}'}
			      =
			      0.
		      \end{equation*}
		      Hence the two commutator terms cancel.
	\end{enumerate}
\end{lem}

\begin{figure}[t]
	\centering
	\begin{tikzpicture}[
			scale=0.92,
			line cap=round,
			vertex/.style={circle, draw, thick, fill=white, minimum size=\vertexsize, inner sep=0pt},
			ham/.style={circle, draw, thick, fill=black, minimum size=\vertexsize, inner sep=0pt},
			path/.style={thick},
			twohalo/.style={line width=10pt, draw=Goldenrod!70!white, opacity=0.24},
			onehalo/.style={line width=10pt, draw=RoyalBlue!55!white, opacity=0.22},
			onepath/.style={thick},
			onepathcont/.style={thick, dotted},
			piecebrace/.style={decorate,decoration={brace,amplitude=4pt}},
			piecebracemirror/.style={decorate,decoration={brace,mirror,amplitude=4pt}},
			paneltitle/.style={font=\scriptsize, align=center},
			touchingedge/.style={thick, gray, dashed}
		]
		\def\side{1.08}
		\def\rtthree{0.8660254}
		\def\osolid{0.92}
		\def\ofar{1.34}
		\def\evertex{1.08}
		\def\esolid{1.42}
		\def\efar{1.78}
		\def\obracesep{0.24}
		\def\ebracesep{0.30}
		\def\mbracesep{0.30}
		\def\labelsep{0.34}
		\begin{scope}[xshift=0cm,yshift=0cm]
			\node[paneltitle] at (0.70,2.82) {(a) Packing $\PP=\mathcal{R}\sqcup\pa\sqcup\pa'$\\Hamiltonian vertex $\boldsymbol{j}$};
			\coordinate (elltop) at ({\rtthree*\side},0);
			\coordinate (etop) at (0,{0.5*\side});
			\coordinate (jtop) at (0,{-0.5*\side});
			\coordinate (mtop) at ({-0.5*\evertex},{-0.5*\side-\rtthree*\evertex});
			\coordinate (otlo) at ({\rtthree*\side+\osolid},0);
			\coordinate (otfar) at ({\rtthree*\side+\ofar},0);
			\coordinate (eptop) at ({-0.5*\evertex},{0.5*\side+\rtthree*\evertex});
			\coordinate (epsolidtop) at ({-0.5*\esolid},{0.5*\side+\rtthree*\esolid});
			\coordinate (epfartop) at ({-0.5*\efar},{0.5*\side+\rtthree*\efar});
			\coordinate (msolidtop) at ({-0.5*\esolid},{-0.5*\side-\rtthree*\esolid});
			\coordinate (mfartop) at ({-0.5*\efar},{-0.5*\side-\rtthree*\efar});

			\draw[twohalo] (otfar) -- (elltop) -- (etop) -- (epfartop);
			\draw[onehalo] (mtop) -- (mfartop);
			\draw[path,dotted] (otfar) -- (otlo);
			\draw[path] (otlo) -- (elltop) -- (etop) -- (eptop) -- (epsolidtop);
			\draw[path,dotted] (epsolidtop) -- (epfartop);
			\draw[onepath] (mtop) -- (msolidtop);
			\draw[onepathcont] (msolidtop) -- (mfartop);
			\draw[touchingedge] (jtop) -- (elltop);
			\draw[touchingedge] (jtop) -- (etop);
			\draw[touchingedge] (jtop) -- (mtop);

			\node[ham] (jt) at (jtop) {};
			\node[vertex] (ot1) at (elltop) {};
			\node[vertex] (et) at (etop) {};
			\node[vertex] (ept) at (eptop) {};
			\node[vertex] (mt) at (mtop) {};

			\node[left=2pt] at (jt) {\scriptsize $\boldsymbol{j}$};
			\node[left=2pt] at (et) {\scriptsize $\boldsymbol{j}'$};
			\node[below=2pt] at (ot1) {\scriptsize $\ellbm$};
			\node[right=2pt] at (mt) {\scriptsize $\boldsymbol{m}$};
			\coordinate (obstarttop) at ($(elltop)!-0.12!(otfar)+(0,\obracesep)$);
			\coordinate (obendtop) at ($(elltop)!1.08!(otfar)+(0,\obracesep)$);
			\coordinate (ebstarttop) at ($(etop)!-0.12!(epfartop)+({\ebracesep*\rtthree},{0.5*\ebracesep})$);
			\coordinate (ebendtop) at ($(etop)!1.08!(epfartop)+({\ebracesep*\rtthree},{0.5*\ebracesep})$);
			\coordinate (mbstarttop) at ($(mtop)!-0.3!(mfartop)+({-\mbracesep*\rtthree},{0.5*\mbracesep})$);
			\coordinate (mbendtop) at ($(mtop)!1.3!(mfartop)+({-\mbracesep*\rtthree},{0.5*\mbracesep})$);
			\draw[piecebrace] (obstarttop) -- (obendtop);
			\node at ($(obstarttop)!0.55!(obendtop)+(0,\labelsep)$) {\scriptsize $\pa_O$};
			\draw[piecebracemirror] (ebstarttop) -- (ebendtop);
			\node at ($(ebstarttop)!0.55!(ebendtop)+({1.25*\labelsep*\rtthree},{0.45*\labelsep})$) {\scriptsize $\pa_E$};
			\draw[piecebracemirror] (mbstarttop) -- (mbendtop);
			\node at ($(mbstarttop)!0.55!(mbendtop)+({-1.25*\labelsep*\rtthree},{0.75*\labelsep})$) {\scriptsize $\pa'$};
			\node[anchor=west] at ($(etop)!0.45!(otfar)+(0.,1.10)$) {\large $\pa=\pa_E\md\pa_O$};
		\end{scope}

		\begin{scope}[xshift=5.45cm,yshift=0cm]
			\node[paneltitle] at (0.70,2.82) {(b) Packing $\PP'=\mathcal{R}\sqcup\widetilde{\pa}\sqcup\pa'_E$\\Hamiltonian vertex $\boldsymbol{j}'$};
			\coordinate (ellbot) at ({\rtthree*\side},0);
			\coordinate (ebot) at (0,{0.5*\side});
			\coordinate (jbot) at (0,{-0.5*\side});
			\coordinate (mbot) at ({-0.5*\evertex},{-0.5*\side-\rtthree*\evertex});
			\coordinate (oblo) at ({\rtthree*\side+\osolid},0);
			\coordinate (obfar) at ({\rtthree*\side+\ofar},0);
			\coordinate (epstart) at ({-0.5*\evertex},{0.5*\side+\rtthree*\evertex});
			\coordinate (epsolidbot) at ({-0.5*\esolid},{0.5*\side+\rtthree*\esolid});
			\coordinate (epfarbot) at ({-0.5*\efar},{0.5*\side+\rtthree*\efar});
			\coordinate (msolidbot) at ({-0.5*\esolid},{-0.5*\side-\rtthree*\esolid});
			\coordinate (mfarbot) at ({-0.5*\efar},{-0.5*\side-\rtthree*\efar});

			\draw[twohalo] (obfar) -- (ellbot) -- (jbot) -- (mbot) -- (mfarbot);
			\draw[onehalo] (epstart) -- (epfarbot);
			\draw[path,dotted] (obfar) -- (oblo);
			\draw[path] (oblo) -- (ellbot) -- (jbot) -- (mbot) -- (msolidbot);
			\draw[path,dotted] (msolidbot) -- (mfarbot);
			\draw[onepath] (epstart) -- (epsolidbot);
			\draw[onepathcont] (epsolidbot) -- (epfarbot);
			\draw[touchingedge] (ebot) -- (ellbot);
			\draw[touchingedge] (ebot) -- (jbot);
			\draw[touchingedge] (ebot) -- (epstart);

			\node[vertex] (ob1) at (ellbot) {};
			\node[ham] (eb) at (ebot) {};
			\node[vertex] (epb) at (epstart) {};
			\node[vertex] (jb) at (jbot) {};
			\node[vertex] (mb) at (mbot) {};

			\node[left=2pt] at (eb) {\scriptsize $\boldsymbol{j}'$};
			\node[left=2pt] at (jb) {\scriptsize $\boldsymbol{j}$};
			\node[below=2pt] at (ob1) {\scriptsize $\ellbm$};
			\node[right=2pt] at (mb) {\scriptsize $\boldsymbol{m}$};
			\coordinate (obstartbot) at ($(ellbot)!-0.12!(obfar)+(0,\obracesep)$);
			\coordinate (obendbot) at ($(ellbot)!1.08!(obfar)+(0,\obracesep)$);
			\coordinate (ebstartbot) at ($(ebot)!0.45!(epfarbot)+({\ebracesep*\rtthree},{0.5*\ebracesep})$);
			\coordinate (ebendbot) at ($(ebot)!1.15!(epfarbot)+({\ebracesep*\rtthree},{0.5*\ebracesep})$);
			\coordinate (mbstartbot) at ($(mbot)!-0.3!(mfarbot)+({-\mbracesep*\rtthree},{0.5*\mbracesep})$);
			\coordinate (mbendbot) at ($(mbot)!1.3!(mfarbot)+({-\mbracesep*\rtthree},{0.5*\mbracesep})$);
			\draw[piecebrace] (obstartbot) -- (obendbot);
			\node at ($(obstartbot)!0.55!(obendbot)+(0,\labelsep)$) {\scriptsize $\pa_O$};
			\draw[piecebracemirror] (ebstartbot) -- (ebendbot);
			\node at ($(ebstartbot)!0.55!(ebendbot)+({1.25*\labelsep*\rtthree},{0.45*\labelsep})$) {\scriptsize $\pa'_E$};
			\draw[piecebracemirror] (mbstartbot) -- (mbendbot);
			\node at ($(mbstartbot)!0.55!(mbendbot)+({-1.25*\labelsep*\rtthree},{0.75*\labelsep})$) {\scriptsize $\pa'$};
			\node[anchor=west] at ($(oblo)+(-1.1,-1.2)$) {\large $\pa'\md\boldsymbol{j}\md\pa_O=\widetilde{\pa}$};
		\end{scope}
	\end{tikzpicture}
	\caption{The paired $(2,1)$ mixed-touch configurations in Lemma~\ref{lem:mixed-touch-involution}.
		The filled Hamiltonian vertex is $\boldsymbol{j}$ in panel (a) and $\boldsymbol{j}'$ in panel (b), while the braces label the path pieces before and after the correspondence.
		Dotted segments omit path continuations, dashed edges are touching edges, and the yellow and blue halos mark the doubly and singly touched components, respectively.}
	\label{fig:mixed-touch-cancellation}
\end{figure}

\begin{proof}
	Since distinct packed components are mutually non-adjacent, a vertex touching three components would have three pairwise non-adjacent neighbors, one in each component, and would center a claw; hence at most two components are touched.
	On one component, two non-consecutive touched vertices are non-adjacent and, with any touched vertex on the other component, form a claw at $\jbm$; hence each touched component contributes at most two touched vertices, necessarily consecutive.
	A nonzero commutator requires an odd total number of touched vertices, leaving only the $(2,1)$ profile displayed in Fig.~\ref{fig:mixed-touch-cancellation}.
	If the singly touched vertex were not an endpoint of its component, its two path-neighbors and $\jbm$ would be three pairwise non-adjacent neighbors of that vertex, a claw.

	We use the notation of Fig.~\ref{fig:mixed-touch-cancellation}(a), so that the doubly touched component is $\pa=\pa_E\md\pa_O$ with $\pa_E=\pa'_E\md\boldsymbol{j}'$ and the singly touched component is $\pa'$.
	The two vertices of $\pa$ adjacent to $\jbm$ are $\boldsymbol{j}'$ and the endpoint $\ellbm$ of $\pa_O$ next to $\boldsymbol{j}'$, and $\boldsymbol{m}$ is the endpoint of $\pa'$ adjacent to $\jbm$.
	The split $\pa=\pa_E\md\pa_O$ is chosen so that $|\pa_E|$ is even; hence $\pa_O$, $\pa'_E$, and $\pa'$ all have odd size.
	The move replaces the summand with packing $\PP=\mathcal R\sqcup\pa\sqcup\pa'$ and Hamiltonian vertex $\jbm$ by the summand with packing $\PP'=\mathcal R\sqcup\widetilde\pa\sqcup\pa'_E$, $\widetilde\pa=\pa'\md\jbm\md\pa_O$, and Hamiltonian vertex $\boldsymbol{j}'$.
	Then $\PP'$ is again an odd path packing: $|\widetilde\pa|=|\pa'|+1+|\pa_O|$ is odd because $|\pa'|$ and $|\pa_O|$ are, and $|\pa'_E|$ is odd.
	The path $\pa'_E$ is induced because it is a subpath of $\pa$.
	The path $\widetilde\pa$ is induced because the original packing forbids edges between $\pa$ and $\pa'$, while the $(2,1)$ profile makes $\boldsymbol{m}$ and $\ellbm$ the only vertices of $\widetilde\pa$ adjacent to $\jbm$.

	Since neither $\jbm$ nor $\boldsymbol{j}'$ touches $\mathcal{R}$, the move leaves $\mathcal{R}$ unchanged and only exchanges $\jbm$ and $\boldsymbol{j}'$ between the Hamiltonian insertion and the packed support.
	Writing
	\begin{equation}
		X\equiv V(\mathcal{R})\cup V(\pa_O)\cup V(\pa'_E)\cup V(\pa')
	\end{equation}
	for the common packed support, we have
	\begin{equation}
		V(\PP)=X\sqcup\{\boldsymbol{j}'\},
		\qquad
		V(\PP')=X\sqcup\{\jbm\},
	\end{equation}
	and hence $|\PP|=|\PP'|$.
	Claw-freeness moreover gives
	\begin{equation}
		\Gamma[\jbm]\subseteq\Gamma[X],
		\qquad
		\Gamma[\boldsymbol{j}']\subseteq\Gamma[X].
	\end{equation}
	Indeed, $\jbm\in\Gamma[X]$ because it is adjacent to $\boldsymbol{m}\in V(\pa')\subseteq X$.
	If a neighbor $w$ of $\jbm$ satisfied $w\notin\Gamma[X]$, then $w$ would be adjacent to neither $\boldsymbol{m}$ nor $\ellbm$, both of which lie in $X$, while $\boldsymbol{m}\not\sim\ellbm$ because they lie in distinct packed components of $\PP$.
	Then $w$, $\boldsymbol{m}$, and $\ellbm$ would be three pairwise non-adjacent neighbors of $\jbm$, a claw.
	Symmetrically, $\boldsymbol{j}'\in\Gamma[X]$ via $\ellbm$, and any other neighbor $w\notin\Gamma[X]$ would form a claw at $\boldsymbol{j}'$ with leaves $w$, $\ellbm$, and the endpoint of $\pa'_E$ adjacent to $\boldsymbol{j}'$, the latter pair being non-adjacent because $\pa=\pa'_E\md\boldsymbol{j}'\md\pa_O$ is induced.
	Combined with the two decompositions of the packed supports, these inclusions yield
	\begin{equation}
		\Gamma[\PP]=\Gamma[X]=\Gamma[\PP'],
	\end{equation}
	and in particular $\complementset{\PP}=\complementset{\PP'}$.
	The paired summands therefore lie in the same $n$ layer of Eq.~\eqref{eq:path-packing-charge},
	\begin{equation}
		\indepcoeff{\complementset{\PP}}{n}
		=
		\indepcoeff{\complementset{\PP'}}{n}.
	\end{equation}
	Applying the construction again returns the original summand, so it is an involution.

	Choose the orientations of the odd path products so that
	\begin{equation}
		\pathprodH{\pa}=\pathprodH{\pa_O}\hc_{\boldsymbol{j}'}\pathprodH{\pa'_E},
		\qquad
		\pathprodH{\widetilde\pa}=\pathprodH{\pa_O}\hc_{\jbm}\pathprodH{\pa'}.
	\end{equation}
	Since $\jbm$ and $\boldsymbol{j}'$ do not touch any component of $\mathcal R$, the left-hand side of the identity in \textup{(iii)} is $\pathprodH{\mathcal R}$ times
	\begin{equation}
		\pathprodH{\pa_O}\hc_{\boldsymbol{j}'}\pathprodH{\pa'_E}\pathprodH{\pa'}\hc_{\jbm}
		+
		\pathprodH{\pa_O}\hc_{\jbm}\pathprodH{\pa'}\pathprodH{\pa'_E}\hc_{\boldsymbol{j}'}.
	\end{equation}
	The factors from $\mathcal R$ and the two odd pieces $\pa'$ and $\pa'_E$ commute past each other.
	The vertex $\jbm$ touches exactly the endpoint of $\pa'$ and commutes with $\pa'_E$, while $\boldsymbol{j}'$ touches exactly the endpoint of $\pa'_E$ and commutes with $\pa'$.
	The expression therefore reduces to
	\begin{equation}
		-\pathprodH{\pa_O}(\hc_{\boldsymbol{j}'}\hc_{\jbm}+\hc_{\jbm}\hc_{\boldsymbol{j}'})\pathprodH{\pa'}\pathprodH{\pa'_E}=0,
	\end{equation}
	because $\jbm\sim\boldsymbol{j}'$.
\end{proof}

\begin{thm}
	\label{thm:generalized-conserved-charges}
	Let $G$ be a claw-free graph and let $2K$ denote the size of the smallest even bubble wand in $G$ (set $K=\infty$ if none exists).
	Then the charges $\packingcharge{G}{m}{c}$ of Eq.~\eqref{eq:path-packing-charge} satisfy
	\begin{align}
		[H_G,\packingcharge{G}{m}{c}]=0
	\end{align}
	for all integers $c\ge 1$ and $m\ge c$ with $m\equiv c\pmod 2$ and $m-c<2K$.
\end{thm}
\begin{proof}
	Fix $c$ and $m$ with $m-c<2K$.
	Using Eq.~\eqref{eq:path-packing-charge} and $H_G=\sum_{\jbm\in V(G)}\hc_{\jbm}$, we obtain
	\begin{align}
		\label{eq:generalized-charge-commutator-expansion}
		\qty[H_G,\packingcharge{G}{m}{c}]
		=
		\sum_{n=0}^{(m-c)/2}
		\sum_{\PP\in\oddpathsetlengthcomp{G}{m-2n}{c}}
		\sum_{\jbm\in V(G)}
		\indepcoeff{\complementset{\PP}}{n}\,
		\qty[\hc_{\jbm},\pathprodH{\PP}].
	\end{align}
	We classify a fixed $(n,\PP,\jbm)$ term according to which packed components are touched by $\jbm$, in the sense fixed before Lemma~\ref{lem:mixed-touch-involution}.
	When $\jbm$ touches exactly one component, we call the summand a \emph{one-component-touch summand}; the unique touched component is the \emph{active} component, and the remaining $c-1$ components form the spectator packing.
	Figure~\ref{fig:one-component-touch-schematic} schematizes such a summand.

	\begin{figure}[t]
		\centering
		\begin{tikzpicture}[
				scale=1.0,
				line cap=round,
				line join=round,
				vertex/.style={circle, draw, thick, fill=white, minimum size=\vertexsize, inner sep=0pt},
				spectatorv/.style={circle, draw, thick, fill=white, minimum size=\vertexsize, inner sep=0pt},
				activev/.style={circle, draw, thick, fill=white, minimum size=\vertexsize, inner sep=0pt},
				ham/.style={circle, draw, thick, fill=black, minimum size=\vertexsize, inner sep=0pt},
				spectatorhood/.style={line width=20pt, draw=gray!30, opacity=0.62},
				activehood/.style={line width=22pt, draw=Red!18, opacity=0.56},
				resultedgehood/.style={line width=5pt, draw=RoyalBlue},
				resultvertexhood/.style={fill=RoyalBlue},
				spectatoredge/.style={line width=1.2pt, draw=black},
				activeedge/.style={line width=1.2pt, draw=black},
				touchingedge/.style={thick, gray, dashed},
				paneltitle/.style={font=\scriptsize, align=left}
			]

			\begin{scope}[yshift=0cm]
				\node[paneltitle, anchor=west] at (-0.15,1.28) {(a) Endpoint shortening, $c=4$};
				\coordinate (s1a) at (0.00,0.18);
				\coordinate (s1b) at (0.66,-0.22);
				\coordinate (s1c) at (1.25,0.20);
				\coordinate (s2a) at (2.05,-0.25);
				\coordinate (s2b) at (2.62,0.34);
				\coordinate (s2c) at (3.15,-0.02);
				\coordinate (s2d) at (3.72,0.42);
				\coordinate (s2e) at (4.25,-0.14);
				\coordinate (s3) at (5.35,0.04);
				\coordinate (a1) at (6.95,-0.12);
				\coordinate (a2) at (7.48,0.32);
				\coordinate (a3) at (8.05,0.04);
				\coordinate (a4) at (8.62,0.40);
				\coordinate (h) at (9.18,0.06);
				\draw[spectatorhood] (s1a) -- (s1b) -- (s1c);
				\draw[spectatorhood] (s2a) -- (s2b) -- (s2c) -- (s2d) -- (s2e);
				\fill[gray!30, opacity=0.62] (s3) circle [radius=0.33];
				\draw[activehood] (a1) -- (a2) -- (a3) -- (a4) -- (h);
				\begin{scope}[transparency group, opacity=0.46]
					\draw[resultedgehood] (a1) -- (a2) -- (a3) -- (a4);
					\foreach \v in {a1,a2,a3,a4} {
							\fill[resultvertexhood] (\v) circle [radius=0.19];
						}
				\end{scope}
				\draw[spectatoredge] (s1a) -- (s1b) -- (s1c);
				\draw[spectatoredge] (s2a) -- (s2b) -- (s2c) -- (s2d) -- (s2e);
				\draw[activeedge] (a1) -- (a2) -- (a3) -- (a4) -- (h);
				\node[spectatorv] at (s1a) {};
				\node[spectatorv] at (s1b) {};
				\node[spectatorv] at (s1c) {};
				\node[spectatorv] at (s2a) {};
				\node[spectatorv] at (s2b) {};
				\node[spectatorv] at (s2c) {};
				\node[spectatorv] at (s2d) {};
				\node[spectatorv] at (s2e) {};
				\node[spectatorv] at (s3) {};
				\node[activev] at (a1) {};
				\node[activev] at (a2) {};
				\node[activev] at (a3) {};
				\node[activev] at (a4) {};
				\node[ham] at (h) {};
			\end{scope}

			\begin{scope}[yshift=-2.35cm]
				\node[paneltitle, anchor=west] at (-0.15,1.28) {(b) Endpoint extension, $c=4$};
				\coordinate (s1a) at (0.00,0.18);
				\coordinate (s1b) at (0.66,-0.22);
				\coordinate (s1c) at (1.25,0.20);
				\coordinate (s2a) at (2.05,-0.25);
				\coordinate (s2b) at (2.62,0.34);
				\coordinate (s2c) at (3.15,-0.02);
				\coordinate (s2d) at (3.72,0.42);
				\coordinate (s2e) at (4.25,-0.14);
				\coordinate (s3) at (5.35,0.04);
				\coordinate (a1) at (6.95,0.28);
				\coordinate (a2) at (7.55,-0.14);
				\coordinate (a3) at (8.18,0.18);
				\coordinate (h) at (8.88,0.18);
				\draw[spectatorhood] (s1a) -- (s1b) -- (s1c);
				\draw[spectatorhood] (s2a) -- (s2b) -- (s2c) -- (s2d) -- (s2e);
				\fill[gray!30, opacity=0.62] (s3) circle [radius=0.33];
				\draw[activehood] (a1) -- (a2) -- (a3) -- (h);
				\begin{scope}[transparency group, opacity=0.46]
					\draw[resultedgehood] (a1) -- (a2) -- (a3) -- (h);
					\foreach \v in {a1,a2,a3,h} {
							\fill[resultvertexhood] (\v) circle [radius=0.19];
						}
				\end{scope}
				\draw[spectatoredge] (s1a) -- (s1b) -- (s1c);
				\draw[spectatoredge] (s2a) -- (s2b) -- (s2c) -- (s2d) -- (s2e);
				\draw[activeedge] (a1) -- (a2) -- (a3);
				\draw[touchingedge] (a3) -- (h);
				\node[spectatorv] at (s1a) {};
				\node[spectatorv] at (s1b) {};
				\node[spectatorv] at (s1c) {};
				\node[spectatorv] at (s2a) {};
				\node[spectatorv] at (s2b) {};
				\node[spectatorv] at (s2c) {};
				\node[spectatorv] at (s2d) {};
				\node[spectatorv] at (s2e) {};
				\node[spectatorv] at (s3) {};
				\node[activev] at (a1) {};
				\node[activev] at (a2) {};
				\node[activev] at (a3) {};
				\node[ham] at (h) {};
			\end{scope}
		\end{tikzpicture}
		\caption{One-component-touch summands with three spectators ($c=4$); only packing vertices and the filled Hamiltonian vertex $\jbm$ are shown.
			The gray and red halos indicate the closed neighborhoods of the spectator and active components, respectively; they may overlap provided that no vertex of one packed component lies in the halo of another.
			Panels (a) and (b) show endpoint shortening and extension.
			Blue marks the active path after the move, and the dashed edge in panel (b) marks the touching edge for the extension.}
		\label{fig:one-component-touch-schematic}
	\end{figure}

	A summand in which $\jbm$ touches no packed component contributes zero, since $\hc_{\jbm}$ then commutes with $\pathprodH{\PP}$; this includes the case where $\jbm$ is itself a singleton packed component.
	For a one-component-touch summand, write $\PP=\mathcal{R}\sqcup\pa$ with active component $\pa$ and spectator packing $\mathcal{R}$, and fix $\mathcal{R}$ first.
	Set $G_{\mathcal{R}}=G\setminus\Gamma[\mathcal{R}]$, an induced subgraph of $G$ and hence again claw-free.
	Let $2K_{\mathcal{R}}$ denote the size of the smallest even bubble wand in $G_{\mathcal{R}}$ (set $K_{\mathcal{R}}=\infty$ if none exists).
	Every even bubble wand in $G_{\mathcal{R}}$ is also an induced even bubble wand in $G$, so $K_{\mathcal{R}}\ge K$.

	The active components compatible with $\mathcal{R}$ are exactly the odd induced paths of $G_{\mathcal{R}}$, since the packing condition is precisely that $\pa$ avoids $\Gamma[\mathcal{R}]$; for such a $\pa$,
	\begin{align*}
		G\setminus\Gamma_G[\mathcal{R}\sqcup\pa]
		=
		G_{\mathcal{R}}\setminus\Gamma_{G_{\mathcal{R}}}[\pa]\,.
	\end{align*}
	Thus, after fixing $\mathcal{R}$, the residual graph entering the coefficient in Eq.~\eqref{eq:path-packing-charge} is exactly the residual graph used in the local conserved-charge formula~\eqref{eq:local-charge} for the graph $G_{\mathcal{R}}$.
	Moreover $m-|\mathcal{R}|$ is a positive odd integer; write $m-|\mathcal{R}|=2k_{\mathcal{R}}+1$.
	Indeed, in the $n$th layer of Eq.~\eqref{eq:path-packing-charge}, the active component has size $|\pa|=2k_{\mathcal{R}}+1-2n$, which is precisely the path size appearing in the $n$th layer of $\localcharge{G_{\mathcal{R}}}{2k_{\mathcal{R}}+1}$ in Eq.~\eqref{eq:local-charge}.
	Hence the scalar coefficient $\indepcoeff{\complementset{\mathcal{R}\sqcup\pa}}{n}$ is the coefficient of $\pathprodH{\pa}$ in $\localcharge{G_{\mathcal{R}}}{2k_{\mathcal{R}}+1}$.
	The inequality $|\mathcal{R}|\ge c-1$ gives $k_{\mathcal{R}}\le (m-c)/2<K\le K_{\mathcal{R}}$.

	In any one-component-touch summand $\jbm$ does not touch the spectator and so lies in $V(G_{\mathcal{R}})$; extending the inner sum to all $\jbm\in V(G_{\mathcal{R}})$ adds only zero summands, namely those indexed by vertices that do not touch $\pa$ in the above sense, including the term with $\jbm$ equal to the unique vertex of $\pa$ when $\pa$ is a singleton.
	The two local moves on the active path---endpoint shortening, with $\jbm$ an endpoint of $\pa$, and endpoint extension, with $\jbm$ outside $\pa$---are the two path-product terms in Eq.~\eqref{eq:local-charge-single-path-comm}, shown in the two panels of Fig.~\ref{fig:one-component-touch-schematic}.
	Let $\mathcal{C}_{\mathcal{R}}^{(1)}$ denote the sum of the one-component-touch contribution with this fixed spectator $\mathcal{R}$, after adding these zero summands.
	Then Eq.~\eqref{eq:local-charge} gives
	\begin{align*}
		\mathcal{C}_{\mathcal{R}}^{(1)}
		 & =
		\sum_{n=0}^{k_{\mathcal{R}}}
		\sum_{\pa\in\pathsetlength{G_{\mathcal{R}}}{2k_{\mathcal{R}}+1-2n}}
		\indepcoeff{G_{\mathcal{R}}\setminus\Gamma_{G_{\mathcal{R}}}[\pa]}{n}
		\sum_{\jbm\in V(G_{\mathcal{R}})}
		\qty[\hc_{\jbm},\pathprodH{\mathcal{R}}\pathprodH{\pa}].
	\end{align*}
	For every $\jbm\in V(G_{\mathcal{R}})$ we have $\qty[\hc_{\jbm},\pathprodH{\mathcal{R}}]=0$, because $\jbm$ lies outside the closed neighborhood of the spectator packing.
	Therefore $\qty[\hc_{\jbm},\pathprodH{\mathcal{R}}\pathprodH{\pa}]=\pathprodH{\mathcal{R}}\qty[\hc_{\jbm},\pathprodH{\pa}]$, and hence
	\begin{align*}
		\mathcal{C}_{\mathcal{R}}^{(1)}
		 & =
		\pathprodH{\mathcal{R}}
		\sum_{n=0}^{k_{\mathcal{R}}}
		\sum_{\pa\in\pathsetlength{G_{\mathcal{R}}}{2k_{\mathcal{R}}+1-2n}}
		\indepcoeff{G_{\mathcal{R}}\setminus\Gamma_{G_{\mathcal{R}}}[\pa]}{n}
		\sum_{\jbm\in V(G_{\mathcal{R}})}
		\qty[\hc_{\jbm},\pathprodH{\pa}]
		\\
		 & =
		\pathprodH{\mathcal{R}}\,
		\qty[H_{G_{\mathcal{R}}},\localcharge{G_{\mathcal{R}}}{2k_{\mathcal{R}}+1}],
	\end{align*}
	which vanishes for every $\mathcal{R}$ by Theorem~\ref{thm:local-charge-conservation-pbc} applied inside $G_{\mathcal{R}}$.

	Mixed-touch summands cancel by Lemma~\ref{lem:mixed-touch-involution}.
\end{proof}

Inside the abstract graph Clifford algebra, and in any faithful finite-dimensional representation of it, the generalized charges $\packingcharge{G}{m}{c}$ with $\oddpathsetlengthcomp{G}{m}{c}\ne\emptyset$ are linearly independent.
Indeed, the $n=0$ layer of Eq.~\eqref{eq:path-packing-charge} is supported on packings of size $m$ with exactly $c$ connected components, while every $n>0$ term has packed support size at most $m-2$.
Since the ordered monomials form a basis, the $n=0$ contributions with different $c$ cannot cancel at the largest value of $m$ appearing in a linear relation; descending in $m$ proves linear independence.

These linear independence statements should not be confused with algebraic independence.
In fact, the products close into a larger algebra: products of the nonlocal independent-set charges in Eq.~\eqref{eq:nonlocal-charge-independent-set}, for example, can be expanded in terms of the generalized charges $\packingcharge{G}{m}{c}$ together with the generalized cycle symmetries of Ref.~\cite{unified-graph-th}.
We do not develop this algebra here.

\subsection{Oriented even path operators}
\label{subsec:even-path-operators}

Equation~\eqref{eq:path-product} gives $\pathprodH{\pa^{-1}}=-\pathprodH{\pa}$ for every even path $\pa$, so defining an even-path operator requires a choice of orientation for each even induced path.

For each ordered adjacent pair $i\sim j$, set
\begin{align}
	\eta_{ij}
	=
	\begin{cases}
		+1, & i\to j, \\
		-1, & j\to i,
	\end{cases}
	\qquad
	\eta_{ji}=-\eta_{ij}.
	\label{eq:edge-orientation-sign}
\end{align}
We call the orientation an \emph{induced-path orientation} if every induced three-vertex path $\ibm\md\jbm\md\ellbm$ obeys
\begin{align}
	\eta_{\ibm,\jbm}\eta_{\jbm,\ellbm}=1.
	\label{eq:induced-path-orientation-sign}
\end{align}
Equivalently, every induced three-vertex path must be oriented as $\ibm\to\jbm\to\ellbm$ or as $\ibm\leftarrow\jbm\leftarrow\ellbm$; applying this to consecutive triples forces all edges of an induced path to point consistently along it.

Once an induced-path orientation is fixed, every unoriented induced path $\pa$ has a preferred directed ordering, which we denote by $\omega(\pa)$.
Concretely, given any listing $\pa=\ell_1\md\cdots\md\ell_m$ with $m\ge 2$, we have $\omega(\pa)=\pa$ if $\eta_{\ell_1,\ell_2}=+1$ and $\omega(\pa)=\pa^{-1}$ if $\eta_{\ell_1,\ell_2}=-1$.
For an even path this convention fixes the otherwise ambiguous sign by
\begin{align}
	\pathprodH{\omega(\pa)}
	=
	\eta_{\ell_1,\ell_2}\,\pathprodH{\pa},
	\qquad |\pa| \text{ even}.
	\label{eq:even-path-orientation-sign}
\end{align}
Induced-path orientations need not exist for general claw-free graphs; see Remark~\ref{rmk:induced-path-orientation-existence}.

\begin{define}[Oriented even path operator]
	\label{def:even-path-operator}
	Fix an induced-path orientation, with the directed ordering $\omega(\pa)$ fixed above.
	For each integer $k\ge 1$, the oriented even path operator is
	\begin{align}
		\mathcal{E}_{G,\omega}^{(2k)}
		\equiv
		\sum_{n=0}^{k-1}
		\sum_{\pa\in\pathsetlength{G}{2k-2n}}
		\indepcoeff{\complementset{\pa}}{n}\,
		\pathprodH{\omega(\pa)}.
		\label{eq:even-path-charge}
	\end{align}
	The inner sum runs over one representative of each reversal class of induced paths; the value is independent of the choice because $\omega(\pa)=\omega(\pa^{-1})$.
\end{define}

This is the even-parity, oriented counterpart of the odd local conserved charge $\localcharge{G}{2k+1}$ in Eq.~\eqref{eq:local-charge}.

The range condition below is set by the smallest induced even hole, and is therefore stronger than the even-bubble-wand condition used for the odd local conserved charges.
The reason is that an endpoint loop can attach at the first edge of a path of even size, leaving a single vertex outside the induced cycle; such a configuration is not an even bubble wand, whose path part has at least two vertices.

\begin{thm}
	\label{thm:even-path-commutator}
	Let $G$ be a claw-free graph that admits an induced-path orientation $\omega$, and let $2K$ denote the size of the smallest induced even hole in $G$ (set $K=\infty$ if none exists).
	Then for every integer $k$ with $1\le k<K$,
	\begin{align}
		\frac{1}{2}\qty[H_G,\mathcal{E}_{G,\omega}^{(2k)}]
		=
		\sum_{\jbm\in V(G)}
		B_{\jbm,\omega}^{(k-1)}\,\hc_{\jbm},
		\label{eq:even-path-commutator}
	\end{align}
	where, for every integer $s\ge 0$,
	\begin{align}
		B_{\jbm,\omega}^{(s)}
		\equiv
		\sum_{\ibm\to\jbm}
		b_{\ibm}^2\,\indepcoeff{\complementset{\ibm\md\jbm}}{s}
		-
		\sum_{\jbm\to\ibm}
		b_{\ibm}^2\,\indepcoeff{\complementset{\jbm\md\ibm}}{s}.
		\label{eq:even-path-singleton-coeff}
	\end{align}
	Here $\complementset{\ibm\md\jbm}=\complementset{\jbm\md\ibm}$; the listing marks only the edge orientation.
\end{thm}
\begin{proof}
	Fix $k$ with $1\le k<K$ and use the convention $\indepcoeff{F}{t}=0$ for $t<0$.
	By Eq.~\eqref{eq:induced-path-orientation-sign}, all edges of an induced path carry the same $\eta$ relative to a given listing.
	Two listings that share an edge in the same order therefore have the same $\eta_{\ell_1,\ell_2}$, which for an even listing is the orientation sign in Eq.~\eqref{eq:even-path-orientation-sign}.

	Let $\rho=\ell_1\md\cdots\md\ell_M$ be a path occurring in $\mathcal{E}_{G,\omega}^{(2k)}$, so that $M$ is even and $M\le 2k$, and consider its non-path terms $\pathsvd_{\rho}$ and $\pathloop_{\rho}^{\pm}$ in Eq.~\eqref{eq:local-charge-single-path-comm}.
	No endpoint-loop term of $\rho$ has an even associated cycle.
	Indeed, after reversing $\rho$ if necessary, such a term comes from an off-path vertex $\jbm$ adjacent to exactly $\ell_{q-1},\ell_q,\ell_M$ with $2\le q\le M-2$ (Lemma~\ref{lem:odd-neighbor-classification}), and its cycle $(\ell_q,\ell_{q+1},\ldots,\ell_M,\jbm)$ is induced of length $M-q+2\le M\le 2k<2K$.
	An even cycle of that length would contradict the definition of $K$.

	The remaining non-path terms cancel in pairs, as in the proof of Lemma~\ref{lem:odd-local-nonpath-cancellation}.
	The pairings of Lemma~\ref{lem:path-cancellation-pairings} replace a single vertex $\ell_q$ of $\rho$, with $2\le q\le M-1$, by the off-path vertex $\jbm$, and the partner $\rho'$ has the same size and closed neighborhood as $\rho$, hence the same coefficient in $\mathcal{E}_{G,\omega}^{(2k)}$.
	For $q\ge 3$ the listings of $\rho$ and $\rho'$ share the edge $\ell_1\md\ell_2$.
	The value $q=2$ occurs only for the three-neighbor pairing, because the loop pairing requires an odd cycle whereas $M-q+2=M$ is even, and there the two listings share the edge $\ell_{M-1}\md\ell_M$, which exists since $M\ge 4$.
	In both cases the orientation multiplies $\pathprodH{\rho}$ and $\pathprodH{\rho'}$ by the same sign, so the unoriented cancellation applies unchanged.
	The commutator is therefore supported on odd induced paths, including singletons.

	Fix an odd induced path $\pa=\ell_1\md\cdots\md\ell_{2m+1}$ with $1\le m\le k$, listed in the direction inherited from the orientation, and set $s=k-m$; longer odd paths are not produced.
	The listing leaves $\pathprodH{\pa}$ unchanged because $|\pa|$ is odd.
	Every path contributing to $\pathprodH{\pa}$ is $\pa$ with one endpoint deleted or added, so it shares an edge with $\pa$ and enters $\mathcal{E}_{G,\omega}^{(2k)}$ with orientation sign $+1$.
	With the endpoint cliques $\mathcal{K}_{\pa}^{\pm}$ of the proof of Theorem~\ref{thm:local-charge-conservation-pbc} and the endpoint signs in Eq.~\eqref{eq:local-charge-single-path-comm}, the coefficient of $\pathprodH{\pa}$ in $\frac{1}{2}\qty[H_G,\mathcal{E}_{G,\omega}^{(2k)}]$ is
	\begin{align}
		D_{\pa}^{(k)}
		=
		\indepcoeff{\complementset{\pa_{[2;]}}}{s}
		-
		\indepcoeff{\complementset{\pa_{[;-1]}}}{s}
		+
		\sum_{\jbm\in\mathcal{K}_{\pa}^-}
		b_{\jbm}^2\,
		\indepcoeff{\complementset{\jbm\md\pa}}{s-1}
		-
		\sum_{\jbm\in\mathcal{K}_{\pa}^+}
		b_{\jbm}^2\,
		\indepcoeff{\complementset{\pa\md\jbm}}{s-1}.
		\label{eq:even-path-odd-coeff}
	\end{align}

	This is Eq.~\eqref{eq:local-charge-even-coeff} with the parities of the active and the produced path interchanged, and the two endpoint-clique recursions used there apply verbatim for every $s\ge 0$, both reducing to $\indepcoeff{F}{0}=1$ at $s=0$.
	Hence $D_{\pa}^{(k)}=0$ for every odd induced path with at least one edge.

	A singleton path $\pa=(\jbm)$ has no edge, so the orientation sign survives.
	Only the two-vertex paths at $\jbm$ contribute, each with coefficient $\indepcoeff{\complementset{\ibm\md\jbm}}{k-1}$ and orientation sign $\eta_{\ibm,\jbm}$, so the incident edge $\{\ibm,\jbm\}$ contributes $+b_{\ibm}^2\,\indepcoeff{\complementset{\ibm\md\jbm}}{k-1}\,\hc_{\jbm}$ for $\ibm\to\jbm$ and the same expression with the opposite sign for $\jbm\to\ibm$.
	Summing over the edges at $\jbm$ gives Eq.~\eqref{eq:even-path-singleton-coeff} at $s=k-1$.
\end{proof}

\begin{rmk}
	\label{rmk:induced-path-orientation-existence}
	Existence reduces to a linear system over $\mathbb{F}_2$.
	Fix any reference orientation, let $x_e\in\mathbb{F}_2$ indicate whether edge $e$ is flipped from the reference, and let $\rho_{ij}=\pm 1$ record the reference direction of $\{i,j\}$.
	Each induced three-vertex path $\ibm\md\jbm\md\ellbm$ imposes
	\begin{align}
		x_{\{\ibm,\jbm\}}+x_{\{\jbm,\ellbm\}}
		=
		\frac{1-\rho_{\ibm,\jbm}\rho_{\jbm,\ellbm}}{2}
		\pmod{2},
		\label{eq:induced-path-orientation-F2}
	\end{align}
	and an induced-path orientation exists if and only if this system has a solution.

	Such an orientation can fail to exist for general claw-free graphs.
	For the cone over a five-cycle with apex $s$, each non-adjacent rim pair $i,k$ gives an induced path $i\md s\md k$, which forces $s\md i$ and $s\md k$ to have opposite orientations at $s$.
	The five constraints form an odd cycle and are inconsistent.
\end{rmk}

\begin{rmk}
	In the range $1\le k<K$, Theorem~\ref{thm:even-path-commutator} gives the commutator formula for the oriented even path operators.
	The operator $\mathcal{E}_{G,\omega}^{(2k)}$ is conserved when the singleton coefficients $B_{\jbm,\omega}^{(k-1)}$ vanish for every $\jbm\in V(G)$.
	When $K=\infty$, this gives a full family provided $B_{\jbm,\omega}^{(s)}=0$ for every $\jbm\in V(G)$ and every $s\ge 0$.

	For homogeneous couplings $b_{\jbm}=1$, this condition reduces to the identities
	\begin{align}
		\sum_{\ibm\to\jbm}
		\indepcoeff{\complementset{\ibm\md\jbm}}{s}
		=
		\sum_{\jbm\to\ibm}
		\indepcoeff{\complementset{\jbm\md\ibm}}{s}
		\qquad
		(\jbm\in V(G),\ s\ge 0).
		\label{eq:uniform-edge-residual-balance}
	\end{align}
	A sufficient condition for Eq.~\eqref{eq:uniform-edge-residual-balance} is that, for each vertex $\jbm$, the incoming edges can be paired with the outgoing edges so that paired edges have isomorphic residual graphs.
	This is precisely what happens in the homogeneous periodic graph $C_M^2$ of the Fendley model, as used in Subsection~\ref{subsec:ffd-uniform-periodic}.

	Homogeneous couplings alone do not imply Eq.~\eqref{eq:uniform-edge-residual-balance}.
	For example, the oriented path graph $1\to2\to3$ with $b_{\jbm}=1$ has $B_{1,\omega}^{(0)}=-1$.
\end{rmk}

\section{Eigenstates of the antisymmetric Hamiltonian $H_A$}
\label{sec:eigenstates-antisymmetric}

\label{sec:HA}

We construct the eigenstates of the companion paper~\cite{sajat-solving-ffd-1} from the path-product modes of Theorem~\ref{thm:path-product-expansion}.

Consider a general connected ECF graph $G=(V,E)$ and the tensor-product graph Clifford algebra generated by the left and right actions reviewed in the companion paper~\cite{sajat-solving-ffd-1}.
In this abstract algebra, write
\begin{equation}
	H_G = \sum_{\jbm \in V} b_{\jbm} h_{\jbm},
	\qquad
	\tilde H_G = \sum_{\jbm \in V} b_{\jbm} \tilde h_{\jbm},
	\qquad
	H_A = H_G - \tilde H_G.
\end{equation}
The left and right graph-Clifford actions commute, i.e. $[h_{\jbm},\tilde h_{\ellbm}]=0$ for all $\jbm,\ellbm\in V$.
Fix a total ordering $v_1 < v_2 < \cdots < v_{|V|}$ on the vertices, which labels the tensor factors of the auxiliary spin Hilbert space $\mathcal{H}_{\mathrm{def}}\simeq(\mathbb{C}^2)^{\otimes |V|}$.
Here $\ket{0}$ denotes an up spin and $\ket{1}$ a down spin.
On this auxiliary spin Hilbert space, $X_{\jbm}$ and $Z_{\jbm}$ denote the Pauli matrices acting on the tensor factor labeled by $\jbm$.
The reference vector is the all-up state $\refst \equiv \ket{0\cdots 0}$.
In the defining representation, left- and right-multiplication by the generators are represented by
\begin{align}
	h_{\jbm}
	 & =
	X_{\jbm} \prod_{\substack{\ibm > \jbm \\ (\ibm,\jbm)\in E}} Z_{\ibm},
	 &
	\tilde h_{\jbm}
	 & =
	X_{\jbm} \prod_{\substack{\ibm < \jbm \\ (\ibm,\jbm)\in E}} Z_{\ibm}.
\end{align}
For every word $h_{\ell_1}\cdots h_{\ell_q}$ in the abstract algebra, with no ordering assumed on the sequence $\ell_1,\ldots,\ell_q$, the operator-state correspondence gives
\begin{equation}
	\label{eq:word-on-ref}
	h_{\ell_1}\cdots h_{\ell_q}\refst
	=
	\tilde h_{\ell_q}\cdots \tilde h_{\ell_1}\refst.
\end{equation}
In particular, $h_{\jbm}\refst = \tilde h_{\jbm}\refst$ for every $\jbm\in V$, and therefore $H_A\refst=0$.

Assume now that the auxiliary edge operator $\chi$ is embedded in the original graph Clifford algebra $\AAA$, rather than merely adjoined in the extended algebra.
Whether this is possible can be determined from the adjacency matrix of the frustration graph: the commutation pattern of the edge operator must lie in the $\mathbb{Z}_2$-row span of $A$ (see the companion paper for details and explicit constructions).
An embedded $\chi$ is a basis element of $\AAA$, so we can write
\begin{equation}
	\label{eq:chi-word}
	\chi = c \, h_{\iota_1}\cdots h_{\iota_r},
	\qquad
	\iota_1<\cdots<\iota_r
\end{equation}
for some subset of generators, written in the fixed total order, and a phase $c$ with $|c|=1$.
Define $\tilde\chi \in \tilde\AAA$ by reversing the word and replacing every generator by its right-copy counterpart:
\begin{equation}
	\label{eq:tildechi-def}
	\tilde\chi
	\equiv
	c \, \tilde h_{\iota_r}\cdots \tilde h_{\iota_1}.
\end{equation}
This serves as the edge operator for $\tilde H_G$: it satisfies $\tilde\chi^2=1$, anticommutes with each $\tilde h_{\jbm}$ in the simplicial clique, and commutes with every other $\tilde h_{\jbm}$.
Applying Eq.~\eqref{eq:word-on-ref} to the word in Eq.~\eqref{eq:chi-word}, we obtain
\begin{equation}
	\label{eq:chi-tildechi-ref}
	\chi \refst = \tilde\chi \refst.
\end{equation}
Theorem~\ref{thm:path-product-expansion} applies verbatim to the right copy, with $h_{\jbm}$ and $\chi$ replaced by $\tilde h_{\jbm}$ and $\tilde\chi$.
For $k>0$, we fix the phase of $\tilde\Psi_{-k}$ by
\begin{equation}
	\label{eq:tildepsi-minus-definition}
	\tilde\Psi_{-k}
	\equiv
	\frac{1}{\NN_k}
	\sum_{\pa = \chi \md \ell_1 \md \cdots \md \ell_n \in \pathset{G_\chi}{\chi}}
	u_k^{|\pa|-1}
	P_{\complementset{\pa}}(u_k^2)
	\qty(\prod_{m=1}^{n} b_{\ell_m})
	\tilde\chi \tilde h_{\ell_1}\cdots \tilde h_{\ell_n}.
\end{equation}

\begin{thm}
	\label{thm:Psi-tildePsi-ket}
	In the defining representation of a connected ECF graph, assume that the edge operator $\chi$ is embedded in the graph Clifford algebra.
	Then
	\begin{equation}
		\label{eq:fewbody-claim}
		\Psi_k \refst = \tilde\Psi_{-k} \refst.
	\end{equation}
	In particular, $\Psi_k \refst$ is an eigenstate of $H_A$ with eigenvalue $2\epsilon_k$.
\end{thm}
\begin{proof}
	Separating the minimal rooted path from the longer ones in Eq.~\eqref{eq:free-fermion-mode-path-expansion}, we obtain
	\begin{align}
		\Psi_k \refst
		 & =
		\frac{1}{\NN_k}
		P_{G \setminus K_s}(u_k^2)
		\chi \refst
		\nonumber                                                                              \\
		 & \quad +
		\frac{1}{\NN_k}
		\sum_{\substack{\pa = \chi \md \ell_1 \md \cdots \md \ell_n \in \pathset{G_\chi}{\chi} \\ n \geq 1}}
		(-u_k)^n
		P_{\complementset{\pa}}(u_k^2)
		\qty(\prod_{m=1}^{n} b_{\ell_m})
		\chi h_{\ell_1} \cdots h_{\ell_n} \refst.
		\label{eq:fewbody-proof-start}
	\end{align}
	For a nontrivial induced path, only $\ell_1$ can lie in the simplicial clique $K_s$: otherwise $(\chi,\ell_m)$ with $m\geq2$ would be a chord.
	Therefore $\chi$ anticommutes with $h_{\ell_1}$ and commutes with $h_{\ell_m}$ for all $m \geq 2$, so
	\begin{equation}
		\chi h_{\ell_1} \cdots h_{\ell_n} \refst
		=
		- h_{\ell_1} \cdots h_{\ell_n} \chi \refst
		=
		- \tilde\chi h_{\ell_1} \cdots h_{\ell_n} \refst,
	\end{equation}
	where we used Eq.~\eqref{eq:chi-tildechi-ref} and $[\tilde\chi,h_{\jbm}] = 0$.
	Equation~\eqref{eq:word-on-ref} gives
	\begin{equation}
		h_{\ell_1} \cdots h_{\ell_n} \refst
		=
		\tilde h_{\ell_n} \cdots \tilde h_{\ell_1} \refst.
	\end{equation}
	Since $\pa$ is induced, only consecutive generators anticommute, so reversing the order gives one minus sign per edge:
	\begin{equation}
		\tilde h_{\ell_n} \cdots \tilde h_{\ell_1}
		=
		(-1)^{n-1}
		\tilde h_{\ell_1} \cdots \tilde h_{\ell_n}.
	\end{equation}
	Hence
	\begin{equation}
		\chi h_{\ell_1} \cdots h_{\ell_n} \refst
		=
		(-1)^n
		\tilde\chi \tilde h_{\ell_1} \cdots \tilde h_{\ell_n} \refst.
	\end{equation}
	Substituting this relation into Eq.~\eqref{eq:fewbody-proof-start}, using $u_{-k}=-u_k$ so that $(-u_{-k})^n=u_k^n$, and using Eq.~\eqref{eq:chi-tildechi-ref} for the minimal rooted path, we obtain Eq.~\eqref{eq:tildepsi-minus-definition} acting on $\refst$, which is Eq.~\eqref{eq:fewbody-claim}.

	Finally, Eq.~\eqref{eq:eigenoperator} gives $[H_G,\Psi_k] = 2\epsilon_k \Psi_k$.
	Because $\chi$ is embedded in the left graph Clifford algebra, which commutes with the right copy, $[\tilde H_G,\Psi_k] = 0$.
	Together with $H_A \refst = 0$, this yields
	\begin{equation}
		H_A \Psi_k \refst
		=
		[H_A,\Psi_k]\refst
		=
		2\epsilon_k \Psi_k \refst\,.
	\end{equation}
\end{proof}
Multiparticle states can be generated by products of the fermionic modes, as in the companion paper.

\section{Example: the transverse-field Ising chain}
\label{sec:ising-example}

We illustrate Theorem~\ref{thm:path-product-expansion} with the transverse-field Ising chain~\cite{jordan-wigner,Schulz-Mattis-Lieb,pfeuty-transverse-ising}, the simplest example where our path-product expansion applies.

\subsection{Setup}

The Hamiltonian of the transverse-field Ising chain with open boundary conditions and $L$ sites is
\begin{equation}
	\label{eq:TFI-Hamiltonian}
	H_{\text{TFI}} = J \sum_{j=1}^{L-1} X_j X_{j+1} + g \sum_{j=1}^{L} Z_j,
\end{equation}
where $J$ is the ferromagnetic coupling and $g$ the transverse field strength.
We take the critical case $J = g = 1$.

We introduce $M = 2L - 1$ generators:
\begin{align}
	\hc_{2j-1} & = Z_j,         & j & = 1, 2, \ldots, L,   \\
	\hc_{2j}   & = X_j X_{j+1}, & j & = 1, 2, \ldots, L-1.
\end{align}
These generators satisfy $\{\hc_j, \hc_{j+1}\} = 0$ for $1 \le j \le M-1$, while all non-adjacent pairs commute.
The Hamiltonian~\eqref{eq:TFI-Hamiltonian} takes the standard form~\eqref{eq:Hamiltonian}:
\begin{equation}
	H_{\text{TFI}} = \sum_{j=1}^{M} \hc_j.
\end{equation}

The frustration graph $G$ is the path graph on $M$ vertices, denoted by $P_M$, and is therefore even-hole-free and claw-free (ECF).
We can choose the simplicial clique at the endpoint $K_s = \{1\}$.
One may choose the edge operator $\chi=X_1$, which satisfies $\{\chi, \hc_1\} = 0$ and $[\chi, \hc_j] = 0$ for $j \geq 2$.
The extended graph $\extendedgraph{G}{\chi}$ is shown in Fig.~\ref{fig:ising}.

\begin{figure}[t]
	\centering
	\begin{tikzpicture}[scale=1]
		\node[circle, draw, fill=black, minimum size=\vertexsize, inner sep=0pt, label=below:{$\chi$}] (v0) at (0,0) {};
		\foreach \i in {1,...,7} {
		\node[circle, draw, fill=white, minimum size=\vertexsize, inner sep=0pt, label=below:{$\i$}] (v\i) at (\i,0) {};
		}
		\draw[thick] (v0) -- (v1);
		\foreach \i in {1,...,6} {
				\pgfmathtruncatemacro{\j}{\i+1}
				\draw[thick] (v\i) -- (v\j);
			}
	\end{tikzpicture}
	\caption{The extended frustration graph $\extendedgraph{G}{\chi}$ for the Ising algebra at $M=7$.
		The white vertices form $G=P_7$, and the black edge vertex $\chi$ is attached to the simplicial clique $K_s=\{1\}$.}
	\label{fig:ising}
\end{figure}

\subsection{Independence polynomial}

Write $P_M(x)\equiv P_{G}(x)$ for the independence polynomial of $G = P_M$.
It satisfies the recursion~\eqref{eq:independence-polynomial-recursion} with respect to the endpoint clique:
\begin{equation}
	\label{eq:path-recursion}
	P_M(x) = P_{M-1}(x) - x \, P_{M-2}(x),
\end{equation}
with initial conditions $P_0(x) = P_{-1}(x) = 1$.

The characteristic equation for this recursion is $\alpha^2 - \alpha + x = 0$, with roots
\begin{equation}
	\alpha_\pm(x) = \frac{1 \pm \sqrt{1 - 4x}}{2}.
\end{equation}
Parametrizing $x = u^2$ and $\alpha_\pm = u e^{\pm ip}$ where $u = 1/(2\cos p)$, we obtain the general solution
\begin{equation}
	\label{eq:Pm-general}
	P_M(u^2) = \frac{\sin((M+2)p)}{(2\cos p)^{M+1} \sin p}.
\end{equation}
Equivalently, $P_M(u^2)=u^{M+1}U_{M+1}(\cos p)$, where $U_M$ is the Chebyshev polynomial of the second kind.

The positive roots of $P_M(u_k^2) = 0$ are determined by $\sin((M+2)p_k) = 0$, giving
\begin{equation}
	\label{eq:quantization}
	p_k = \frac{k\pi}{M+2}, \quad k = 1, 2, \ldots, \left\lfloor \frac{M+1}{2} \right\rfloor = L.
\end{equation}
The upper bound is the condition $p_k<\pi/2$, which is equivalent to $u_k=1/(2\cos p_k)>0$.
The single-particle energies are $\epsilon_k = 1/u_k = 2\cos p_k$.

\subsection{Application of Theorem~\ref{thm:path-product-expansion}}

For the path graph with simplicial clique $K_s = \{1\}$, the induced paths in $\pathsetG$ starting from the edge vertex $\chi$ are
\begin{equation}
	\pa_m = \chi \md 1 \md 2 \md \cdots \md m, \quad m = 0, 1, \ldots, M,
\end{equation}
where $\pa_0 = (\chi)$ is the minimal rooted path.

The path $\pa_m$ has size $|\pa_m| = m + 1$ and residual graph $\complementset{\pa_m}$, the path graph with $M - m - 1$ vertices; for $m = M$ this is empty, with $P_{-1}(x) = 1$.

The path products~\eqref{eq:path-product} are
\begin{equation}
	\pathprodH{\pa_m} = \chi \hc_1 \hc_2 \cdots \hc_m,
\end{equation}
with $\pathprodH{\pa_0} = \chi$.

Applying Theorem~\ref{thm:path-product-expansion}, the free-fermion mode is
\begin{equation}
	\label{eq:ising-psi-explicit}
	\Psi_k = \frac{1}{\NN_k} \sum_{m=0}^{M} (-u_k)^m \, P_{M-m-1}(u_k^2) \, \chi \hc_1 \hc_2 \cdots \hc_m.
\end{equation}
Using the explicit form~\eqref{eq:Pm-general} of the independence polynomial, the coefficient of each path product becomes
\begin{equation}
	(-u_k)^m P_{M-m-1}(u_k^2) = (-1)^m \frac{\sin((M-m+1)p_k)}{(2\cos p_k)^{M} \sin p_k}.
\end{equation}
Up to an overall normalization, the free-fermion mode is therefore
\begin{equation}
	\label{eq:ising-psi-sin}
	\Psi_k \propto \sum_{m=0}^{M} (-1)^m \sin((M-m+1)p_k) \, \chi \hc_1 \hc_2 \cdots \hc_m.
\end{equation}

\subsection{Verification via Jordan-Wigner transformation}

We now verify that $\Psi_k$ in Eq.~\eqref{eq:ising-psi-sin} is indeed a free-fermion mode by showing that it reduces to the standard Jordan-Wigner fermion~\cite{jordan-wigner}.

Define the Majorana operators
\begin{align}
	\gamma_0      & = X_1, \nonumber                                                          \\
	\gamma_{2j-1} & = \left(\prod_{r=1}^{j-1} Z_r\right) Y_j,   & j & = 1,\ldots,L, \nonumber \\
	\gamma_{2j}   & = \left(\prod_{r=1}^{j} Z_r\right) X_{j+1}, & j & = 1,\ldots,L-1.
\end{align}
These satisfy the Clifford algebra $\{\gamma_m, \gamma_n\} = 2\delta_{mn}$.
The generators can be expressed as Majorana bilinears:
\begin{equation}
	\hc_{2j-1} = i \gamma_{2j-1} \gamma_{2j-2}\quad (j=1,\ldots,L),
	\qquad
	\hc_{2j} = i \gamma_{2j} \gamma_{2j-1}\quad (j=1,\ldots,L-1).
\end{equation}

With the edge operator $\chi = \gamma_0 = X_1$, the path products simplify.
For example:
\begin{align}
	\chi \hc_1       & = \gamma_0 \cdot (i \gamma_1 \gamma_0) = -i \gamma_1,    \\
	\chi \hc_1 \hc_2 & = (-i\gamma_1) \cdot (i \gamma_2 \gamma_1) = - \gamma_2.
\end{align}
By induction, we obtain
\begin{equation}
	\label{eq:path-to-majorana}
	\chi \hc_1 \hc_2 \cdots \hc_m = (-i)^m \gamma_m.
\end{equation}
The path product telescopes to a single Majorana fermion because consecutive Majorana pairs in the Jordan-Wigner string cancel.

Substituting~\eqref{eq:path-to-majorana} into~\eqref{eq:ising-psi-sin}:
\begin{equation}
	\label{eq:ising-fourier-mode}
	\Psi_k \propto \sum_{m=0}^{M} i^m \sin((M-m+1)p_k) \, \gamma_m.
\end{equation}
This is precisely the Fourier mode of the standard Jordan-Wigner solution, with standing-wave coefficients determined by the open boundary conditions.
The orthogonality and normalization of these standing waves give the fermionic anticommutation relations~\eqref{eq:anticommutation}.

Writing $H_{\text{TFI}} = \frac{i}{2} \sum_{m,n=0}^{M} \gamma_m A_{mn} \gamma_n$, where $A$ is the antisymmetric tridiagonal matrix with $A_{m,m-1}=1=-A_{m-1,m}$ for $1\leq m\leq M$, an eigenoperator of the form
\begin{equation}
	\Psi_k = \sum_{m=0}^{M} i^m a_m \gamma_m
\end{equation}
obeys the discrete eigenvalue equation
\begin{equation}
	a_{m-1} + a_{m+1} = 2\cos p_k \, a_m,
	\qquad a_{-1} = a_{M+1} = 0\,.
\end{equation}
The solution is a standing wave
\begin{equation}
	a_m \propto \sin((m+1)p_k).
\end{equation}
The coefficient in Eq.~\eqref{eq:ising-psi-sin} is the same standing wave up to a $k$-dependent overall sign:
\begin{equation}
	\sin((M-m+1)p_k)=\sin(k\pi-(m+1)p_k)=(-1)^{k+1}\sin((m+1)p_k),
\end{equation}
because $(M+2)p_k=k\pi$.
Thus the path-product expression reproduces the conventional Majorana mode.

For ECF frustration graphs of higher connectivity, several induced paths can connect $\chi$ to a given vertex $\jbm$, so the modes of models beyond the Jordan-Wigner paradigm~\cite{fendley-fermions-in-disguise} are linear combinations of several path products rather than single Majorana operators, and the free-fermion structure is hidden.

\section{Example: the Fendley model}

\label{sec:ffd-example}
The Hamiltonian of the Fendley model reads
\begin{align}
	H_{\text{FFD}} & = \sum_{j=1}^M \hc_j, & \hc_j & = b_j Z_j Z_{j+1} X_{j+2}.
\end{align}
Here $X_j$ and $Z_j$ denote the Pauli matrices acting on spin site $j$.
In contrast to the transverse-field Ising chain, the next-nearest-neighbor generators also anticommute:
\begin{equation}
	\{ \hc_j, \hc_{j+1} \} = \{ \hc_j, \hc_{j+2} \} = 0,
\end{equation}
while $\hc_j$ and $\hc_{j'}$ commute for $\vert j-j'\vert \ge 3$.

On the open chain ($1\le j\le M$), the frustration graph is the graph square $P_M^2$ of the path graph on $M$ vertices: vertices are labeled $1,\dots,M$, and $j$ and $j'$ are adjacent if and only if $1\le \vert j-j'\vert\le 2$.
The nearest-neighbor and next-nearest-neighbor edges combine into the two-layered zigzag layout used in Figures~\ref{fig:ffdGchi} and~\ref{fig:ffd-open-local-five-path}.
With periodic boundary conditions, indices are read modulo $M$ and the path graph $P_M$ is replaced by the cycle graph $C_M$, giving the periodic frustration graph $C_M^2$ in the two-cycle layout of Figure~\ref{fig:ffd-periodic-frustration-graph}.
We assume $M>6$ throughout the periodic case.

Subsections~\ref{subsec:ffd-path-structure}--\ref{subsec:ffd-operator-weight} work out the path-product expansion~\eqref{eq:free-fermion-mode-path-expansion} of the modes $\Psi_k$ on the open chain, where multiple induced paths now connect $\chi$ to a given vertex, and estimate how the resulting operator weight is distributed far from the boundary vertex $\chi$.
Subsections~\ref{subsec:ffd-open-inhomogeneous} and~\ref{subsec:ffd-uniform-periodic} specialize the local conserved-charge formula~\eqref{eq:local-charge} to the open inhomogeneous chain, with frustration graph $P_M^2$, and to the periodic homogeneous chain, with frustration graph $C_M^2$, where translated path-product terms carry common coefficients and can be collected into the closed formula~\eqref{eq:ffd-uniform-general-charge-formula}.
Finally, Subsection~\ref{subsec:ffd-catalan-tree} reorganizes the homogeneous periodic charges into the same Catalan-tree pattern as the local conserved charges of the spin-$1/2$ XXX chain and its $SU(N)$-invariant generalizations~\cite{anshelevich-heisenberg-first-integrals,GM-higher-conserved-XXX,GM-higher-conserved-XXZ,GM-catalan-tree}.

\begin{figure}[t]
	\centering
	\begin{tikzpicture}[
			scale=1,
			line cap=round,
			line join=round,
			vertex/.style={circle, draw, fill=white, minimum size=\vertexsize, inner sep=0pt},
			pathclass/.style={opacity=.82, line width=1.9pt},
			pathone/.style={pathclass, draw=Red},
			pathtwo/.style={pathclass, draw=RoyalBlue, loosely dashed},
			paththree/.style={pathclass, draw=ForestGreen!70!black, dotted}
		]
		\def\N{12}
		\def\yodd{0.85}
		\def\yeven{-0.85}

		\pgfmathtruncatemacro{\NmOne}{\N-1}
		\pgfmathtruncatemacro{\NmTwo}{\N-2}

		\coordinate (c0) at (0,\yodd);
		\foreach \i in {1,...,\N} {
				\ifodd\i\relax
					\coordinate (c\i) at (\i,\yodd);
				\else
					\coordinate (c\i) at (\i,\yeven);
				\fi
			}

		\draw[thick, gray!45] (c0) -- (c1);

		\foreach \i in {1,...,\NmOne} {
				\pgfmathtruncatemacro{\j}{\i+1}
				\draw[thick, gray!45] (c\i) -- (c\j);
			}

		\foreach \i in {1,...,\NmTwo} {
				\pgfmathtruncatemacro{\j}{\i+2}
				\draw[thick, gray!45] (c\i) -- (c\j);
			}

		\coordinate (cc9) at ($(c9)+(-0.85,0.45)$);
		\coordinate (cc10) at ($(cc9)+(1.60,-2.72)$);
		\coordinate (cc11) at ($(c11)+(0.35,0.45)$);
		\coordinate (cc12) at ($(cc11)+(1.60,-2.72)$);

		\draw[opacity=.2, rounded corners=10pt, draw=Red,fill=Red]
		(cc9) -- (cc10) -- (cc12) -- (cc11) -- cycle;

		\draw[pathone]
		([yshift=2.4pt]c0) -- ([yshift=2.4pt]c1)
		-- ([yshift=2.4pt]c2) -- ([yshift=2.4pt]c4) -- ([yshift=2.4pt]c6);
		\draw[pathtwo]
		(c0) -- (c1) -- (c3) -- (c4) -- (c6);
		\draw[paththree]
		([yshift=-2.4pt]c0) -- ([yshift=-2.4pt]c1)
		-- ([yshift=-2.4pt]c3) -- ([yshift=-2.4pt]c5) -- ([yshift=-2.4pt]c6);

		\node[circle, draw, fill=black, minimum size=\vertexsize, inner sep=0pt,
		label=above:{$\chi$}] (v0) at (c0) {};

		\foreach \i in {1,...,\N} {
		\ifodd\i\relax
			\node[vertex, label=above:{$\i$}] (v\i) at (c\i) {};
		\else
			\node[vertex, label=below:{$\i$}] (v\i) at (c\i) {};
		\fi
		}
	\end{tikzpicture}
	\caption{The extended graph $\extendedgraph{G}{\chi}$ of the open Fendley model at $M=12$, with $G=P_{12}^2$.
		For $n=4$ and $r=6$, the solid red, dashed blue, and dotted green paths are $\chi\md1\md2\md4\md6$, $\chi\md1\md3\md4\md6$, and $\chi\md1\md3\md5\md6$, respectively.
		Their common coefficient $(-u)^4P_{9,12}(u^2)$ in Eq.~\eqref{eq:Xi-FFD} is determined by the shaded residual graph.}
	\label{fig:ffdGchi}
\end{figure}

\subsection{Path structure and Krylov basis elements}
\label{subsec:ffd-path-structure}

We choose the edge operator $\chi$ such that
\begin{equation}
	\{ \chi, \hc_1\} = 0.
\end{equation}
This corresponds to the simplicial clique $K_s=\{1\}$ in Figure~\ref{fig:ffdGchi}.
By Theorem~\ref{thm:Krylov-path-expansion} and Lemma~\ref{lem:single-path-commutator}, repeated application of the adjoint action $\frac{1}{2} [ H_{\text{FFD}}, \cdot ]$ maps the path products to linear combinations of operator products
\begin{align}
	\mathtt{H}[\chi]                   & \overset{\frac{1}{2} [ H_{\text{FFD}}, \cdot ]}{\qquad\to\qquad}- \mathtt{H}[\chi \md 1]                                                                                              \\
	\mathtt{H}[\chi \md 1]             & \qquad\to\qquad -b_1^2 \mathtt{H}[\chi] - \mathtt{H}[\chi \md 1 \md 2] - \mathtt{H}[\chi \md 1 \md 3]                                                                                 \\
	\mathtt{H}[\chi \md 1 \md 2]       & \qquad\to\qquad - b_2^2 \mathtt{H}[\chi \md 1] - \mathtt{H}[\chi \md 1 \md 2 \md 4]                                                                                                   \\
	\mathtt{H}[\chi \md 1 \md 3]       & \qquad\to\qquad - b_3^2 \mathtt{H}[\chi \md 1] - \mathtt{H}[\chi \md 1 \md 3 \md 4] - \mathtt{H}[\chi \md 1 \md 3 \md 5]                                                              \\
	\mathtt{H}[\chi \md 1 \md 2 \md 4] & \qquad\to\qquad \underline{+ \chi \hc_1 \hc_2 \hc_3 \hc_4} - b_4^2 \mathtt{H}[\chi \md 1 \md 2] - \mathtt{H}[\chi \md 1 \md 2 \md 4 \md 5] - \mathtt{H}[\chi \md 1 \md 2 \md 4 \md 6] \\
	\mathtt{H}[\chi \md 1 \md 3 \md 4] & \qquad\to\qquad \underline{- \chi \hc_1 \hc_2 \hc_3 \hc_4} - b_4^2 \mathtt{H}[\chi \md 1 \md 3]  - \mathtt{H}[\chi \md 1 \md 3 \md 4 \md 6]                                           \\
	\mathtt{H}[\chi \md 1 \md 3 \md 5] & \qquad\to\qquad -b_5^2 \mathtt{H}[\chi \md 1 \md 3] - \mathtt{H}[\chi \md 1 \md 3 \md 5 \md 6] - \mathtt{H}[\chi \md 1 \md 3 \md 5 \md 7]                                             \\
	                                   & \qquad\;\;\vdots\qquad \nonumber
\end{align}
The underscored terms are not path products: their vertex sequences are not induced paths.
However, in the Krylov recursion $\phi_{j+1}=\frac{1}{2}[H_{\text{FFD}},\phi_j]$, such non-path contributions occur in opposite-sign pairs and cancel, so the basis elements remain linear combinations of induced-path products:
\begin{align}
	\phi_0 & = \mathtt{H}[\chi]                                                                                                                                                                                                      \\
	\phi_1 & = -\mathtt{H}[\chi \md 1]                                                                                                                                                                                               \\
	\phi_2 & = b_1^2 \mathtt{H}[\chi] +  \mathtt{H}[\chi \md 1 \md 2] + \mathtt{H}[\chi \md 1 \md 3]                                                                                                                                 \\
	\phi_3 & = -(b_1^2 + b_2^2 + b_3^2) \mathtt{H}[\chi \md 1] - \mathtt{H}[\chi \md 1 \md 2 \md 4] - \mathtt{H}[\chi \md 1 \md 3 \md 4] - \mathtt{H}[\chi \md 1 \md 3 \md 5]                                                        \\
	\phi_4 & =   (b_1^2 + b_2^2 + b_3^2)  b_1^2 \mathtt{H}[\chi] +  (b_1^2 + b_2^2 + b_3^2 +b_4^2) \mathtt{H}[\chi \md 1 \md 2] +  (b_1^2 + b_2^2 + b_3^2 +b_4^2 + b_5^2)  \mathtt{H}[\chi \md 1 \md 3]  \nonumber                   \\
	       & + \mathtt{H}[\chi \md 1 \md 2 \md 4 \md 5] + \mathtt{H}[\chi \md 1 \md 2 \md 4 \md 6] + \mathtt{H}[\chi \md 1 \md 3 \md 4 \md 6] + \mathtt{H}[\chi \md 1 \md 3 \md 5 \md 6]  + \mathtt{H}[\chi \md 1 \md 3 \md 5 \md 7] \\
	       & \;\; \vdots \nonumber
\end{align}

For a chain of size $M$, write $\Phi_M(u)$ for the generating function~\eqref{eq:Phi-def}.
The coefficient of each induced-path product $\mathtt{H}[\mathcal{L}]$ in $\Phi_M(u)$ is a rational function of $u$.
The common denominator is the independence polynomial defined in \eqref{eq:independence-polynomial}, denoted by $P_M(u^2)$ in this subsection.
After factoring out this denominator as $\Phi_M(u) = \Xi_M(u)/P_M(u^2)$, we obtain the path-product operator \eqref{eq:gen-path-product-operator}:
\begin{align}\label{eq:Xi-FFD}
	\Xi_M(u)
	 & =
	P_{2,M}(u^2) \mathtt{H}[\chi]
	+
	\sum_{r=1}^M
	\sum_{n=\lceil (r+1)/2\rceil}^{\lfloor 2(r+1)/3 \rfloor}
	(-u)^{n} P_{r+3,M}(u^2)
	\sum_{\mathcal{L} : \, \vert\mathcal{L}\vert = n + 1, \ell_n = r}
	\mathtt{H}[\mathcal{L}].
\end{align}
For the Fendley model, an induced path $\mathcal{L}=\chi\md\ell_1\md\cdots\md\ell_n$ satisfies $\ell_{j+1}-\ell_j\in\{1,2\}$, and two consecutive increments cannot both equal $1$.

In \eqref{eq:Xi-FFD}, the polynomials $P_{i,j}(x)$ are the independence polynomials of the interval frustration graph induced by the vertices $i,i+1,\ldots,j$ when $1\le i\le j\le M$.
We set $P_{i,j}(x)=1$ whenever $i>j$.
They satisfy recursion relations starting from both ends of the chain:
\begin{align}
	P_{i,j}(x)
	 & =
	P_{i+1,j}(x)-x b_i^2 P_{i+3,j}(x),
	 &
	P_{i,j}(x)
	 & =
	P_{i,j-1}(x) - x b_j^2 P_{i,j-3}(x).
\end{align}
The independence polynomial of the total chain is then $P_M(u^2) \equiv P_{1,M}(u^2)$.
Then the fermion modes are defined as
\begin{align}
	\Psi_k & = \frac{\Xi_M(u_k)}{\NN_k} & \NN_k & = 2\sqrt{-u_k^2 P_{2,M}(u_k^2) P_M'(u_k^2)}.
\end{align}

For homogeneous couplings $b_k=1$, the interval polynomials depend only on the interval length.
In particular, $P_{m,M}(x)=P_{M-m+1}(x)$, the independence polynomial of a shorter open chain.
The first few terms in \eqref{eq:Xi-FFD} then read as
\begin{align}\label{eq:Xi-FFD-expanded}
	\Xi_M(u)
	 & = P_{M-1}(u^2)\, \mathtt{H}[\chi] - u\, P_{M-3}(u^2)\, \mathtt{H}[\chi \md 1]
	\nonumber                                                                        \\
	 & \quad
	+ u^2\, \big[
		P_{M-4}(u^2)\, \mathtt{H}[\chi \md 1 \md 2]
		+ P_{M-5}(u^2)\, \mathtt{H}[\chi \md 1 \md 3]
		\big]
	\nonumber                                                                        \\
	 & \quad
	- u^3\, \big[
		P_{M-6}(u^2)\, \big(\mathtt{H}[\chi \md 1 \md 2 \md 4]
		+ \mathtt{H}[\chi \md 1 \md 3 \md 4]\big)
		+ P_{M-7}(u^2)\, \mathtt{H}[\chi \md 1 \md 3 \md 5]
		\big]
	\nonumber                                                                        \\
	 & \quad
	+ u^4\, \big[
		P_{M-7}(u^2)\, \mathtt{H}[\chi \md 1 \md 2 \md 4 \md 5]
	\nonumber                                                                        \\
	 & \qquad
		+ P_{M-8}(u^2)\, \big(\mathtt{H}[\chi \md 1 \md 2 \md 4 \md 6]
		+ \mathtt{H}[\chi \md 1 \md 3 \md 4 \md 6]
		+ \mathtt{H}[\chi \md 1 \md 3 \md 5 \md 6]\big)
	\nonumber                                                                        \\
	 & \qquad
		+ P_{M-9}(u^2)\, \mathtt{H}[\chi \md 1 \md 3 \md 5 \md 7]
		\big] + \cdots\,.
\end{align}
No $P_{M-2}$ term appears because the only range-zero path is the minimal rooted path $(\chi)$.

\subsection{The coefficients}
\label{subsec:ffd-coefficients}

For the asymptotic estimate, we pass to the homogeneous half-infinite limit $M\to\infty$ with $b_j=1$, and consider long rooted paths of range $r$ with $n$ non-$\chi$ vertices, with $r/M\ll 1$ before the limit is taken.
In the rooted convention of Subsection~\ref{sec:main-result}, such a path is $\mathcal{L} = \chi \md \ell_1 \md \ell_2 \md \ldots \md \ell_n$, with $\vert \mathcal{L}\vert = n+1$ and $r = \ell_n$.

For a fermionic mode $\Psi_k$ with single-particle energy $\epsilon_k=1/u_k$, Eq.~\eqref{eq:Xi-FFD} gives the coefficient of a path with $n$ non-$\chi$ vertices and range $r$ as $(-u_k)^nP_{r+3,M}(u_k^2)/\NN_k=(-u_k)^nP_{M-r-2}(u_k^2)/\NN_k$, and only its absolute value enters the estimate below.
The energies of the fermion modes are determined by the roots $x_k$ of the independence polynomial $P_M$ of the frustration graph $P_M^2$, where $P_M(x_k) = 0$ and $x_k = u_k^2 = 1/\epsilon_k^2 > 0$.

For the Fendley model~\cite{fendley-fermions-in-disguise}, this polynomial satisfies the three-term recursion
\begin{equation}
	\label{Precursion2}
	P_m(x)=P_{m-1}(x)-x P_{m-3}(x),
\end{equation}
with $P_0(x)=P_{-1}(x)=P_{-2}(x)=1$.
Its characteristic polynomial is
\begin{equation}
	\label{alpharho}
	\alpha^{3}(x)-\alpha^{2}(x)+x=0.
\end{equation}
The discriminant of this cubic is
\begin{equation}
	\Delta = x (4 - 27 x).
\end{equation}
For $0 < x < 4/27$ we have $\Delta>0$ and three real roots.
Set $r_m(x)=P_m(x)/P_{m-1}(x)$, with $r_{-1}(x)=r_0(x)=1$.
If $r_{m-1}(x),r_{m-2}(x)>2/3$, then Eq.~\eqref{Precursion2} gives $r_m(x)=1-x/[r_{m-1}(x)r_{m-2}(x)]>1-(4/27)/(4/9)=2/3$.
Induction therefore gives $r_m(x)>2/3$ for all $m\geq1$, and hence $P_m(x)>0$ for all $m\geq0$.

Hence every positive zero of $P_m$, including the physical roots $x_k$ of $P_M$, lies in $x\geq4/27$.
For $x>4/27$, the characteristic cubic has a negative real root $\alpha_{0}(x)$ and a complex conjugate pair $\alpha_{+}(x)=\alpha_{-}^{*}(x)$, which we parametrize as
\begin{equation}
	\alpha_{\pm}= \vert \alpha \vert e^{\pm i p/3},
\end{equation}
where the phase angle corresponds to the momentum variable $p$ of \cite{fendley-fermions-in-disguise}.

Setting $B\equiv2\cos(p/3)$, comparison of the coefficients of the characteristic cubic gives
\begin{align}\label{eq:cubic-root-relations}
	\alpha_0 + \vert \alpha \vert B & = 1 & \vert \alpha \vert (\alpha_0  B  + \vert \alpha \vert) & = 0 & - \alpha_0 \vert \alpha \vert^2 & = x,
\end{align}
which in turn allows us to express the roots as
\begin{align}
	x & =\frac{B^{2}}{(B^{2}-1)^{3}}, & \vert \alpha \vert & =\frac{B}{B^{2}-1}, & \alpha_{0} & =-\frac{1}{B^{2}-1}.\label{eq:param}
\end{align}
The physical branch has $0\leq p<\pi$, equivalently $1<B\leq2$; the endpoint $p=0$ gives $x=4/27$, while $p\to\pi$ gives $x\to\infty$.

The recurrence and initial values give the generating function
\begin{equation}
	\sum_{m=0}^{\infty} P_m(x) t^m = \frac{1-x t -x t^2}{1-t+x t^3}.
\end{equation}
For $x>4/27$, the roots are distinct, and the factorization $1-t+x t^3=(1-\alpha_0 t)(1-\alpha_+t)(1-\alpha_-t)$ gives $P_m(x)=\sum_{a\in\{0,+,-\}}\alpha_a^{m+3}/(3\alpha_a-2)$ by partial fractions.
The value at $x=4/27$ is obtained by continuity.

Using the parametrization \eqref{eq:param}, this becomes
\begin{equation}
	P_{m}(x) = \frac{\alpha_{0}^{m+2} +B \vert \alpha \vert^{m+2} \bigg( 3\cos{\frac{p}{3}} \frac{\sin\left(\frac{m+3}{3}p \right)}{\sin(p/3)} + \cos\left(\frac{m+3}{3}p\right)\bigg)}{1+2 B^2}. \label{eq:polycheb}
\end{equation}
The quotient $\sin((m+3)p/3)/\sin(p/3)$ is understood by continuity at $p=0$, where it is $O(m)$.
Equation~\eqref{eq:param} gives $\vert\alpha\vert=B(-\alpha_0)\geq-\alpha_0$, so the $\vert\alpha\vert^{m+2}$ term in Eq.~\eqref{eq:polycheb} determines whether $P_m(x)$ grows or decays exponentially for fixed $x$.
The sequence $P_m(x)$ decays exponentially when $\vert\alpha\vert<1$ ($B>(1+\sqrt{5})/2$ or $p<3\pi/5$), is bounded and oscillatory when $\vert\alpha\vert=1$, and grows exponentially when $\vert\alpha\vert>1$.

Returning to $x=u^2$, write $\vert\alpha(u)\vert$ for $\vert\alpha\vert$ evaluated at $x=u^2$.
The dominant factor in Eq.~\eqref{eq:polycheb} is $\vert\alpha(u)\vert^{M-r}=\vert\alpha(u)\vert^M\vert\alpha(u)\vert^{-r}$.
We absorb the first factor into the normalization and retain the range-dependent factor $\vert\alpha(u)\vert^{-r}$.
The remaining trigonometric factor does not change the exponential growth or decay with $r$.

\subsection{Path counting and entropy}
\label{subsec:ffd-free-energy}

The coefficient estimate above is not the full operator weight, because the number of admissible paths also grows exponentially as the range is increased with $n/r$ fixed.
We therefore analyze the summation in Eq.~\eqref{eq:Xi-FFD} statistically, treating the logarithm of the path count as an entropy for this counting problem, which is possible because the coefficients of the paths depend only on $r$ and $n$.

For a path of range $r$ with $n$ non-$\chi$ vertices specified by $(\ell_1,\ell_2,\dots,\ell_n)$, set the path-vertex density $\rho=n/r$.
Its increment sequence is $(s_1,s_2,\dots,s_{n-1})$, where $s_k=\ell_{k+1}-\ell_k$.
Every increment is $1$ or $2$, but every internal $1$ is preceded and followed by a $2$.
Up to boundary corrections that are subexponential in the range, such an increment sequence can be constructed by concatenating the strings
\begin{equation}
	2\quad \text{ and } \quad 12.
\end{equation}
A trailing 1 is permitted; it shifts $n$ and $r$ by $O(1)$ and is neglected here.

If we have $ar$ insertions of 2 and $br$ insertions of 12, then the number of non-$\chi$ vertices and the range are expressed as
\begin{equation}
	n=(a+2b)r,\qquad 1=2a+3b.
\end{equation}
This yields
\begin{equation}
	a=2-3\rho, \qquad b=2\rho-1.
\end{equation}
The nonnegativity condition $a,b\geq 0$ therefore constrains the density to the interval $1/2\leq \rho\leq 2/3$.

For a fixed $a$ and $b$ the number of paths is
\begin{equation}
	\Omega=\binom{(a+b)r}{ar}.
\end{equation}
Taking the logarithm and using Stirling's formula, we obtain the leading term
\begin{equation}
	s=\log(\Omega)=  r\left[ (a+b)\log(a+b)-a\log(a)-b\log(b)\right].
\end{equation}
Substituting the relations above gives
\begin{equation}
	\begin{split}
		s= r [ & (1-\rho)\log(1-\rho) - (2-3\rho)\log(2-3\rho)-(2\rho-1)\log(2\rho-1)].
	\end{split}
\end{equation}

\subsection{Free energy and operator weight}
\label{subsec:ffd-operator-weight}

Combining the prefactor and the entropy, we obtain the ``free energy'' for the paths as
\begin{equation}\label{upower}
	e^{-f}= \left|u^{\rho r}\vert\alpha(u)\vert^{-r}\right|^2 e^s.
\end{equation}
The squared modulus appears because the operator weight is the sum of squared absolute coefficients of the path-product terms with fixed $(r,n)$.

Taking the logarithm of Eq.~\eqref{upower} gives
\begin{equation}
	\begin{split}
		-f= r\Big( & 2\rho\log\vert u\vert -2\log(|\alpha|) + [ (1-\rho)\log(1-\rho) - (2-3\rho)\log(2-3\rho)-(2\rho-1)\log(2\rho-1)]\Big).
	\end{split}
\end{equation}

Extremizing over $\rho$ determines the saddle-point value $\rho_*$ by
\begin{equation}
	\label{urho}
	u^2\frac{(2-3\rho_*)^3}{(1-\rho_*)(2\rho_*-1)^2}=1.
\end{equation}
For fixed spectral parameter $u$, this saddle-point density is independent of the range $r$ at the leading exponential order.
At range $r$, the dominant paths have $n\simeq \rho_* r$ non-$\chi$ vertices.

At the saddle point $\rho=\rho_*$, the exponent becomes
\begin{equation}
	\label{fujra}
	\begin{split}
		-f_*= r\Big( & -2\log(|\alpha|) + [ \log(1-\rho_*) - 2\log(2-3\rho_*)+\log(2\rho_*-1)]\Big).
	\end{split}
\end{equation}
Introduce $\beta=(1-2\rho_*)/(2-3\rho_*)$.
Since $1-\beta=(1-\rho_*)/(2-3\rho_*)$, the saddle-point equation \eqref{urho} gives $u^2=\beta^2(1-\beta)$.
Thus $\beta$ satisfies the characteristic equation \eqref{alpharho} with $x=u^2$.
In the relevant range $1/2<\rho_*<2/3$, this $\beta$ is negative, and hence it is the real negative root $\alpha_0$ introduced above:
\begin{equation}\label{eq:rhoalpha0}
	\alpha_0 = \frac{1-2\rho_*}{2-3\rho_*}.
\end{equation}
With this identification, the saddle-point equation reduces to $u^2 = \alpha_0^2 (1 - \alpha_0)$.
The logarithmic factors on the right-hand side of \eqref{fujra} may be collected into a single log, whose argument may be expressed in terms of $\alpha_0$ as
\begin{equation}\label{eq:rhoz}
	\frac{(1-\rho_*)(2\rho_*-1)}{(2-3\rho_*)^2 \vert\alpha\vert^2}  =  -\frac{\alpha_0 (1-\alpha_0)}{\vert\alpha\vert^2}
	= 1.
\end{equation}
Here the final equality follows by comparing the saddle relation $u^2 = \alpha_0^2(1-\alpha_0)$ with the root relation $u^2=-\alpha_0\vert\alpha\vert^2$ from Eq.~\eqref{eq:cubic-root-relations}.
Thus the optimized leading exponential rate is $f_*=0$.

Hence, after summing over the dominant path density at fixed range, the combined operator weight neither grows nor decays exponentially with the distance from the boundary.
Polynomial prefactors, oscillatory factors in $P_m(x)$, and finite-size normalization effects are not resolved by this saddle estimate.

Equation~\eqref{eq:rhoalpha0} gives the optimized density as a function of momentum $p$:
\begin{equation}
	\rho_*(p) = \frac{4 \cos^2(p/3)+1}{8 \cos^2(p/3)+1}.
\end{equation}
This function grows monotonically from $\rho_*(0) = 5/9$ to $\rho_*(\pi) = 2/3$.
Thus a fermion mode with energy $\epsilon$ has a preferred path-vertex density $\rho_*$ in the path products $\mathtt{H}[\mathcal{L}]$.

\subsection{Local conserved charges: open inhomogeneous Fendley model}
\label{subsec:ffd-low-order-charges}
\label{subsec:ffd-open-inhomogeneous}

We specialize the local conserved-charge formula~\eqref{eq:local-charge} of Section~\ref{sec:conserved-charges} to the open Fendley frustration graph $P_M^2$ with arbitrary couplings.
In this subsection we write $\localcharge{M}{2k+1}\equiv\localcharge{P_M^2}{2k+1}$.

The induced paths through a left endpoint $j$ have successive index increments $w_a\in\{1,2\}$ subject to no two consecutive unit jumps, since $(w_a,w_{a+1})=(1,1)$ would make $j$ and $j+2$ adjacent in $P_M^2$.
On $G=P_M^2$ with arbitrary couplings $b_j$, the lowest charge is the Hamiltonian
\begin{align}
	\localcharge{M}{1}=\sum_{j=1}^{M}\hc_j.
	\label{eq:ffd-H1-example}
\end{align}
The three induced three-vertex paths through left endpoint $j$ give
\begin{align}
	\localcharge{M}{3}
	=
	\sum_{j=1}^{M-3}
	\bigl(
	\pathprodH{j\md (j{+}1)\md (j{+}3)}
	+
	\pathprodH{j\md (j{+}2)\md (j{+}3)}
	\bigr)
	+
	\sum_{j=1}^{M-4}
	\pathprodH{j\md (j{+}2)\md (j{+}4)}
	+
	\sum_{j=1}^{M}
	\indepcoeff{\complementset{(j)}}{1}\,\hc_j,
	\label{eq:ffd-H3-example}
\end{align}
where $\complementset{(j)}=P_M^2\setminus\Gamma[j]$ is the residual graph~\eqref{eq:residual-graph} and
\begin{align}
	\indepcoeff{\complementset{(j)}}{1}
	=
	-
	\sum_{\substack{1\le\ell\le M \\ |\ell-j|>2}}
	b_\ell^2.
	\label{eq:ffd-open-singleton-residual-example}
\end{align}

For $\localcharge{M}{5}$ we obtain
\begin{align}
	\localcharge{M}{5}
	 & = \sum_{j=1}^{M-8}
	\pathprodH{j\md (j{+}2)\md (j{+}4)\md (j{+}6)\md (j{+}8)}
	\nonumber                   \\
	 & \quad + \sum_{j=1}^{M-7}
	\bigl(
	\pathprodH{j\md (j{+}2)\md (j{+}4)\md (j{+}6)\md (j{+}7)}
	+ \pathprodH{j\md (j{+}2)\md (j{+}4)\md (j{+}5)\md (j{+}7)}
	\nonumber                   \\
	 & \qquad\qquad
	+ \pathprodH{j\md (j{+}2)\md (j{+}3)\md (j{+}5)\md (j{+}7)}
	+ \pathprodH{j\md (j{+}1)\md (j{+}3)\md (j{+}5)\md (j{+}7)}
	\bigr)
	\nonumber                   \\
	 & \quad + \sum_{j=1}^{M-6}
	\bigl(
	\pathprodH{j\md (j{+}2)\md (j{+}3)\md (j{+}5)\md (j{+}6)}
	+ \pathprodH{j\md (j{+}1)\md (j{+}3)\md (j{+}5)\md (j{+}6)}
	\nonumber                   \\
	 & \qquad\qquad\qquad
	+ \pathprodH{j\md (j{+}1)\md (j{+}3)\md (j{+}4)\md (j{+}6)}
	\bigr)
	\nonumber                   \\
	 & \quad + \sum_{j=1}^{M-3}
	\bigl(
	\indepcoeff{\complementset{j\md (j{+}1)\md (j{+}3)}}{1}\,
	\pathprodH{j\md (j{+}1)\md (j{+}3)}
	+ \indepcoeff{\complementset{j\md (j{+}2)\md (j{+}3)}}{1}\,
	\pathprodH{j\md (j{+}2)\md (j{+}3)}
	\bigr)
	\nonumber                   \\
	 & \quad + \sum_{j=1}^{M-4}
	\indepcoeff{\complementset{j\md (j{+}2)\md (j{+}4)}}{1}\,
	\pathprodH{j\md (j{+}2)\md (j{+}4)}
	\nonumber                   \\
	 & \quad + \sum_{j=1}^{M}
	\indepcoeff{\complementset{(j)}}{2}\,\hc_j.
	\label{eq:ffd-H5-example}
\end{align}
The residual coefficients are
\begin{align}
	\indepcoeff{\complementset{j\md (j{+}1)\md (j{+}3)}}{1}
	=
	\indepcoeff{\complementset{j\md (j{+}2)\md (j{+}3)}}{1}
	 & =
	-
	\sum_{\substack{1\le\ell\le M   \\ \ell\le j-3\ \mathrm{or}\ \ell\ge j+6}}
	b_\ell^2,
	\nonumber                       \\
	\indepcoeff{\complementset{j\md (j{+}2)\md (j{+}4)}}{1}
	 & =
	-
	\sum_{\substack{1\le\ell\le M   \\ \ell\le j-3\ \mathrm{or}\ \ell\ge j+7}}
	b_\ell^2,
	\nonumber                       \\
	\indepcoeff{\complementset{(j)}}{2}
	 & =
	\sum_{\substack{1\le\ell<r\le M \\ |\ell-j|>2,\ |r-j|>2\\ |r-\ell|>2}}
	b_\ell^2 b_r^2.
	\label{eq:ffd-H5-residual-example}
\end{align}
These residual formulas are Eq.~\eqref{eq:indep-coeff} evaluated on the residual graph $P_M^2\setminus\Gamma[\pa]$.
To check the eight five-vertex families in the first three sums, list each in increasing order as $\pa=i_1\md i_2\md i_3\md i_4\md i_5$ with $i_1=j$, and set $d_a=i_{a+1}-i_a$ for $a=1,\ldots,4$.
In $P_M^2$ each $d_a$ is either $1$ or $2$, and two consecutive unit steps would make the first and third vertices adjacent, producing a chord.
Thus the allowed increment sequences are
\begin{align}
	(d_1,d_2,d_3,d_4)
	\in
	\{2222,2221,2212,2122,1222,2121,1221,1212\}.
\end{align}
The upper limit of the corresponding $j$-sum is the condition $j+d_1+\cdots+d_4\le M$.

\begin{figure}[t]
	\centering
	\begin{tikzpicture}[
			scale=1,
			line cap=round,
			line join=round,
			vertex/.style={circle, draw, thick, fill=white, minimum size=\vertexsize, inner sep=0pt},
			selected/.style={circle, draw=Red, thick, fill=Red!15, minimum size=9pt, inner sep=0pt},
			graphline/.style={thick, gray!45},
			pathedge/.style={opacity=1, line width=1pt, draw=Red}
		]
		\def\N{12}
		\def\xstep{0.62}
		\def\yodd{0.54}
		\def\yeven{-0.54}
		\pgfmathtruncatemacro{\NmOne}{\N-1}
		\pgfmathtruncatemacro{\NmTwo}{\N-2}
		\node at ({-\xstep},0) {$\cdots$};
		\node at ({\xstep*\N+\xstep},0) {$\cdots$};
		\foreach \x/\lab in {0/{$j{-}2$},1/{$j{-}1$},2/{$j$},3/{$j{+}1$},4/{$j{+}2$},5/{$j{+}3$},6/{$j{+}4$},7/{$j{+}5$},8/{$j{+}6$},9/{$j{+}7$},10/{$j{+}8$},11/{$j{+}9$},12/{$j{+}10$}} {
		\ifodd\x\relax
			\coordinate (v\x) at ({\xstep*\x},\yodd);
		\else
			\coordinate (v\x) at ({\xstep*\x},\yeven);
		\fi
		}
		\foreach \x in {0,...,\NmOne} {
				\pgfmathtruncatemacro{\y}{\x+1}
				\draw[graphline] (v\x) -- (v\y);
			}
		\foreach \x in {0,...,\NmTwo} {
				\pgfmathtruncatemacro{\y}{\x+2}
				\draw[graphline] (v\x) -- (v\y);
			}
		\draw[pathedge] (v2) -- (v3) -- (v5) -- (v7) -- (v8);
		\foreach \x/\lab in {0/{$j{-}2$},1/{$j{-}1$},4/{$j{+}2$},6/{$j{+}4$},9/{$j{+}7$},10/{$j{+}8$},11/{$j{+}9$},12/{$j{+}10$}} {
		\ifodd\x\relax
			\node[vertex, label=above:\lab] at (v\x) {};
		\else
			\node[vertex, label=below:\lab] at (v\x) {};
		\fi
		}
		\foreach \x/\lab in {2/{$j$},3/{$j{+}1$},5/{$j{+}3$},7/{$j{+}5$},8/{$j{+}6$}} {
		\ifodd\x\relax
			\node[selected, label=above:\lab] at (v\x) {};
		\else
			\node[selected, label=below:\lab] at (v\x) {};
		\fi
		}
	\end{tikzpicture}
	\caption{The five-vertex induced path $j\md(j{+}1)\md(j{+}3)\md(j{+}5)\md(j{+}6)$ in the open Fendley frustration graph $P_M^2$.}
	\label{fig:ffd-open-local-five-path}
\end{figure}

The open graph $P_M^2$ is ECF, so $\localcharge{M}{2k+1}$ commutes with the corresponding Hamiltonian of the Fendley model for arbitrary couplings by Corollary~\ref{cor:local-charge-conservation}.
If the open chain is further specialized to $b_j=1$, boundary paths still have shorter residual graphs and do not reduce to a single translation-invariant coefficient.

\paragraph{Open generalized conserved-charge example.}
The lowest-order open-chain mixed case, with $1<c<m$, in which both packed components can be non-singleton odd paths is $\packingcharge{P_M^2}{6}{2}$.
By Eq.~\eqref{eq:path-packing-charge},
\begin{align}
	\packingcharge{P_M^2}{6}{2}
	=
	\sum_{\PP\in\oddpathsetlengthcomp{P_M^2}{6}{2}}
	\pathprodH{\PP}
	+
	\sum_{\PP\in\oddpathsetlengthcomp{P_M^2}{4}{2}}
	\indepcoeff{\complementset{\PP}}{1}\,
	\pathprodH{\PP}
	+
	\sum_{\PP\in\oddpathsetlengthcomp{P_M^2}{2}{2}}
	\indepcoeff{\complementset{\PP}}{2}\,
	\pathprodH{\PP}.
	\label{eq:ffd-first-mixed-packing-example}
\end{align}
The first sum contains the packings with component sizes $(3,3)$ and $(1,5)$; the residual terms contain the smaller packings $(1,3)$ and $(1,1)$.
For example, for $1\le j\le M-10$ the $(3,3)$ packing
\begin{align}
	\pathprodH{(j\md (j{+}1)\md (j{+}3))\sqcup((j{+}6)\md (j{+}8)\md (j{+}10))}
	\label{eq:ffd-open-packing-representative}
\end{align}
appears in the first sum; see Figure~\ref{fig:ffd-open-packing-33}.
Both components are induced paths in $P_M^2$, and their nearest vertices are separated by distance three in the underlying path $P_M$, so no edge of $P_M^2$ lies between the two components.

\begin{figure}[t]
	\centering
	\begin{tikzpicture}[
			scale=1,
			line cap=round,
			line join=round,
			vertex/.style={circle, draw, thick, fill=white, minimum size=\vertexsize, inner sep=0pt},
			packa/.style={circle, draw=Red, thick, fill=Red!15, minimum size=9pt, inner sep=0pt},
			packb/.style={circle, draw=NavyBlue, thick, fill=NavyBlue!15, minimum size=9pt, inner sep=0pt},
			graphline/.style={thick, gray!45},
			patha/.style={opacity=1, line width=1pt, draw=Red},
			pathb/.style={opacity=1, line width=1pt, draw=NavyBlue}
		]
		\def\N{18}
		\def\xstep{0.62}
		\def\yodd{0.54}
		\def\yeven{-0.54}
		\pgfmathtruncatemacro{\NmOne}{\N-1}
		\pgfmathtruncatemacro{\NmTwo}{\N-2}
		\node at ({-\xstep},0) {$\cdots$};
		\node at ({\xstep*\N+\xstep},0) {$\cdots$};
		\foreach \x/\lab in {0/{$j{-}4$},1/{$j{-}3$},2/{$j{-}2$},3/{$j{-}1$},4/{$j$},5/{$j{+}1$},6/{$j{+}2$},7/{$j{+}3$},8/{$j{+}4$},9/{$j{+}5$},10/{$j{+}6$},11/{$j{+}7$},12/{$j{+}8$},13/{$j{+}9$},14/{$j{+}10$},15/{$j{+}11$},16/{$j{+}12$},17/{$j{+}13$},18/{$j{+}14$}} {
		\ifodd\x\relax
			\coordinate (p\x) at ({\xstep*\x},\yodd);
		\else
			\coordinate (p\x) at ({\xstep*\x},\yeven);
		\fi
		}
		\foreach \x in {0,...,\NmOne} {
				\pgfmathtruncatemacro{\y}{\x+1}
				\draw[graphline] (p\x) -- (p\y);
			}
		\foreach \x in {0,...,\NmTwo} {
				\pgfmathtruncatemacro{\y}{\x+2}
				\draw[graphline] (p\x) -- (p\y);
			}
		\draw[patha] (p4) -- (p5) -- (p7);
		\draw[pathb] (p10) -- (p12) -- (p14);
		\foreach \x/\lab in {0/{$j{-}4$},1/{$j{-}3$},2/{$j{-}2$},3/{$j{-}1$},6/{$j{+}2$},8/{$j{+}4$},9/{$j{+}5$},11/{$j{+}7$},13/{$j{+}9$},15/{$j{+}11$},16/{$j{+}12$},17/{$j{+}13$},18/{$j{+}14$}} {
		\ifodd\x\relax
			\node[vertex, label=above:\lab] at (p\x) {};
		\else
			\node[vertex, label=below:\lab] at (p\x) {};
		\fi
		}
		\foreach \x/\lab in {4/{$j$},5/{$j{+}1$},7/{$j{+}3$}} {
		\ifodd\x\relax
			\node[packa, label=above:\lab] at (p\x) {};
		\else
			\node[packa, label=below:\lab] at (p\x) {};
		\fi
		}
		\foreach \x/\lab in {10/{$j{+}6$},12/{$j{+}8$},14/{$j{+}10$}} {
		\node[packb, label=below:\lab] at (p\x) {};
		}
	\end{tikzpicture}
	\caption{A representative $(3,3)$ term in $\packingcharge{P_M^2}{6}{2}$; the red and blue paths are the two components in Eq.~\eqref{eq:ffd-open-packing-representative}.}
	\label{fig:ffd-open-packing-33}
\end{figure}

\subsection{Local conserved charges: homogeneous periodic Fendley model}
\label{subsec:ffd-uniform-periodic}

On the periodic graph $C_M^2$ with homogeneous couplings $b_j=1$, we write $H_{\mathrm{hom,per}}\equiv\left.H_{C_M^2}\right|_{b_j=1}$ for the Hamiltonian.
Translation invariance forces each path-product to enter the conserved charges with the same coefficient as its translates.

In the remainder of this subsection and in Subsection~\ref{subsec:ffd-catalan-tree}, we write $\pathprod{\ell_1\md\cdots\md\ell_w}\equiv h_{\ell_1}\cdots h_{\ell_w}$ for the corresponding bare homogeneous path product.

\begin{figure}[t]
	\centering
	\begin{tikzpicture}[
			scale=1,
			line cap=round,
			line join=round,
			vertex/.style={circle, draw, thick, fill=white, minimum size=\vertexsize, inner sep=0pt},
			ringline/.style={thick, gray!55},
			zigzagline/.style={thick, RoyalBlue!75!black}
		]
		\def\Router{2.35}
		\def\Rinner{1.48}
		\def\angstep{45}
		\def\angshift{22.5}
		\foreach \p in {0,...,7} {
				\coordinate (o\p) at ({90-\angstep*\p}:\Router);
				\coordinate (e\p) at ({90-\angstep*\p-\angshift}:\Rinner);
			}
		\draw[ringline] (0,0) circle [radius=\Router];
		\draw[ringline] (0,0) circle [radius=\Rinner];
		\draw[zigzagline] (o0) -- (e0) -- (o1) -- (e1) -- (o2) -- (e2) -- (o3) -- (e3) -- (o4) -- (e4) -- (o5) -- (e5) -- (o6) -- (e6) -- (o7) -- (e7) -- cycle;
		\foreach \p in {0,...,7} {
				\pgfmathtruncatemacro{\idxodd}{2*\p+1}
				\pgfmathtruncatemacro{\idxeven}{2*\p+2}
				\node[vertex] at (o\p) {};
				\node[vertex] at (e\p) {};
				\node[font=\scriptsize] at ($(0,0)!1.14!(o\p)$) {$\idxodd$};
				\node[font=\footnotesize] at ($(0,0)!0.78!(e\p)$) {$\idxeven$};
			}
	\end{tikzpicture}
	\caption{The periodic Fendley frustration graph $C_{16}^2$.
		Distance-two edges form the outer odd and inner even cycles, while nearest-neighbor edges form the blue zigzag between them.}
	\label{fig:ffd-periodic-frustration-graph}
\end{figure}

\begin{lem}\label{lem:CM2-geometry}
	For $M>6$, the smallest induced even hole in $C_M^2$ has size $2\Big\lceil\frac{M}{4}\Big\rceil$.
	Moreover, $C_M^2$ is even-bubble-wand-free.
	The cyclic orientation $i\to i+1$ and $i\to i+2$ of all edges (indices modulo $M$) is an induced-path orientation.
\end{lem}
\begin{proof}
	Let $\mathcal{C}$ be an induced $q$-hole of $C_M^2$ with vertices $v_1,\dots,v_q$, listed in cyclic order on the underlying cycle $C_M$.
	This is also their order on $\mathcal{C}$, since otherwise $\mathcal{C}$ would have a chord.
	Let $d_i$ be the cyclic distance from $v_i$ to $v_{i+1}$ on $C_M$ (indices modulo $q$).
	Each $d_i$ lies in $\{1,2\}$, and two consecutive gaps cannot both equal $1$, since then $v_i$ and $v_{i+2}$ would form a chord.
	Hence $M=\sum_i d_i\le 2q$, so an even $q$ satisfies $q\ge 2\lceil M/4\rceil$.
	To attain this bound, set $q=2\lceil M/4\rceil$ and $t=2q-M$.
	For $M>6$, we have $0\le t\le\lfloor q/2\rfloor$, so there is a cyclic sequence of $q$ gaps containing $t$ isolated entries equal to $1$ and $q-t$ entries equal to $2$.
	The gaps sum to $M$ and therefore specify a $q$-cycle in $C_M^2$.
	Any two nonconsecutive vertices of this cycle have cyclic distance at least $3$, because each arc between them contains at least two gaps and no two gaps equal to $1$ are consecutive.
	The cycle is therefore induced and has size $q=2\lceil M/4\rceil$.

	Suppose that $C_M^2$ contains an even bubble wand, and let $w$ be the endpoint of its path part adjacent to two consecutive hole vertices $a,b$ and $u$ its predecessor.
	It suffices to show that every neighbor of $w$ is adjacent to a hole vertex, since Definition~\ref{def:even-bubble-wand} requires $u$ to be non-adjacent to the hole.
	The gap from $a$ to $b$ on $C_M$ is either $1$ or $2$.
	If it is $2$, the intermediate vertex $x$ is the unique common neighbor of $a$ and $b$, so $w=x$.
	Every neighbor of $x$ lies in $\{a,b,a-1,b+1\}$ and is therefore adjacent to $a$ or $b$.
	If the gap is $1$, the only common neighbors of $a$ and $b$ are $a-1$ and $b+1$; by symmetry take $w=a-1$.
	Let $c$ be the hole vertex preceding $a$ on $C_M$.
	Since $w$ is not on the hole, the gap from $c$ to $a$ is $2$, and hence $c=a-2$.
	The only neighbor of $w$ not already adjacent to $a$ or $b$ is $a-3$, which is adjacent to $c$.
	Thus every neighbor of $w$ is adjacent to the hole in either case, contradicting the defining condition on $u$.
	Therefore $C_M^2$ is even-bubble-wand-free.

	Finally, identify the vertices with $\mathbb{Z}/M\mathbb{Z}$ and orient $i\to i+1$ and $i\to i+2$.
	If $i\md j\md\ell$ is an induced path, then $i$ and $\ell$ lie on opposite sides of $j$; otherwise they would be adjacent in $C_M^2$.
	Hence either $j-i,\ell-j\in\{1,2\}$ or $i-j,j-\ell\in\{1,2\}$.
	In both cases the two edges point consistently along the path, which proves Eq.~\eqref{eq:induced-path-orientation-sign}.
\end{proof}

On $C_M^2$, the odd local conserved charges $\localcharge{C_M^2}{2r+1}$ are given by Eq.~\eqref{eq:local-charge} with indices read cyclically with period $M$ and $1\le j\le M$.
The graph $C_M^2$ is claw-free and, by Lemma~\ref{lem:CM2-geometry}, even-bubble-wand-free.
Corollary~\ref{cor:local-charge-ebwf} therefore shows that $\localcharge{C_M^2}{2r+1}$ commutes with $H_{C_M^2}$ for arbitrary $b_j$.

For $b_j=1$, orient every edge clockwise.
By Lemma~\ref{lem:CM2-geometry}, this is an induced-path orientation of $C_M^2$, so the even path operators $\mathcal{E}_{C_M^2,\omega}^{(2r)}$ of Definition~\ref{def:even-path-operator} are defined.
In Theorem~\ref{thm:even-path-commutator}, the parameter $K$ is half the smallest induced even-hole size, not the even-bubble-wand parameter used for the odd charges.
Lemma~\ref{lem:CM2-geometry} gives $K=\lceil M/4\rceil$, and $M\ge 4r+3$ implies $r<\lceil M/4\rceil=K$.

By the cyclic symmetry of $C_M^2$, the incoming and outgoing residual graphs at $j$ agree pairwise after a shift of vertex labels.
Since $b_i^2=1$, the corresponding singleton terms in Eq.~\eqref{eq:even-path-singleton-coeff} cancel, and hence $B_{j,\omega}^{(s)}=0$.
Hence $\mathcal{E}_{C_M^2,\omega}^{(2r)}$ commutes with $H_{\mathrm{hom,per}}$ in the range of Theorem~\ref{thm:even-path-commutator}.

We combine the odd family and these homogeneous even operators into the homogeneous periodic operator
\begin{align}
	\mathcal{H}_{\mathrm{hom,per}}^{(k)}
	 & \equiv
	\begin{cases}
		\left.\localcharge{C_M^2}{k-2}\right|_{b_j=1},           & k \text{ odd},  \\
		\left.\mathcal{E}_{C_M^2,\omega}^{(k-2)}\right|_{b_j=1}, & k \text{ even},
	\end{cases}
	\label{eq:ffd-homogeneous-periodic-charge-definition}
\end{align}
where in the even case $\omega$ is the clockwise orientation fixed above.

For $M\ge2k-1$, $\mathcal{H}_{\mathrm{hom,per}}^{(k)}$ is a local conserved charge with spin-chain support $2k-3$ in the defining representation $h_j=Z_jZ_{j+1}X_{j+2}$.
Indeed, each leading monomial is the product of $k-2$ generators $h_j$ along an induced path in $C_M^2$ all of whose steps have size two.
The leftmost and rightmost three-site supports are separated by $2(k-3)$ sites, so their union covers $2(k-3)+3=2k-3$ consecutive sites.
In particular, $\mathcal{H}_{\mathrm{hom,per}}^{(3)}=H_{\mathrm{hom,per}}$ and has three-site spin-chain support.
The shift by two aligns the superscript $k$ with the indexing of the Catalan basis of Subsection~\ref{subsec:ffd-catalan-tree}.

\begin{define}[Periodic building blocks]
	For any $s\ge 3$ and $m\ge 0$, define
	\begin{align}
		F_{s,m}
		 & \equiv
		\sum_{j=1}^{M}
		\sum_{\substack{
		d_a\in\{1,2\}\ (1\le a\le s+m-3) \\
		\#\{a:d_a=1\}=m                  \\
				(d_a,d_{a+1})\ne(1,1)\ (1\le a<s+m-3)
			}}
		\pathprod{j\md (j{+}d_1)\md \cdots \md (j{+}d_1+\cdots+d_{s+m-3})}.
		\label{eq:ffd-periodic-Fsm-definition}
	\end{align}
	All indices are read modulo $M$.
	Before reduction modulo $M$, each monomial in $F_{s,m}$ contains $s+m-2$ generators and occupies an interval of $2s+m-3$ sites in the defining representation $h_j=Z_jZ_{j+1}X_{j+2}$; if these sites remain distinct modulo $M$, this interval is its spin-chain support.
	For $s=3$ and $m=0$ the increment sequence is empty, so $F_{3,0}=\sum_{j=1}^{M}h_j$.
	We also set $F_{s,m}=0$ for $s<3$.
\end{define}

If $\pa$ is one of the paths summed in $F_{s,m}$ and its defining interval does not wrap around the periodic chain, then the residual graph $\complementset{\pa}$ is the square of a path on $M-2s-m+1$ vertices.
Therefore the degree-$n$ coefficient of the residual independence polynomial appearing in the homogeneous specializations of the odd local conserved charge~\eqref{eq:local-charge} and the oriented even path operator~\eqref{eq:even-path-charge} is
\begin{align}
	\indepcoeff{\complementset{\pa}}{n}
	 & =
	(-1)^n\binom{M-2s-m-2n+3}{n}.
	\label{eq:ffd-uniform-residual-binomial}
\end{align}

Combining the appropriate charge definition with Eq.~\eqref{eq:ffd-uniform-residual-binomial} gives the following closed expression for every $k\ge 3$ and $M\ge 2k-1$:
\begin{align}
	\mathcal{H}_{\mathrm{hom,per}}^{(k)}
	 & =
	\sum_{\substack{0 \le m+n < \lfloor k/2\rfloor \\ m,n\ge 0}}
	(-1)^n
	\binom{M-2k+2n+m+3}{n}
	F_{k-2n-m,m}.
	\label{eq:ffd-uniform-general-charge-formula}
\end{align}
In both parities the leading term is $F_{k,0}$.

The first odd cases of Eq.~\eqref{eq:ffd-uniform-general-charge-formula} are
\begin{align}
	\mathcal{H}_{\mathrm{hom,per}}^{(3)}
	 & =
	F_{3,0},
	\nonumber \\
	\mathcal{H}_{\mathrm{hom,per}}^{(5)}
	 & =
	F_{5,0}
	+
	F_{4,1}
	-
	(M-5)F_{3,0},
	\nonumber \\
	\mathcal{H}_{\mathrm{hom,per}}^{(7)}
	 & =
	F_{7,0}
	+
	F_{6,1}
	+
	F_{5,2}
	-
	(M-9)F_{5,0}
	-
	(M-8)F_{4,1}
	+
	\binom{M-7}{2}F_{3,0}.
	\label{eq:ffd-uniform-low-odd-examples}
\end{align}
The first even cases are
\begin{align}
	\mathcal{H}_{\mathrm{hom,per}}^{(4)}
	 & =
	F_{4,0}
	+
	F_{3,1},
	\nonumber \\
	\mathcal{H}_{\mathrm{hom,per}}^{(6)}
	 & =
	F_{6,0}
	+
	F_{5,1}
	+
	F_{4,2}
	-
	(M-7)F_{4,0}
	-
	(M-6)F_{3,1}.
	\label{eq:ffd-uniform-low-even-examples}
\end{align}

\subsection{Catalan-tree pattern in the homogeneous periodic Fendley model}
\label{subsec:ffd-catalan-tree}

A triangular recombination of the homogeneous periodic Fendley charges yields a conserved family with the same Catalan-tree coefficients as the local conserved charges of the spin-$1/2$ XXX chain~\cite{anshelevich-heisenberg-first-integrals,GM-higher-conserved-XXX,GM-higher-conserved-XXZ,GM-catalan-tree}.
Throughout this subsection we restrict to $G=C_M^2$ with $b_j=1$.

For $a,b\ge 0$, set
\begin{align}
	\mathcal{C}_{a,b}
	 & \equiv
	\binom{a+b}{b}
	-
	\binom{a+b}{b-1},
	\label{eq:ffd-generalized-catalan-number}
\end{align}
with the convention $\binom{a}{-1}=0$.
In the triangular range $0\le b\le a$, these are the Catalan-triangle numbers~\cite{hilton-pedersen-catalan}.
Equation~\eqref{eq:ffd-generalized-catalan-number} gives $\mathcal{C}_{a,a+1}=0$ for $a\ge 0$, and we set $\mathcal{C}_{a,-1}\equiv 0$ for $a\ge 0$ and $\mathcal{C}_{-1,0}\equiv 1$ so that the recursion below includes the corner $(a,b)=(0,0)$.
The coefficients satisfy the following recursion:
\begin{align}
	\mathcal{C}_{a,b}
	 & =
	\mathcal{C}_{a-1,b}
	+
	\mathcal{C}_{a,b-1}
	,
	\label{eq:ffd-catalan-triangle-recursion}
\end{align}
for every $a,b$ with $0\le b\le a$, including the boundaries.
The diagonal entries recover the ordinary Catalan numbers:
\begin{align}
	\mathcal{C}_{n,n}
	 & =
	\binom{2n}{n}
	-
	\binom{2n}{n-1}
	=
	\frac{1}{n+1}\binom{2n}{n}.
	\label{eq:ffd-catalan-diagonal}
\end{align}

\begin{define}[Catalan charge]
	For every $k\ge 3$, define the Catalan charge of the homogeneous periodic Fendley model by
	\begin{align}
		\catalancharge{k}
		 & \equiv
		F_{k,0}
		+
		\sum_{\substack{0<n+m<\lfloor k/2\rfloor \\ n\ge 0,\ m\ge 1}}
		(-1)^n
		\mathcal{C}_{n+m-1,n}
		F_{k-2n-m,m}.
		\label{eq:ffd-catalan-charge-definition}
	\end{align}
	For $M\ge2k-1$, the leading term $F_{k,0}$ has spin-chain support $2k-3$ in the defining representation $h_j=Z_jZ_{j+1}X_{j+2}$, and hence so does $\catalancharge{k}$.
\end{define}

Corollary~\ref{cor:ffd-catalan-charge-conservation} establishes their conservation for $M\ge2k-1$.

The Catalan charges up to $\catalancharge{8}$ are
\begin{alignat}{2}
	\catalancharge{3}
	       & =
	F_{3,0},
	\qquad &
	\catalancharge{4}
	       & =
	F_{4,0}
	+
	F_{3,1},
	\nonumber        \\
	\catalancharge{5}
	       & =
	F_{5,0}
	+
	F_{4,1},
	\qquad &
	\catalancharge{6}
	       & =
	F_{6,0}
	+
	F_{5,1}
	+
	F_{4,2}
	-
	F_{3,1},
	\nonumber        \\
	\catalancharge{7}
	       & =
	F_{7,0}
	+
	F_{6,1}
	+
	F_{5,2}
	-
	F_{4,1},
	\qquad &
	\catalancharge{8}
	       & =
	F_{8,0}
	+
	F_{7,1}
	+
	F_{6,2}
	+
	F_{5,3}
	\nonumber        \\
	       &       &
	       & \quad
	-
	F_{5,1}
	-
	2F_{4,2}
	+
	2F_{3,1}.
	\label{eq:ffd-catalan-low-order-examples}
\end{alignat}
The coefficients $\pm 2$ in $\catalancharge{8}$ are the first nontrivial Catalan-triangle coefficients in this basis.
Thus the Catalan basis removes the extensive $M$-dependent lower-support contributions in the homogeneous charge formula~\eqref{eq:ffd-uniform-general-charge-formula} in both parity sectors.

\begin{lem}
	\label{lem:ffd-catalan-recombination}
	Let $G=C_M^2$ with $b_j=1$ and every edge oriented clockwise as in Subsection~\ref{subsec:ffd-uniform-periodic}, let $k\ge 3$, and assume $M\ge 2k-1$.
	Then the homogeneous periodic charge~\eqref{eq:ffd-uniform-general-charge-formula} satisfies
	\begin{align}
		\mathcal{H}_{\mathrm{hom,per}}^{(k)}
		 & =
		\sum_{j=0}^{\lfloor (k-3)/2\rfloor}
		(-1)^j
		\binom{M-2k+2j+3}{j}
		\catalancharge{k-2j}.
		\label{eq:ffd-catalan-triangular-relation}
	\end{align}
\end{lem}

\begin{proof}
	Set $A=M-2k+3$, so that Eq.~\eqref{eq:ffd-uniform-general-charge-formula} reads
	\begin{align}
		\mathcal{H}_{\mathrm{hom,per}}^{(k)}
		 & =
		\sum_{n=0}^{\lfloor (k-3)/2\rfloor}
		\sum_{\substack{m\ge 0 \\ m+n<\lfloor k/2\rfloor}}
		(-1)^n
		\binom{A+2n+m}{n}
		F_{k-2n-m,m}.
		\label{eq:ffd-hom-binomial-basis}
	\end{align}
	When $k$ is even, the only omitted boundary term has $F_{2,0}=0$ by convention.
	Let $\mathcal{R}_k$ denote the right-hand side of Eq.~\eqref{eq:ffd-catalan-triangular-relation}.
	Using $\mathcal{C}_{-1,0}=1$ and $\mathcal{C}_{r-1,r}=0$ for $r\ge1$, we include the leading term as the $r=m=0$ summand when substituting Eq.~\eqref{eq:ffd-catalan-charge-definition} into $\mathcal{R}_k$:
	\begin{align}
		\mathcal{R}_k
		 & =
		\sum_{j=0}^{\lfloor (k-3)/2\rfloor}
		\sum_{\substack{0\le r+m<\lfloor k/2\rfloor-j \\ r\ge 0,\ m\ge 0}}
		(-1)^{j+r}
		\binom{A+2j}{j}
		\mathcal{C}_{r+m-1,r}
		F_{k-2j-2r-m,m}
		\nonumber                                     \\
		 & =
		\sum_{n=0}^{\lfloor (k-3)/2\rfloor}
		\sum_{\substack{m\ge 0                        \\ m+n<\lfloor k/2\rfloor}}
		(-1)^n
		\left(
		\sum_{j=0}^{n}
		\binom{A+2j}{j}
		\mathcal{C}_{n-j+m-1,n-j}
		\right)
		F_{k-2n-m,m}.
		\label{eq:ffd-catalan-expanded-right-hand-side}
	\end{align}
	In the second equality we set $n=j+r$, so that $r+m<\lfloor k/2\rfloor-j$ becomes $n+m<\lfloor k/2\rfloor$.
	Comparing coefficients in Eqs.~\eqref{eq:ffd-hom-binomial-basis} and~\eqref{eq:ffd-catalan-expanded-right-hand-side}, it remains to prove
	\begin{align}
		\binom{A+2n+m}{n}
		 & =
		\sum_{j=0}^{n}
		\binom{A+2j}{j}
		\mathcal{C}_{n-j+m-1,n-j}.
		\label{eq:ffd-catalan-convolution}
	\end{align}
	To prove Eq.~\eqref{eq:ffd-catalan-convolution}, introduce the ordinary Catalan generating function
	\begin{align}
		\mathrm{Cat}(z)
		 & \equiv
		\frac{1-\sqrt{1-4z}}{2z}
		=
		\sum_{n\ge 0}
		\mathcal{C}_{n,n}z^n
		=
		1+z+2z^2+5z^3+14z^4+\cdots.
	\end{align}
	We also use the auxiliary series
	\begin{align}
		T_u(z)
		 & \equiv
		\sum_{j\ge 0}
		\binom{u+2j}{j}z^j
		=
		\frac{\mathrm{Cat}(z)^u}{\sqrt{1-4z}}.
		\label{eq:ffd-ballot-generating-function}
	\end{align}
	The second equality in Eq.~\eqref{eq:ffd-ballot-generating-function} is the $p=2$ case of the Fuss--Catalan identity obtained by Lagrange inversion~\cite{gessel-lagrange-inversion}; we verify it directly.
	Put $y(z)=\mathrm{Cat}(z)-1$.
	The Catalan equation gives $y=z(1+y)^2$, hence $z=y/(1+y)^2$ and $\sqrt{1-4z}=(1-y)/(1+y)$.
	Therefore
	\begin{align*}
		\frac{\mathrm{Cat}(z)^u}{\sqrt{1-4z}}
		 & =
		\frac{(1+y)^{u+1}}{1-y}.
	\end{align*}
	Then Cauchy's integral formula gives, for $j\ge 0$,
	\begin{align*}
		[z^j]\frac{\mathrm{Cat}(z)^u}{\sqrt{1-4z}}
		 & =
		\frac{1}{2\pi i}
		\oint_{\Gamma_z}
		\frac{(1+y(z))^{u+1}}{(1-y(z))\,z^{j+1}}
		\,\dd{z}
		\nonumber \\
		 & =
		\frac{1}{2\pi i}
		\oint_{\Gamma_y}
		\frac{(1+y)^{u+1}}{1-y}
		\cdot
		\frac{(1+y)^{2j+2}}{y^{j+1}}
		\cdot
		\frac{1-y}{(1+y)^3}
		\,\dd{y}
		\nonumber \\
		 & =
		[y^j](1+y)^{u+2j}
		=
		\binom{u+2j}{j},
	\end{align*}
	where $\Gamma_z$ is a small positively oriented loop around $z=0$ and $\Gamma_y$ is its preimage.

	We next use the identity
	\begin{align}
		\frac{T_{A+m}(z)}{T_A(z)}
		 & =
		\mathrm{Cat}(z)^m
		=
		\sum_{r\ge 0}
		\mathcal{C}_{r+m-1,r}z^r,
		\qquad m\ge0.
		\label{eq:ffd-catalan-ratio}
	\end{align}
	The first equality follows from Eq.~\eqref{eq:ffd-ballot-generating-function}, and the second is the Catalan-number convolution identity~\cite{heisenberg-mpo-charges} with the boundary values fixed above.

	Multiplying Eq.~\eqref{eq:ffd-catalan-ratio} by $T_A(z)$ and extracting the coefficient of $z^n$ gives Eq.~\eqref{eq:ffd-catalan-convolution}, and hence Eq.~\eqref{eq:ffd-catalan-triangular-relation}.
\end{proof}

\begin{cor}
	\label{cor:ffd-catalan-charge-conservation}
	For $G=C_M^2$ with $b_j=1$ and every $k\ge 3$, $\catalancharge{k}$ satisfies
	\begin{align}
		\bigl[H_{\mathrm{hom,per}},\catalancharge{k}\bigr]
		 & =
		0,
		\qquad
		M\ge 2k-1.
		\label{eq:ffd-catalan-charge-conservation}
	\end{align}
\end{cor}

\begin{proof}
	Proceed by strong induction on $k\ge3$.
	Every $\catalancharge{k-2j}$ with $j\ge 1$ on the right-hand side of Eq.~\eqref{eq:ffd-catalan-triangular-relation} commutes with $H_{\mathrm{hom,per}}$ by the induction hypothesis.
	Since the $j=0$ term is $\catalancharge{k}$ with coefficient one and $\mathcal{H}_{\mathrm{hom,per}}^{(k)}$ is conserved, subtracting the $j\ge 1$ terms gives Eq.~\eqref{eq:ffd-catalan-charge-conservation}.
\end{proof}

For example, the first few triangular decompositions are
\begin{align}
	\mathcal{H}_{\mathrm{hom,per}}^{(3)}
	 & =
	\catalancharge{3},
	 &
	\mathcal{H}_{\mathrm{hom,per}}^{(4)}
	 & =
	\catalancharge{4},
	\nonumber \\
	\mathcal{H}_{\mathrm{hom,per}}^{(5)}
	 & =
	\catalancharge{5}
	-
	(M-5)\catalancharge{3},
	 &
	\mathcal{H}_{\mathrm{hom,per}}^{(6)}
	 & =
	\catalancharge{6}
	-
	(M-7)\catalancharge{4},
	\nonumber \\
	\mathcal{H}_{\mathrm{hom,per}}^{(7)}
	 & =
	\catalancharge{7}
	-
	(M-9)\catalancharge{5}
	+
	\binom{M-7}{2}\catalancharge{3}.
	\label{eq:ffd-catalan-low-triangular-examples}
\end{align}

Alternatively, following the original conservation proofs for local conserved charges of interacting integrable chains~\cite{anshelevich-heisenberg-first-integrals,GM-higher-conserved-XXX,GM-catalan-tree,xyz-all-charges,nienhuis-xxz,hubbard-charges,fukai-doctoral-thesis,hubbard-charges-completeness,xyz-mpo-charges}, one may compute $[H_{\mathrm{hom,per}},\catalancharge{k}]$ directly and exhibit its cancellation.
We only sketch how the cancellation reduces to the Catalan-triangle recursion~\eqref{eq:ffd-catalan-triangle-recursion}.
We call each path-product monomial appearing in the translation-invariant sum $F_{s,m}$ of Eq.~\eqref{eq:ffd-periodic-Fsm-definition} an $(s,m)$-operator, where the second index $m$ counts the unit increments.

Figure~\ref{fig:ffd-catalan-cancellation-k10} shows only cancellations between distinct source blocks; contributions that cancel within a single $F_{s,m}$ are omitted.
For each remaining target, the three sources $F_{s,m-1}$, $F_{s-1,m}$, and $F_{s,m+1}$ contribute with total coefficient $(-1)^n\bigl(\mathcal{C}_{n+m-2,n}-\mathcal{C}_{n+m-1,n}+\mathcal{C}_{n+m-1,n-1}\bigr)$, which vanishes by the Catalan-triangle recursion~\eqref{eq:ffd-catalan-triangle-recursion}.

\paragraph{XXX-chain realization of the Catalan-tree pattern.}
The same Catalan-tree pattern organizes the local conserved charges of the XXX chain, whose explicit forms were obtained independently by Anshelevich and by Grabowski and Mathieu~\cite{anshelevich-heisenberg-first-integrals,GM-higher-conserved-XXX}, and it persists for their $SU(N)$-invariant generalizations~\cite{GM-higher-conserved-XXZ,GM-catalan-tree}.
These charges also admit a compact matrix product operator representation~\cite{heisenberg-mpo-charges,xyz-mpo-charges}.

\begin{define}
	For the XXX chain of length $L$ in the same indexing convention, $s$ is the interval support and $s-m$ is the number of $\vec{\sigma}$ factors.
	For $2\le s\le L$ and $s-m\ge 2$, set
	\begin{align}
		F_{s,m}^{\mathrm{XXX}}
		 & \equiv
		(-1)^{\lfloor (s+m)/2\rfloor}
		\sum_{j=1}^{L}
		\sum_{j=i_1<i_2<\cdots<i_{s-m}=j+s-1}
		\nonumber\\
		 & \hspace{6em}
		\vec{\sigma}_{i_1}\cdot
		\Bigl(
		\vec{\sigma}_{i_2}\times
		\bigl(
			\vec{\sigma}_{i_3}\times
			\cdots\times
				(\vec{\sigma}_{i_{s-m-1}}\times\vec{\sigma}_{i_{s-m}})
			\cdots
			\bigr)
		\Bigr).
		\label{eq:xxx-common-building-block-definition}
	\end{align}
	Here $\vec{\sigma}_i=(\sigma_i^x,\sigma_i^y,\sigma_i^z)$ is the vector of Pauli matrices at site $i$, and all site indices are read modulo $L$.
	The sum over $j$ is the translation sum, as in Eq.~\eqref{eq:ffd-periodic-Fsm-definition}.
\end{define}

At fixed label $k$, the XXX-chain formula is obtained from Eq.~\eqref{eq:ffd-catalan-charge-definition} by the replacement $F_{s,m}\mapsto F_{s,m}^{\mathrm{XXX}}$~\cite{GM-higher-conserved-XXX,GM-higher-conserved-XXZ,GM-catalan-tree}.

\begin{figure}[t]
	\centering
	\newcommand{\ffdsh}[2]{({(#2+0.5)*0.425em},{(#1-1.5)*0.425em})}
	\newcommand{\ffdcirc}[2]{%
		\draw[
			fill=darkgray,
			fill opacity=0.7,
			draw=black,
			line width=0.6pt
		] \ffdsh{#1}{#2} circle [radius=0.17em];
	}
	\newcommand{\ffdtimes}[2]{%
		\draw[
			black,
			line width=1.1pt
		] \ffdsh{#1}{#2} ++(-0.12em,-0.12em) -- ++(0.24em,0.24em)
		\ffdsh{#1}{#2} ++(-0.12em,0.12em) -- ++(0.24em,-0.24em);
	}
	\newcommand{\ffdup}[2]{\node[scale=1.0] at \ffdsh{#1}{#2} {$\uparrow$};}
	\newcommand{\ffdright}[2]{\node[scale=1.0] at \ffdsh{#1}{#2} {$\rightarrow$};}
	\newcommand{\ffdleft}[2]{\node[scale=1.0] at \ffdsh{#1}{#2} {$\leftarrow$};}
	\begin{tikzpicture}[
			baseline=0,
			scale=3,
			line cap=round,
			line join=round,
			gridline/.style={gray, thin, dashed},
			axis/.style={line width=0.8pt},
			diagramlabel/.style={scale=1.0}
		]
		\begin{scope}
			\node[diagramlabel, anchor=west] at \ffdsh{21.10}{1.60} {(a) Structure of $\catalancharge{10}$};
			\fill[RoyalBlue, opacity=0.2] \ffdsh{18}{1.5} -- \ffdsh{19}{1.5} -- \ffdsh{12}{8.5} -- \ffdsh{11}{8.5} -- cycle;
			\fill[RoyalBlue, opacity=0.2] \ffdsh{14}{1.5} -- \ffdsh{15}{1.5} -- \ffdsh{10}{6.5} -- \ffdsh{9}{6.5} -- cycle;
			\fill[RoyalBlue, opacity=0.2] \ffdsh{10}{1.5} -- \ffdsh{11}{1.5} -- \ffdsh{8}{4.5} -- \ffdsh{7}{4.5} -- cycle;
			\fill[RoyalBlue, opacity=0.2] \ffdsh{6}{1.5} -- \ffdsh{7}{1.5} -- \ffdsh{6}{2.5} -- \ffdsh{5}{2.5} -- cycle;
			\node[diagramlabel, anchor=west] at \ffdsh{11.50}{8.80} {$n=0$};
			\node[diagramlabel, anchor=west] at \ffdsh{9.50}{6.80} {$n=1$};
			\node[diagramlabel, anchor=west] at \ffdsh{7.50}{4.80} {$n=2$};
			\node[diagramlabel, anchor=west] at \ffdsh{5.50}{2.80} {$n=3$};
			\foreach \m in {0,1,2,3,4} {
					\node[diagramlabel] at \ffdsh{4.20}{2*\m} {$\m$};
				}
			\foreach \s in {3,4,5,6,7,8,9,10} {
					\node[diagramlabel, anchor=east] at \ffdsh{2*\s}{-1.40} {$\s$};
				}
			\draw[axis] \ffdsh{5}{-1} -- \ffdsh{21}{-1};
			\draw[axis] \ffdsh{5}{-1} -- \ffdsh{5}{9};
			\node[diagramlabel, rotate=90, anchor=south] at \ffdsh{12.50}{-2.70} {$s$};
			\node[diagramlabel] at \ffdsh{5}{10.10} {$m$};

			\ffdright{20}{1}
			\ffdup{19}{2}
			\ffdright{18}{3}
			\ffdup{17}{4}
			\ffdright{16}{5}
			\ffdleft{16}{3}
			\ffdup{15}{6}
			\ffdright{14}{7}
			\ffdleft{14}{5}
			\ffdup{13}{8}
			\ffdleft{12}{7}
			\ffdup{15}{2}
			\ffdright{14}{3}
			\ffdup{13}{4}
			\ffdright{12}{5}
			\ffdleft{12}{3}
			\ffdup{11}{6}
			\ffdleft{10}{5}
			\ffdup{11}{2}
			\ffdright{10}{3}
			\ffdup{9}{4}
			\ffdleft{8}{3}
			\ffdup{7}{2}

			\ffdcirc{20}{0}
			\ffdcirc{18}{2}
			\ffdcirc{16}{4}
			\ffdcirc{14}{6}
			\ffdcirc{12}{8}
			\ffdcirc{14}{2}
			\ffdcirc{12}{4}
			\ffdcirc{10}{6}
			\ffdcirc{10}{2}
			\ffdcirc{8}{4}
			\ffdcirc{6}{2}

			\ffdtimes{20}{2}
			\ffdtimes{18}{4}
			\ffdtimes{16}{6}
			\ffdtimes{14}{8}
			\ffdtimes{16}{2}
			\ffdtimes{14}{4}
			\ffdtimes{12}{6}
			\ffdtimes{12}{2}
			\ffdtimes{10}{4}
			\ffdtimes{8}{2}
		\end{scope}

		\begin{scope}[xshift=1.95cm,yshift=0.95cm]
			\node[diagramlabel] at \ffdsh{7.20}{2} {(b) Local cancellation};
			\foreach \m/\lab in {0/{$m{-}1$},1/{$m$},2/{$m{+}1$}} {
			\draw[gridline] \ffdsh{1}{2*\m} -- \ffdsh{5}{2*\m};
			\node[diagramlabel] at \ffdsh{0.20}{2*\m} {\lab};
			}
			\foreach \s/\lab in {1/{$s{-}1$},2/{$s$}} {
			\draw[gridline] \ffdsh{2*\s}{-1} -- \ffdsh{2*\s}{5};
			\node[diagramlabel, anchor=east] at \ffdsh{2*\s}{-1.50} {\lab};
			}
			\draw[axis] \ffdsh{1}{-1} -- \ffdsh{5.50}{-1};
			\draw[axis] \ffdsh{1}{-1} -- \ffdsh{1}{5};

			\ffdcirc{2}{2}
			\ffdcirc{4}{0}
			\ffdcirc{4}{4}
			\ffdtimes{4}{2}

			\draw[->, line width=0.6pt] \ffdsh{2.55}{2} -- \ffdsh{3.55}{2};
			\draw[->, line width=0.6pt] \ffdsh{4}{0.55} -- \ffdsh{4}{1.55};
			\draw[->, line width=0.6pt] \ffdsh{4}{3.45} -- \ffdsh{4}{2.45};
		\end{scope}
	\end{tikzpicture}
	\caption{The commutator cancellation $\bigl[H_{\mathrm{hom,per}},\catalancharge{k}\bigr]=0$.
		Filled circles are source blocks $F_{s,m}$ from Eq.~\eqref{eq:ffd-periodic-Fsm-definition}, crosses are target $(s,m)$-operators, and arrows denote contributions.
		Panel (a) shows $k=10$, with bands indexed by $n$ as in Eq.~\eqref{eq:ffd-catalan-charge-definition} and the leading term $F_{10,0}$ unshaded; contributions that cancel within one source block are omitted.
		Panel (b) isolates the three source blocks whose coefficients cancel by the Catalan-triangle recursion~\eqref{eq:ffd-catalan-triangle-recursion}.
		For the target $(s,m)=(5,2)$ in panel (a), the cancellation is $\mathcal{C}_{2,2}-\mathcal{C}_{3,2}+\mathcal{C}_{3,1}=2-5+3=0$.}
	\label{fig:ffd-catalan-cancellation-k10}
\end{figure}

\begin{define}[Operators for local conserved charges of the spin-$1/2$ XXX chain]
	For $2\le k\le L$, define the support-$k$ XXX-chain operator by
	\begin{align}
		\catalanchargeXXX{k}
		 & \equiv
		F_{k,0}^{\mathrm{XXX}}
		+
		\sum_{\substack{0<n+m<\lfloor k/2\rfloor \\ n\ge 0,\ m\ge 1}}
		(-1)^n
		\mathcal{C}_{n+m-1,n}
		F_{k-2n-m,m}^{\mathrm{XXX}}.
		\label{eq:xxx-catalan-charge-definition}
	\end{align}
	For $k=2$, this normalization gives $\catalanchargeXXX{2}=F_{2,0}^{\mathrm{XXX}}\equiv H_{\mathrm{XXX}}=-\sum_j\vec{\sigma}_j\cdot\vec{\sigma}_{j+1}$.
\end{define}

\begin{thm}
	\label{thm:xxx-catalan-charge-conservation}
	For $2\le k\le L$, the operator $\catalanchargeXXX{k}$ is a local conserved charge of the periodic spin-$1/2$ XXX chain~\cite{anshelevich-heisenberg-first-integrals,GM-higher-conserved-XXX,GM-higher-conserved-XXZ,GM-catalan-tree}:
	\begin{align}
		\bigl[
			H_{\mathrm{XXX}},
			\catalanchargeXXX{k}
			\bigr]
		 & =
		0.
		\label{eq:xxx-catalan-charge-conservation}
	\end{align}
\end{thm}

In this comparison, $s$ is a literal interval support for $F_{s,m}^{\mathrm{XXX}}$, whereas for Fendley's building block $F_{s,m}$ it is only the first index, the physical interval support in the defining spin representation being $2s+m-3$.
The second index $m$ also has different meanings: it counts holes within the support window in XXX and one-step increments along the induced path in the Fendley model.
The XXX-chain hierarchy begins at physical support $k=2$, whereas the hierarchy in the Fendley model begins at $k=3$.
Whether the homogeneous periodic Fendley model satisfies the algebraic sufficient conditions of the Grabowski--Mathieu Catalan-tree framework~\cite{GM-catalan-tree}, or instead requires a broader algebraic framework encompassing both the Fendley and XXX constructions, remains open.

\section{Summary and outlook}
\label{sec:outlook}

\subsection{Overview}
\label{subsec:summary-overview}

Induced paths are the elementary constituents in our construction.
For hidden free-fermion models with even-hole-free and claw-free frustration graphs, Theorem~\ref{thm:path-product-expansion} expresses each free-fermion mode $\Psi_k$ as a linear combination of path products along induced paths rooted at the auxiliary vertex representing the edge operator $\chi$, with coefficients given by independence polynomials of residual graphs.
For arbitrary claw-free frustration graphs, Theorem~\ref{thm:local-charge-conservation-pbc} and Theorem~\ref{thm:generalized-conserved-charges} construct local and generalized conserved charges from the same induced paths, in ranges controlled by the even-bubble-wand obstruction.
Neither construction uses the transfer matrix or the nonlocal conserved charges of previous approaches~\cite{fendley-fermions-in-disguise,fermions-behind-the-disguise,unified-graph-th}.

\subsection{Free-fermion modes}
\label{subsec:summary-free-fermion}

The path-product formula for $\Psi_k$ is proved through the Krylov generating function $\Phi_G(u)$ associated with the added vertex of the extended frustration graph, which represents the edge operator $\chi$.
Theorem~\ref{thm:Krylov-path-expansion} gives a path-product expansion of this generating function, and the mode decomposition in Theorem~\ref{thm:Phi-Psi-decomposition} then extracts the free-fermion modes as residues at its poles.
The resulting path-product operators are shown to satisfy directly the two defining properties of a free-fermion mode: the eigenoperator equation $[H_G,\Psi_k]=2\epsilon_k\Psi_k$ and the canonical anticommutation relation $\qty{\Psi_k,\Psi_l}=\delta_{k+l,0}$.
Equation~\eqref{eq:Krylov-diff-eq} and the Jacobi identity yield an identity for $\qty{\Phi_G(u),\Phi_G(v)}$ (Theorem~\ref{thm:Phi-anti-commu}); residue extraction then gives the anticommutation relation.
The same identity also implies Corollary~\ref{cor:path-product-identity}, an explicit independence-polynomial identity over induced paths.
In Section~\ref{sec:HA}, acting with the modes $\Psi_k$ on the all-up reference state $\refst=\ket{0\cdots 0}$ produces one-particle eigenstates of the antisymmetric Hamiltonian $H_A$ in the defining representation, giving the proof of the companion-paper theorem announced in Ref.~\cite{sajat-solving-ffd-1}.
For dynamics, Corollary~\ref{cor:edge-autocorrelation} extends the infinite-temperature autocorrelation formula for the edge operator of the Fendley model~\cite{sajat-ffd-corr} to arbitrary even-hole-free and claw-free frustration graphs, a result that may be useful for dissipative extensions of the FFD framework~\cite{sajat-dissipative-ffd}.
For the Fendley model, Subsection~\ref{subsec:ffd-operator-weight} carries out a saddle-point analysis of the path-product expansion of a free-fermion mode: at fixed range, the dominant induced paths have a preferred path-vertex density $\rho_*(p)$ tied to the momentum of the mode, and the optimized operator weight neither grows nor decays exponentially with the distance from the boundary.

The transverse-field Ising example in Section~\ref{sec:ising-example} shows explicitly that the path-product formula reproduces the Jordan-Wigner strings: when the frustration graph is a path, the induced path from $\chi$ to each vertex $\jbm$ in the extended graph is unique, and its path product coincides with the Jordan-Wigner string ending at $\jbm$.
For a general ECF frustration graph, several induced paths can connect $\chi$ to the same endpoint, so several path products can enter one mode, a feature largely invisible in the earlier transfer-matrix construction~\cite{fendley-fermions-in-disguise,fermions-behind-the-disguise,unified-graph-th}.

\subsection{Local conserved charges}
\label{subsec:summary-local-charges}

The path-product expansion also gives explicit higher-order local conserved charges.
Equation~\eqref{eq:local-charge} gives the odd local conserved charges $\localcharge{G}{2k+1}$ as weighted sums of products of graph Clifford generators along odd induced paths on at most $2k+1$ vertices, with independence-polynomial coefficients.
Theorem~\ref{thm:local-charge-conservation-pbc} proves conservation for every claw-free graph up to the range at which the smallest even bubble wand becomes relevant: if the smallest even bubble wand has size $2K$, then $\localcharge{G}{2k+1}$ is conserved for $k<K$.
The local conserved charges are thus assembled from induced-path products of the graph Clifford generators, paralleling the nonlocal conserved charges generated by the transfer matrix, which are assembled from products over independent sets of the same generators as in Eq.~\eqref{eq:nonlocal-charge-independent-set}~\cite{fendley-fermions-in-disguise,fermions-behind-the-disguise}.
The range bound of Theorem~\ref{thm:local-charge-conservation-pbc} yields an extensive family of odd local conserved charges whenever the smallest even bubble wand grows with the system size, and in particular whenever no even bubble wand exists, via Corollary~\ref{cor:local-charge-ebwf}.
When even bubble wands do appear at fixed size, extending the path-product construction beyond the first obstruction is left for future work; natural directions include coefficients incorporating the generalized cycle symmetries of Ref.~\cite{unified-graph-th}, or a reduction by gauging out cycle symmetries as in Ref.~\cite{sajat-claws}.
The Fendley model is the canonical example: its periodic frustration graph is $C_M^2$, which is even-bubble-wand-free, so the periodic chain with arbitrary couplings already carries the full odd family.

The homogeneous periodic Fendley model has, in addition, a family of oriented even path operators.
Translation symmetry makes the singleton coefficients of Theorem~\ref{thm:even-path-commutator} vanish, so these operators are conserved as well.
Equation~\eqref{eq:ffd-uniform-general-charge-formula} packages the odd and even families into a single translation-invariant path-product formula.
Thus the homogeneous periodic Fendley model has local conserved charges of both path-size parities, whereas for arbitrary inhomogeneous couplings the construction proved here retains only the odd local conserved charges.
This is reminiscent of the open spin-$1/2$ XXX chain, where only the local conserved charges involving an even number of spins remain, giving half as many charges as in the periodic chain~\cite{GM-open-spin-chains}.
Lemma~\ref{lem:ffd-catalan-recombination} and Corollary~\ref{cor:ffd-catalan-charge-conservation} show that triangular recombinations of the homogeneous periodic charges in Eq.~\eqref{eq:ffd-uniform-general-charge-formula} produce local conserved charges with a Catalan-tree coefficient pattern analogous to the XXX and $SU(N)$-invariant charges~\cite{anshelevich-heisenberg-first-integrals,GM-higher-conserved-XXX,GM-catalan-tree}.

The path-packing formula unifies the two constructions: a single family $\packingcharge{G}{m}{c}$ built from compatible packings of induced paths contains both the nonlocal independent-set charges and the local path-product conserved charges as special cases.
Theorem~\ref{thm:generalized-conserved-charges} proves conservation of this family for arbitrary claw-free graphs in the same range controlled by the smallest even bubble wand, and hence gives an extensive family of generalized conserved charges when that obstruction is absent.

\subsection{Outlook}
\label{subsec:summary-outlook}

It remains to derive the Hamiltonian reconstruction formula directly from the path-product expansion:
\begin{align}
	\label{eq:hamiltonian-reconstruction-outlook}
	H_G
	=
	\sum_{k = 1}^{\indepnum{G}}\epsilon_k\,\qty[\Psi_{k}, \Psi_{-k}].
\end{align}
For the Fendley model~\cite{fendley-fermions-in-disguise} and for ECF free-fermion models~\cite{fermions-behind-the-disguise,unified-graph-th}, Eq.~\eqref{eq:hamiltonian-reconstruction-outlook} is proved using the transfer matrix.
Substituting the path-product expansion of each $\Psi_k$ into the right-hand side of Eq.~\eqref{eq:hamiltonian-reconstruction-outlook} produces a large sum of products of induced-path operators, whereas the left-hand side is $H_G=\sum_{\jbm}\hc_{\jbm}$, indexed by single vertices.
A direct path-product proof must therefore show that the single-vertex part survives and that all remaining products supported on longer induced paths cancel.
The first requirement should follow by combining Bencs-type independence-polynomial identities~\cite{bencs1} with Eq.~\eqref{eq:path-square-identity}.
The second requirement appears to involve higher path-product analogues of Bencs-type identities for products associated with longer induced paths.
Hidden free-fermion models can have exponentially degenerate energy levels that are not resolved by the physical fermion modes.
For the Fendley model, Ref.~\cite{eric-lorenzo-ffd-1} identifies ancilla fermion modes that account for the nontrivial part of this degeneracy.
It remains to determine how these modes arise from the frustration graph and whether analogous modes exist in other hidden free-fermion models.
The mode decomposition of $\Phi_G(u)$ expresses every element of the Krylov operator space generated by $\chi$ in terms of the physical fermion modes.
A related inverse problem is to express local operators outside this Krylov space in terms of the physical fermion modes and the ancilla fermion modes.
Solving this problem would extend the computation in Corollary~\ref{cor:edge-autocorrelation} from the infinite-temperature autocorrelation of the edge operator to correlation functions of arbitrary local operators.

The induced paths are determined solely by the frustration graph, whereas the constituents used to construct local conserved charges in standard interacting integrable lattice models are model-dependent.
The spin-$1/2$ XXX chain~\cite{Bethe-XXX} and its $SU(N)$-invariant generalizations use nested cross products of Pauli vectors~\cite{anshelevich-heisenberg-first-integrals,GM-higher-conserved-XXX,GM-higher-conserved-XXZ,GM-catalan-tree}, whereas the XYZ chain~\cite{Baxter1971prl,baxter-eight-vertex,Baxter1972-heisenberg,baxter-8-vertex-1,baxter-8-vertex-2,baxter-8-vertex-3,Baxter-Book} admits descriptions in terms of doubling-product operators and related MPO constructions~\cite{xyz-all-charges,heisenberg-mpo-charges,xyz-easy-way,xyz-mpo-charges,yashin-two-parameter-mpo-heisenberg}.
In the one-dimensional Hubbard model~\cite{lieb-wu-hubbard,olmedilla-wadati-hubbard-charges}, the local conserved charges are built from products of XX-chain local conserved densities~\cite{fukai-doctoral-thesis,hubbard-charges,hubbard-charges-completeness}.
It remains to determine whether the local conserved charges of these models admit a model-independent graph-theoretic or algebraic description.
The Catalan-tree pattern shared by the XXX chain and the homogeneous periodic Fendley model raises a related concrete question: given the anisotropic XXZ and elliptic XYZ extensions of the XXX chain, does the periodic Fendley model admit analogous extensions?

Whether the first local conserved charge beyond the Hamiltonian exists is a standard test of spin-chain integrability~\cite{integrability-test}.
For a nearest-neighbor chain, $Q_2$ is the two-site Hamiltonian, while $Q_3$ is its first higher conserved charge and is $3$-local.
Recent results have established rigorous variants of this criterion for several classes of translationally invariant spin chains~\cite{hokkyo-integrability-test,dichotomy,resh-cond-proof}.
Here the oriented even path operators $\mathcal{E}_{G,\omega}^{(2k)}$ are the counterparts of the conventional $Q_3,Q_5,\ldots$ hierarchy.
These operators are conserved in the homogeneous periodic Fendley model but not for generic inhomogeneous couplings, whereas the odd local conserved charges $\localcharge{G}{2k+1}$ persist.
Thus the natural $Q_3$ counterpart can fail to be conserved even though an extensive family of higher local conserved charges survives.
Extending $Q_3$-based integrability tests to inhomogeneous graph-Clifford Hamiltonians remains an open problem.

Beyond the claw-free family, the path-product construction leads to the complementary problem of identifying frustration graphs whose graph Clifford algebras admit no nontrivial local conserved charge beyond $H_G$.
Here locality is naturally measured by induced-path size.
Since Theorem~\ref{thm:local-charge-conservation-pbc} no longer guarantees conservation of the odd local charges in this setting, methods for proving the absence of local conserved charges in spin-chain and lattice models~\cite{xyz-not-int,chiba-mixed-field-ising,shiraishi-nnn-heisenberg,hokkyo-integrability-test,chiba-ising-higher-dim,yamaguchi-chiba-shiraishi-nn-absence,dichotomy,shiraishi-tasaki-higher-d-xy-xyz} could yield sufficient graph-theoretic conditions for their absence.

\section*{Statements and Declarations}

\paragraph{Funding}
K.F. was supported by MEXT KAKENHI Grant-in-Aid for Transformative Research Areas A ``Extreme Universe'' (KAKENHI Grant No.\ JP21H05191).
K.F. was also supported by JSPS KAKENHI Grant-in-Aid for Research Activity Start-up (Grant No.\ JP25K23354), JSPS KAKENHI Grant-in-Aid for Early-Career Scientists (Grant No.\ JP26K17045), JSPS KAKENHI Grant-in-Aid for JSPS Fellows (Grant No.\ JP26KJ0404), and the JSPS Bilateral Program with MTA (Project No.\ 120263802).
B.P. and I.V. were supported by the Hungarian National Research, Development and Innovation Office, NKFIH Grant No.\ K-145904.
B.P. was also supported by the NKFIH excellence grant TKP2021\_NKTA\_64.

\paragraph{Competing interests}
The authors have no relevant financial or non-financial interests to disclose.

\paragraph{Data availability}
No datasets were generated or analyzed during the current study.

\bibliographystyle{unsrtnat}
\bibliography{fukai-ref}

\end{document}